\documentclass{lmcs}
\pdfoutput=1

\usepackage{lastpage}
\lmcsdoi{15}{3}{23}
\lmcsheading{}{\pageref{LastPage}}{}{}%
{Sep.~05,~2018}{Aug.~28,~2019}{}

\usepackage{amsthm,amsmath,amssymb}
\usepackage{xspace}
\usepackage{ifthen}
\usepackage{stmaryrd} % provides \bigsqcap

\usepackage{misc}
\usepackage{inancdefs}

\theoremstyle{definition}
\newtheorem{claim}{Claim}
\newtheorem*{claim*}{Claim}

\let\OriginalQedSymbol\qedsymbol
\renewcommand{\qedsymbol}{\OriginalQedSymbol\setcounter{claim}{0}}
\let\NormalQedSymbol\qedsymbol
\newenvironment{clmproof}[1]{\renewcommand{\qedsymbol}{$\dashv$}\begin{proof}[Proof of claim.]\space#1}{\end{proof}\renewcommand{\qedsymbol}{\NormalQedSymbol}}

\title[Ontology-Mediated Queries With Closed Predicates]{The Data Complexity of Ontology-Mediated Queries with Closed Predicates}

\author[Lutz]{Carsten Lutz}
\address{Fachbereich Informatik, University of Bremen, Germany}
\email{clu@uni-bremen.de}

\author[Seylan]{\.Inan\c{c} Seylan}
\address{ZF Friedrichshafen AG, Germany}
\author[Wolter]{Frank Wolter}
\address{Department of Computer Science, University of Liverpool, United Kingdom}
\email{wolter@liverpool.ac.uk}

\newcommand{\sigmaabox}{\Sigma_\Asf}
\newcommand{\sigmaclosed}{\Sigma_\Csf}

\begin{document}

\begin{abstract}
  % \todo[inline]{What follows is abstract for IJCAI13}  
%   When answering queries in the presence of ontologies, adopting the
%   closed world assumption for some predicates easily results in
%   intractability.
%   % When using ontologies to access instance data, it can be useful to
%   % make a closed world assumption for some predicates and an open world
%   % assumption for others. We consider the data complexity of
%   % conjunctive query (CQ) answering on the 
%    with (open
%   and) closed predicates is tractable, it coincides with answering CQs
%   with all predicates assumed open.
%  % in every tractable case, CQ answering with closed
%  %  and open predicates coincides with CQ answering where all predicates
%  %  are open.
%   In this sense, CQ answering with closed predicates is inherently
%   intractable. Our analysis also yields a dichotomy between \aczero
%   and $\conp$ for CQ answering w.r.t.\ ontologies formulated in
%   \dllite and a dichotomy between \ptime and $\conp$ for
%   $\el$. Interestingly, the situation is less dramatic 
% %  for ontologies  formulated 
% in the more expressive description logic \ELI, where we
%   find ontologies for which CQ answering is in \PTime, but does not
%   coincide with CQ answering where all predicates are open.
  We study the data complexity of ontology-mediated queries in which
  selected predicates can be closed (OMQCs), carrying out a
  non-uniform analysis of OMQCs in which the ontology is formulated in
  one of the lightweight description logics DL-Lite and $\mathcal{EL}$
  or in the expressive description logic $\mathcal{ALCHI}$. We focus
  on separating tractable from non-tractable OMQCs. On the level of
  ontologies, we prove a dichotomy between FO-rewritable and
  \conp-complete for DL-Lite and between \PTime and \conp-complete for
  $\mathcal{EL}$. We also show that in both cases, the meta problem to
  decide tractability is in \PTime. On the level of OMQCs, we show
  that there is no dichotomy (unless {\sc NP} equals \PTime) if both
  concept and role names can be closed. For the case where only concept
  names can be closed, we tightly link the complexity of OMQC
  evaluation to the complexity of generalized surjective CSPs.  We also identify a
  useful syntactic class of OMQCs based on \dlliter that are
  guaranteed to be FO-rewritable.
\end{abstract}

\maketitle

\section{Introduction}

The aim of ontology-mediated querying (OMQ) is to facilitate querying
incomplete and heterogeneous data by adding an ontology that provides
domain knowledge~\cite{DBLP:journals/jods/PoggiLCGLR08,DBLP:conf/rweb/BienvenuO15,DBLP:conf/ijcai/XiaoCKLPRZ18}. 
To account for the incompleteness, OMQ typically
adopts the open world assumption (OWA). In some applications, though,
there are parts of the data for which the closed world assumption
(CWA) is more appropriate.  For example, in a data integration
application some data may have been extracted from the web and thus be
significantly incomplete, suggesting the OWA, while other data may
come from curated relational database systems that are known to be
complete, thus suggesting the CWA. As an extreme case, one may even
use an ontology on top of complete data and thus treat all predicates
in the data under the CWA whereas additional predicates that are
provided by the ontology for more convenient querying are treated
under the OWA \cite{seylan09effective}. It is argued in
\cite{DBLP:conf/lics/BenediktBCP16} that a similar situation emerges
when only a subset of the predicates from a complete database is
published for privacy reasons, with an ontology linking the `visible'
and `invisible' predicates. When admitting both types of predicates in
queries (e.g.\ to analyze which parts of the private data can be
recovered), the CWA is appropriate for the visible predicates while
OWA is required for the invisible ones. A concrete example of mixed
OWA and CWA is given in \cite{Lutz:2013:ODA:2540128.2540276}, namely
querying geo-databases such as OpenStreetMap in which the geo data is
complete, thus suggesting the CWA, while annotations are incomplete
and suggest the OWA.

In this article, we are interested in ontologies formulated in a
description logic (DL).  In the area of DLs, quite a number of
proposals have been brought forward on how to implement a partial CWA,
some of them fairly complex
\cite{calvanese07eql,DBLP:journals/tocl/DoniniNR02,DBLP:conf/owled/GrimmM05,DBLP:journals/jacm/MotikR10,DBLP:conf/semweb/SenguptaKH11}.
In OMQ, a particularly straightforward and natural aproach is to simply
distinguish between OWA predicates and CWA predicates, as suggested
also by the motivating examples given above. The interpretation of CWA
predicates is then fixed to what is explicitly stated in the data
while OWA predicates can be interpreted as any extension thereof
% \cite{Lutz:2013:ODA:2540128.2540276,Magdalena}.
\cite{Lutz:2013:ODA:2540128.2540276}. 

Making the CWA for some predicates, from now on referred to as
\emph{closing} the predicates, has a strong effect on the complexity of
query evaluation. We generally concentrate on data complexity
where only the data is considered an input while the actual query and
ontology are assumed to be fixed; see 
\cite{DBLP:conf/kr/NgoOS16} for an analysis of combined complexity in
the presence of closed predicates.
% In fact,
The (data) complexity of evaluating (rather restricted forms of)
conjunctive queries (CQs) becomes \conp-hard already when ontologies
are formulated in inexpressive DLs such as DL-Lite$_{\mn{core}}$
and~\EL~\cite{franconi11query} whereas CQ evaluation without closed
predicates is FO-rewritable and thus in AC$^0$ for the former and in
\PTime for the latter
\cite{CDLLR07,DBLP:journals/jair/ArtaleCKZ09,hustadt2005data}.  Here,
FO-rewritability is meant in the usual sense of ontology-mediated
querying \cite{CDLLR07,DBLP:conf/rweb/KontchakovZ14,DBLP:conf/ijcai/BienvenuLW13,DBLP:conf/ijcai/Bienvenu0LW16}, that is, we can find a first-order
(FO) query that is equivalent to the original OMQ evaluated w.r.t.\
the ontology. Since intractability comes so quickly, it is not very
informative to analyze complexity on the level of logics, as in the
complexity statements just made; instead, one would like to know
whether closing a \emph{concrete set of predicates} results in
intractability for the \emph{concrete ontology used in an application}
or for the \emph{concrete combination of ontology and query that is
  used}. If it does not, then one should indeed close the predicates
since this may result in additional (that is, more complete) answers
to queries and additionally enables the use of more expressive query
languages for the closed part of the vocabulary. Otherwise, one can
resort to full OWA as an approximation semantics for querying or live
with the fact that evaluating the concrete query at hand is costly.

Such a \emph{non-uniform analysis} has been carried out in two different ways in~\cite{DBLP:journals/lmcs/LutzW17,DBLP:conf/pods/HernichLPW17} 
and in~\cite{DBLP:journals/tods/BienvenuCLW14} for classical OMQ (that is,
without closed predicates) and expressive DLs such as \ALC which give
rise to {\sc coNP} data complexity even when all predicates are
open. The former references aim to classify the complexity of
ontologies, \emph{quantifying over the actual query}: evaluating
queries formulated in a query language \Qmc is in \PTime for an
ontology \Omc if every query from \Qmc can be evaluated in \PTime
w.r.t.\ \Omc and it is \conp-hard if there is at least one Boolean
query from \Qmc that is \conp-hard to evaluate w.r.t.~\Omc. In the
latter reference, an even more fine-grained approach is taken where
the query is not quantified away and thus the aim is to classify the
complexity of \emph{ontology-mediated queries (OMQs)}, that is,
triples $(\Omc,\Sigma_\Asf,q)$ where \Omc is an ontology,
$\Sigma_\Asf$ a data vocabulary (where $\cdot_\Asf$ stands for
`ABox'), and $q$ an actual query. In both cases, a close connection to
the complexity of constraint satisfaction problems (CSPs) with fixed
template is identified. Given a relational structure $\Imc$, called a \emph{template},
the problem to decide for another relational structure $\Jmc$ whether there is a 
homomorphism from $\Jmc$ to $\Imc$ is called the \emph{constraint satisfaction problem
defined by $\Imc$}, and denoted CSP$(\Imc)$. Investigating the computational
complexity of CSP($\Imc$) is an active field of research that brings together algebra, 
graph theory, and logic~\cite{DBLP:conf/stoc/FederV93,DBLP:journals/siamcomp/BulatovJK05,DBLP:conf/csl/Krokhin10,DBLP:journals/corr/Bulatov17a,DBLP:journals/corr/Zhuk17}.
The connection between the complexity of OMQs and CSPs has proved to be very fruitful as it enables the
transfer of deep results available for CSPs to OMQ. In fact, it has
been used to obtain complexity dichotomies and results on the rewritability of OMQs into more conventional
database languages~\cite{DBLP:journals/tods/BienvenuCLW14,DBLP:journals/lmcs/LutzW17,DBLP:conf/pods/HernichLPW17,DBLP:conf/icdt/FeierKL17}.

The aim of this acticle is to carry out both types of analyses, the
\emph{quantified query case} and the \emph{fixed query case}, for OMQs
with closed predicates and for DLs ranging from the simple Horn DLs
DL-Lite and \EL to the expressive DL \alchi.  As the actual queries,
we use CQs, unions thereof (UCQs), and several relevant restrictions
of CQs and UCQs such as unary tree-shaped CQs, both in the directed
and in the undirected sense. Recall that DL-Lite and \EL are
underpinning the profiles OWL~2 QL and OWL~2 EL of the prominent OWL~2
ontology language while \alchi is related to OWL~2 DL
\cite{CDLLR07,DBLP:journals/jair/ArtaleCKZ09,BaBrLu-IJCAI-05}.  As a
starting point and general backdrop of our investigations, we prove
that query evaluation is in {\sc coNP} when the ontology is formulated
in $\mathcal{ALCHI}$, the actual query is a UCQ, and predicates can be
closed. Note that this bound is not a consequence of results on
ontology-mediated querying in description logics with nominals
\cite{ortiz2008data} because nominals are part of the ontology and
thus their number is bounded by a constant while closing a predicate
corresponds to considering a disjunction of nominals whose number is
only bounded by the size of the data (that is, the input size).
% \footnote{It is interesting to contrast this with the recent
%   result from \cite{Benedikt} that when ontologies are guarded
%   existential rules, then closing predicates can result in
%   \PSpace-hardness~\cite{Benedikt}}.

In the quantified query case, we aim to classify all \emph{TBoxes with
  closed predicates}, that is, all pairs $(\Tmc,\Sigma_\Csf)$ where
$\Tmc$ is a TBox formulated in the DL under consideration,
representing the ontology, and $\Sigma_\Csf$ is the set of predicates
(concept and role names) that are closed; all other predicates are
interpreted under the OWA. For the \dlliter dialect of DL-Lite and for
\EL, we obtain characterizations that separate the tractable cases
from the intractable ones and map out the frontier of tractability in
a transparent way (and also cover the fragment \dllitecore of
\dlliter). They essentially state that evaluating tree-shaped CQs is
{\sc coNP}-hard w.r.t.\ $(\Tmc,\Sigma_\Csf)$ if \Tmc entails certain
concept inclusions that mix open and closed predicates in a
problematic way while otherwise UCQ evaluation w.r.t.\
$(\Tmc,\Sigma_\Csf)$ is tractable, that is, FO-rewritable and in
\ptime, respectively.  Notably, this yields a dichotomy between
\aczero and $\conp$ for \dlliter TBoxes with closed predicates and
between \ptime and $\conp$ for $\el$ TBoxes with closed predicates.
It is remarkable that such a dichotomy can be obtained by a rather
direct analysis, especially when contrasted with the case of
expressive DLs such as \ALC without closed predicates for which a
dichotomy between \ptime and $\conp$ is equivalent to the dichotomy
between \ptime and \np for CSPs, a long-standing open problem that was
known as the Feder-Vardi conjecture and has been settled only very
recently
\cite{DBLP:journals/corr/Bulatov17a,DBLP:journals/corr/Zhuk17}.  The
proofs are a bit simpler in the case of \dlliter while they involve
the careful use of a certain version of the Craig interpolation
property in the \EL case.  The characterizations also allow us to
prove that it can be decided in \PTime whether a given TBox with
closed predicates is tractable or {\sc coNP}-complete (assuming
$\ptime \neq \np$), which we from now on call the \emph{meta
  problem}. It turns out that the tractable cases are precisely those
in which closing the predicates in $\Sigma_\Csf$ does not have an
effect on the answers to any query (unless the data is inconsistent
with the TBox).  This can be interpreted as showing that, in the
quantified query case, OMQ with closed predicates is inherently
intractable.\footnote{It is observed in
  \cite{Lutz:2013:ODA:2540128.2540276} that this is not the case for
  the extension $\mathcal{ELI}$ of \EL with inverse roles.}

Fortunately, this is not true in the fixed query case where we aim to
classify all \emph{ontology-mediated queries with closed predicates
  (OMQCs)} which take the form $(\Tmc,\Sigma_\Asf,\Sigma_\Csf,q)$
where $\Tmc$, $\Sigma_\Asf$, and $q$ are as in classical OMQs and
$\Sigma_\Csf \subseteq \Sigma_\Asf$ is a set of closed
predicates. Interestingly, switching to fixed queries results in CSPs
reentering the picture. While classifying the complexity of classical
OMQs based on expressive DLs corresponds to classifying standard CSPs,
we show that classifying
OMQCs %(both with inexpressive ontologies and with expressive ones)
is tightly linked to the classification of \emph{generalized
  surjective CSPs}. % ,
% both for inexpressive DLs such as DL-Lite and \EL and for expressive
% DLs such as \alchi
Surjective CSPs are defined exactly like standard CSPs except that
homomorphisms into the template are required to be surjective. What
might sound like a minor change actually makes complexity analyses
dramatically more difficult. In fact, there are concrete surjective
CSPs defined by a template with only six elements whose complexity is
not understood~\cite{bodirsky2012complexity} while there are no such
open cases for standard CSPs. The complexity of surjective CSPs is
subject to significant research activities
\cite{bodirsky2012complexity,chen2014algebraic} and it appears to be a
widely open question whether a dichotomy between {\sc PTime} and {\sc
  NP} holds for the complexity of surjective CSPs.  A
\emph{generalized} surjective CSP is defined by a finite set $\Gamma$
of templates rather than by a single template and the problem is to
decide whether there is a surjective homomorphism from the input
structure to some interpretation in~$\Gamma$. In the non-surjective
case, every generalized CSP can be translated into an equivalent
non-generalized CSP~\cite{DBLP:journals/ejc/FoniokNT08}. In the
surjective case, such a translation is not known. In this part, we
consider OMQCs where the ontology is formulated in any DL between
\dllitecore and \alchi or between \EL and \alchi, where only concept
names (unary predicates) can be closed, and where the actual queries
are Boolean UCQs in which all CQs are tree-shaped (BtUCQs). Our
result then is that there is a dichotomy between \PTime and \conp
for such OMQs if and only if there is a dichotomy between \PTime
and {\sc NP} for generalized surjective CSPs, a question that is wide
open.
%
% For these cases, we thus tightly link the complexity classification
% of OMQCs to an active research area and show that a full such
% classification is very challenging.
We find it remarkable that, consequently, there is no difference
between classifying OMQCs based on extremely simple DLs such as
\dllitecore and rather expressive ones such as
\alchi. % (see \cite{Magdalena}
% for similar results for combined complexity)
%
For the case where also role names (binary predicates) can be closed,
we show that for every \NP Turing machine $M$, there is an OMQC that
is polynomially equivalent to the complement of $M$'s word problem and
where the ontology can be formulated in \dllite or in \EL (and
queries are BtUCQs). By Ladner's theorem, this precludes the existence
of a dichotomy between \PTime and \conp (unless $\PTime=\NP$) and a
full complexity classification does thus not appear feasible with
today's knowledge in complexity theory. We also show that the meta
problem is undecidable.

% We start in Sections~\ref{sect:prelims} and~\ref{sect:basicres} with
% formally introducing our framework and establishing some preliminary
% results.  In Section~\ref{sect:fullq}, we identify a large and
% practically useful class of OMQCs that are tractable and even
% FO-rewritable; ontologies in these OMQCs are formulated in
% \dlliter, both concept and role names can be closed, and queries
% are quantifier-free
% UCQs. % , and a mild restriction is adopted on how the
% % ontology and closed predicates can interact Note that
% % quantifier-free queries play an important role in practice, for
% % example in the context of SPARQL \cite{sparql}.
% In Section~\ref{sect:closingconcepts}, we establish the connection to
% surjective CSPs for the case where only concept names can be closed
% (and where quantifiers in the query are allowed) and in
% Section~\ref{TMequi} we establish the connection to Turing machines
% when also role names can be closed.

Our results show that there are many natural tractable OMQs 
without closed predicates that become intractable when predicates are
closed. As a final contribution, we identify a family of OMQC where
tractability, and in fact FO-rewritability, is always guaranteed.  We
obtain this class by using \dlliter as the ontology language, unions
of quantifier-free CQs as the query language, and imposing the
additional restriction that the ontology contains no role inclusion
which states that an open role is contained in a closed one. We
believe that this class of OMQCs is relevant for practical
applications. We also prove that the restriction on RIs is needed for
tractability by showing that dropping it gives rise to OMQCs
that are \conp-hard.

\smallskip

This article is structured as follows. In Section~\ref{sect:prelim}, we
introduce description logics, relevant query languages, and
ontology-mediated querying with and without closed predicates.
We also observe that one can assume w.l.o.g. that all
predicates that occur in the data are closed and that UCQs using
open predicates can be combined with FO queries using closed predicates
without an impact on the complexity of query evaluation. In
Section~\ref{sect:basicres}, we prove that UCQ evaluation mediated by
$\mathcal{ALCHI}$ TBoxes with closed predicates is always in \conp.
% It has recently been shown that this is (most likely) not the case in
% a related framework where the ontology is formulated using guarded
% tuple generating dependencies (tgds) as there, OMQC evaluation is
% \PSpace-hard~\cite{Benedikt}.
%
In Section~\ref{sec:dichotomytbox}, we establish the characterizations
for the quantified query case and prove the announced complexity
dichotomies. In Section~\ref{sec:tboxdec}, we show that it is decidable in \ptime
whether a given TBox with closed predicates is tractable. 
We then switch to the case of fixed queries. In
Section~\ref{sect:closingconcepts}, we establish the link between
OMQCs with closed concept names to surjective CSPs and in
Section~\ref{TMequi} we link the general case where also role names
can be closed to the complexity of \np Turing machines and prove that
the meta problem is undecidable. In Section~\ref{sect:fullq}, we show
that evaluating UCQs without quantified variables is FO-rewritable for
\dlliter TBoxes in which no open role is included in a closed
role.

\section{Related Work}

The present article combines and extends the conference publications
\cite{DBLP:conf/ijcai/LutzSW15b,Lutz:2013:ODA:2540128.2540276}.
Classifications of the complexity of OMQs without closed predicates
based on expressive DLs have been studied in
\cite{DBLP:journals/lmcs/LutzW17,DBLP:conf/pods/HernichLPW17} in the
quantified query case and in~\cite{DBLP:journals/tods/BienvenuCLW14}
in the fixed query case. The combined complexity of ontology-mediated
querying with closed predicates has been investigated
in~\cite{DBLP:conf/kr/NgoOS16}. Among other things, it is shown there
that the combined complexity of evaluating OMQCs is 2\ExpTime-complete
when ontologies are formulated in \dlliter or in \EL and the actual
queries are UCQs. The rewritability of OMQCs into disjunctive datalog
with negation as failure is considered
in~\cite{DBLP:conf/ijcai/AhmetajOS16} and it is shown that a
polynomial rewriting is always possible when the ontology is
formulated in $\mathcal{ALCHIO}$ and the actual query is of the form
$A(x)$, $A$ a concept name.
 
The subject of \cite{DBLP:conf/lics/BenediktBCP16} is database
querying when only a subset of the relations in the schema is visible
and the data is subject to constraints, which in its `instance-level
version' is essentially identical to evaluating OMQCs. Among other
results, it is proved (stated in our terminology) that when the
ontology is formulated in the guarded negation fragment of first-order
logic (GNFO) and the actual query is a UCQ, then the combined
complexity of evaluating OMQCs is 2\ExpTime-complete. The lower bound
already applies when the ontology is a set of inclusion dependencies
or a set of linear existential rules (which subsume inclusion
dependencies). Moreover, there are OMQCs based on inclusion
dependencies and UCQs that are \ExpTime-hard in data complexity. These
results are completemented by the observation from \cite{Benedikt}
that there are \PSpace-hard OMQCs where the ontology is a set of linear existential rules and 
the actual query Boolean and atomic. It is interesting to contrast the latter two 
results with our \conp upper bound for $\mathcal{ALCHI}$ and UCQs.

Another related area is the study of combinations of the
open and closed world assumption in data exchange~\cite{DBLP:journals/jcss/LibkinS11}.
In data exchange one usually assumes an open-world semantics according to
which it is possible to extend instances of target schemas in an arbitrary way~\cite{DBLP:books/cu/ArenasBLM2014}. 
In an alternative closed-world semantics approach one only allows to add as much data as needed to the 
target to satisfy constraints of the schema mapping~\cite{DBLP:journals/tods/HernichLS11}. 
In~\cite{DBLP:journals/jcss/LibkinS11}, a mixed
approach is proposed: one can designate different attributes of target schemas as open or closed.
Although similar in spirit to ontology-based data access with closed predicates, 
the techniques required to analyze the mixed approach to data exchange appear
to be very different from those developed in this paper.

More vaguely related to our setup are so-called `nominal schemas' and
`closed variables' in ontologies that are sets of existential rules,
see \cite{DBLP:conf/www/KrotzschMKH11,DBLP:conf/kr/KrotzschR14} and
\cite{DBLP:conf/ijcai/AmendolaLMV18}, respectively. In both cases, the
idea is that certain object identifiers (nominals or variables) can
only be bound to individuals from the ABox, but not to elements of a
model that are introduced by existential quantifiers. When disjunction
is not present in the ontology language under consideration, which is
the main focus of the present article, then the expressive power of
these formalisms is orthogonal to ours. In the presence of
disjunction, nominal schemes and closed variables can simulate closed
predicates.

\section{Preliminaries}\label{sect:prelim}

% \todoin[color=yellow]{GENERAL RULES:

% \begin{itemize}
%   \item Use macros for complexity classes.
%   \item Use macros for logic names. If we do not use macros (defined in math mode) for these, then they get italized in theorem environments.
% This shouldn't happen.
%   \item Use macros for query language names, e.g., CQ.
% \end{itemize}
% }

% \todo[inline]{
% * improve label names from the symbols list on the left panel of Geany

% * standardize dash between logic name and TBox, e.g., ALCHI-TBox. We do not want a dash.

% * get rid of the $\Sigma$-concept notation and replace it with sig.

% * use the defined claim environment for claims. If you do not want numbering in claim, use the defined claim* environment. For proofs of claims, use the defined clmproof environment.
% }

We introduce description logics, relevant query languages, and
ontology-mediated querying with and without closed predicates. We also
observe that one can combine UCQs on open and closed predicates with
full first-order queries on closed predicates without adverse effects on
the decidability or complexity of query evaluation.

\subsection{Description Logics}

For a fully detailed introduction to DLs, we refer the reader to
\cite{Baader-et-al-03b,DBLP:books/daglib/0041477}.  Let $\NC$, $\NR$,
and $\NI$ be countably infinite sets of \emph{concept names},
\emph{role names}, and \emph{individual names}. An \emph{inverse role}
has the form $r^-$ with $r$ a role name. A \emph{role} is a role name
or an inverse role. We set $(r^-)^-=r$, for any role name $r$.  We use
three \emph{concept languages} in this article. \emph{$\mathcal{ALCI}$
  concepts} are defined by the rule
$$
C,D:= A \mid \top \mid \neg C \mid C \sqcap D \mid\exists r.C \mid \exists r^{-}.C 
$$
where $A\in \NC$ and $r\in \NR$. The constructor $\exists r.C$ is called a \emph{qualified existential restriction}. 
We use standard abbreviations and write, for example, $C \sqcup D$ for 
$\neg (\neg C \sqcap \neg D)$ and $\forall r.C$ for $\neg \exists r.\neg C$. \emph{\dllitecore} (or \emph{basic}) \emph{concepts} are 
defined by the rule
$$
B := A \mid \existsr{r}{\top} \mid \existsr{r^-}{\top}
$$
where $A\in \NC$ and $r\in \NR$. We often use $\exists r$ as shorthand
for the concept $\exists r.\top$. 
\emph{$\mathcal{EL}$ concepts} $C$ are defined by the rule
$$
C:= A \mid \top \mid \exists r.C 
$$
where $A\in \NC$ and $r\in \NR$. Thus, \dllitecore and \EL are both fragments of \alci. Note that \dllitecore admits inverse roles
but no qualified existential restrictions and \EL admits qualified existential restructions but no inverse roles. 
% \begin{table}
% \centering
% \begin{tabular}{| c | c | c | c | c | c | c | c | c |}
%   \hline
%                    & $A$       & $\top$    & $\neg C$  & $C\sqcap D$ & $\existsr{r}{C}$ & $\existsr{r^-}{C}$ & $\existsr{r}{}$ & $\existsr{r^-}{}$ \\ \hline 
%    \dllite concept & $\bullet$ &           &           &             &                  &                    & $\bullet$       & $\bullet$ \\ \hline
%    \el concept     & $\bullet$ & $\bullet$ &           & $\bullet$   & $\bullet$        &                    &                 &           \\ \hline
%    \alci concept   & $\bullet$ & $\bullet$ & $\bullet$ & $\bullet$   & $\bullet$        & $\bullet$          &                 &           \\ \hline
% \end{tabular}
% \caption{Various concept languages.}\label{tbl:concept_languages}
% \end{table}

In description logic, ontologies are constructed using concept
inclusions and potentially also role inclusions. An \emph{\ALCI
  concept inclusion (CI)} takes the form $C \sqsubseteq D$ with $C,D$
\ALCI concepts and \EL CIs are defined accordingly.  A
\emph{\dllitecore CI} takes the form $B_{1}\sqsubseteq B_{2}$ or
$B_{1}\sqsubseteq \neg B_{2}$ with $B_{1},B_{2}$ basic concepts. For
any of these three concept languages $\Lmc$, an \emph{$\Lmc$ TBox} is
a finite set of $\Lmc$ CIs. A \emph{role inclusion (RI)} takes the form
$r\sqsubseteq s$, where $r,s$ are roles. A \emph{\dlliter TBox} is a
finite set of \dllitecore CIs and RIs and an \emph{\alchi TBox} is a
finite set of \alci CIs and RIs.

% the `\alchi' column in this table means that an \alchi TBox consists of 
% inclusions of the form $r\sqsubseteq s$, where $r$, $s$ are roles and 
% $C\sqsubseteq D$, where $C$, $D$ are \alci concepts. All TBoxes 
% considered in this paper are finite.

% \begin{table}
% \centering
% \begin{tabular}{| c | c || c | c | c | c |}
%   \hline
%   Left-Hand Side  & Right-Hand Side                       & \dllitecore & \dlliter  & \el                    & \alchi    \\ \hline\hline
%   Role            & Role                                  &             & $\bullet$ &                        & $\bullet$ \\ \hline
%   \dllite concept & \dllite concept                       & $\bullet$   & $\bullet$ &                        &           \\ \hline
%   \dllite concept & $\neg C$, where $C$ a \dllite concept & $\bullet$   & $\bullet$ &                        &           \\ \hline
%   \el concept     & \el concept                           &             &           & $\bullet$              &           \\ \hline
%   \alci concept   & \alci concept                         &             &           &                        & $\bullet$ \\ \hline
% \end{tabular}
% \caption{Various TBox languages.}\label{tbl:tbox_languages}
% \end{table}

In description logic, data are stored in \emph{ABoxes} $\Amc$ which are finite sets of \emph{concept assertions} $A(a)$ and 
\emph{role assertions} $r(a,b)$ with $A \in \NC$, $r \in \NR$, and $a,b 
\in \NI$. For a role name $r$, we sometimes write $r^{-}(a,b)\in \Amc$ for $r(b,a)\in \Amc$.
We use $\mn{Ind}(\Amc)$ to denote the set of individual names used in the ABox \Amc. 

DLs are interpreted in standard first-order interpretations $\Imc$ presented as a pair $(\domain,\interf)$, where 
$\domain$ is a non-empty set called the \emph{domain} of $\inter$ and 
$\interf$ is a function that maps each concept name $A$ to a subset 
$\ext{A}$ of $\domain$ and each role name $r$ to a binary relation 
$\ext{r}$ on $\domain$. The extension of $\interf$ to roles and $\alci$ 
concepts is defined in Table~\ref{tbl:semantics}.
\begin{table}
  \[\begin{array}{rcl}
    \ext{(r^-)}            & = & \{(e,d)\mid (d,e)\in\ext{r}\} \\%\text{, for each role name }r,\\
    \ext{\top}             & = & \domain\\
    \ext{(\neg C)}         & = & \domain\setminus \ext{C}\\
    \ext{(C \sqcap D)}     & = & \ext{C}\cap \ext{D}\\
    \ext{(\existsr{r}{C})} & = & \{d\in\domain\mid \text{there exists }e\in\domain \text{ such that } (d,e)\in \ext{r} \text{ and } e \in \ext{C}\}\\    
  \end{array}\]
  \caption{Semantics of roles and $\alci$ concepts}\label{tbl:semantics}
\end{table}
An interpretation \Imc \emph{satisfies} a CI $C \sqsubseteq D$ if
$C^\Imc \subseteq D^\Imc$, a RI $r \sqsubseteq s$ if
$r^\Imc \subseteq s^\Imc$, a concept assertion $A(a)$ if
$a \in A^\Imc$ and a role assertion $r(a,b)$ if $(a,b) \in
r^\Imc$.
Note that this interpretation of ABox assertions adopts the standard
name assumption (SNA) which implies the unique name assumption. 
% We
%sometimes use $a^{\Imc}$ to denote the individual name $a$ in
%$\Delta^{\Imc}$.\footnote{\color{blue}This is a hack that we should
%  get rid of.}  % In order to avoid enforcing infinite models, we do not
% assume that interpretations contain all individual names in their
% domain.
%~ and use $\mn{Ind}(\Imc)$ to
%~ denote the individual names interpreted by \Imc.
%
An interpretation is a \emph{model of a TBox} \Tmc if it satisfies all 
inclusions in \Tmc and a \emph{model of an ABox} \Amc if it satisfies 
all assertions in~\Amc. A concept $C$ is \emph{satisfiable w.r.t.~a 
TBox $\Tmc$} if there exists a model $\Imc$ of $\Tmc$ with 
$C^{\Imc}\not=\emptyset$. As usual, we write $\tbox\models C 
\sqsubseteq D$ ($\Tmc \models r \sqsubseteq s$) if every model of \Tmc 
satisfies the CI $C\sqsubseteq D$ (resp.\ RI $r \sqsubseteq s$). 

A \emph{predicate} is a concept or role name. A \emph{signature}
$\Sigma$ is a finite set of predicates.  We use ${\sf sig}(C)$ to
denote the set of predicates that occur in the concept $C$ and
likewise for other syntactic objects such as TBoxes and ABoxes. An
ABox is a \emph{$\Sigma$-ABox} if it uses only predicates from
$\Sigma$.  We denote by ${\sf sub}(C)$ the set of subconcepts of the
concept $C$ and by ${\sf sub}(\Tmc)$ the set of subconcepts of
concepts that occur in the TBox $\Tmc$. The \emph{size} of any
syntactic object $O$, denoted $|O|$, is the number of
symbols needed to write it with concept, role, and individual names
viewed as a single symbol.

It will sometimes be convenient to regard interpretations as ABoxes
and vice versa. For an ABox $\abox$, the \emph{interpretation
  $\Imc_{\Amc}$ corresponding to $\abox$} is defined as follows:
\[
\begin{array}{rcl}
  \Delta^{\inter_\abox} & = & \adom{\abox}\\
  A^{\inter_\abox}      & = & \{a\mid A(a)\in\abox\}\text{, for all }A\in\conceptnames\\
  r^{\inter_\abox}      & = & \{(a,b)\mid r(a,b)\in\abox\}\text{, for all }r\in\rolenames.
\end{array}
\]
Conversely, every %at most countable 
interpretation $\Imc$ defines the (possibly infinite) ABox $\Amc_{\Imc}$ 
in which we regard the elements of the domain $\Delta^{\Imc}$ of $\Imc$ as individual names and let
$A(d)\in \Amc_{\Imc}$ if $d\in A^{\Imc}$ and $r(d,d')\in \Amc_{\Imc}$ if $(d,d')\in r^{\Imc}$.

%We use the following operations on interpretations.
A \emph{homomorphism} $h$ from an interpretation $\Imc_{1}$ to an
interpretation $\Imc_{2}$ is a mapping $h$ from $\Delta^{\Imc_{1}}$ to
$\Delta^{\Imc_{2}}$ such that $d\in A^{\Imc_{1}}$ implies
$h(d)\in A^{\Imc_{2}}$ for all $A\in \NC$ and
$d\in \Delta^{\Imc_{1}}$, and $(d,d')\in r^{\Imc_{1}}$ implies
$(h(d),h(d'))\in r^{\Imc_{2}}$ for all $r\in \NR$ and
$d,d'\in \Delta^{\Imc_{1}}$. We say that $h$ \emph{preserves} a set
$N \subseteq \NI$ of individual names if $h(a)=a$ for all $a\in N$.
The \emph{restriction} $\Imc|_{D}$ of an interpretation $\Imc$ to a
non-empty subset $D$ of $\Delta^{\Imc}$ is defined by setting
$\Delta^{\Imc|_{D}}=D$, $A^{\Imc|_{D}}=A^{\Imc}\cap D$, for all
$A\in \NC$, and $r^{\Imc|_{D}}=r^{\Imc}\cap (D\times D)$ for all
$r\in \NR$. The \emph{$\Sigma$-reduct} $\Jmc$ of an interpretation
$\Imc$ is obtained from $\Imc$ by setting $P^{\Jmc}=P^{\Imc}$ for all
predicates $P\in \Sigma$ and $P^{\Jmc}=\emptyset$ for all predicates
$P\not\in\Sigma$.

% Let $\Sigma$ be a signature.  Introduce copies $X^0$ and $X^1$ of
% every non-$\Sigma$-predicate $X$. We denote by $C^0$ and $C^1$ the
% resulting concept if each non-$\Sigma$ predicate $X$ in $C$ is
% replaced by $X^0$ and, respectively, $X^1$. Similarly, we denote by
% $\tbox^0$ and $\tbox^1$ the TBoxes obtained from $\tbox$ by
% replacing all concepts $C$ in $\tbox$ by $C^0$ and $C^1$,
% respectively.

\subsection{Query Languages}

The query languages used in this article are fragments of first-order
logic using predicates of arity one and two only.  Fix a countably
infinite set $\NV$ of \emph{variables}. A \emph{first-order query
  (FOQ)} $q(\vec{x})$ is %given by
a first-order formula
% $\varphi$ and a tuple $\vec{x}=x_{1},\ldots,x_{n}$ of individual
% variables such that the free variables of $\varphi$
whose free variables are contained in $\vec{x}$ and that is
constructed from atoms $A(x)$ and $r(x,y)$ using conjunction,
negation, disjunction, and existential quantification, where
$A\in \NC$ and $r\in \NR$. The variables in $\vec{x}$ are the
\emph{answer variables} of $q(\vec{x})$. The \emph{arity} of
$q(\vec{x})$ is defined as the length of $\vec{x}$ and a FOQ of arity
$0$ is called \emph{Boolean}. If the answer variables $\vec{x}$ of a
query $q(\vec{x})$ are not relevant, we simply write $q$ for
$q(\vec{x})$.  An \emph{assignment $\pi$ in an interpretation $\Imc$}
is a mapping from $\NV$ into $\Delta^{\Imc}$.  A tuple
$\vec{a}=a_{1},\ldots,a_{n}$ of individual names in $\Delta^{\Imc}$ is
an \emph{answer to $q(\vec{x})$ in $\Imc$} if there exists an
assignment $\pi$ in $\Imc$ such that $\Imc\models_{\pi}q$ (in
the standard first-order sense) and $\pi(x_{i})=a_{i}$ for
$1\leq i \leq n$.  In this case, we write $\Imc\models q(\vec{a})$.

A \emph{conjunctive query (CQ)} is a FOQ in prenex normal form that
uses no operators except conjunction and existential quantification. A
\emph{union of CQs (UCQ)} is a disjunction of CQs with the
  same answer variables. 
%The \emph{width} of a UCQ $q$ is the number
%of its variables.\footnote{\color{blue} changed from length to width
%  and removed notation $|q|$ as this denotes size; do we \emph{really}
%  need width? can't we just use size?} 
Every CQ $q$ can be viewed as an ABox
$\Amc_{q}$ by regarding the variables of $q$ as individual names.

A CQ $q(x)$ with one answer variable $x$ is a
\emph{directed tree CQ (dtCQ)} if it satisfies the following
conditions:
\begin{enumerate}

\item the directed graph $G_q=(V_q,E_q)$ is a tree with root $x$,
  where $V_q$ is the set of variables used in $q$ and $E_q$ contains an
  edge $(x_1,x_2)$ whenever there is an atom $r(x_1,x_2)$ in $q$;

\item if $r(x,y),s(x,y)$ are conjuncts of $q(x)$ then $r=s$.

\end{enumerate}
We sometimes regard a dtCQ $q$ as a $\mathcal{EL}$ concept $C_{q}$ in
the natural way such that for every interpretation $\Imc$ and
$a\in \Delta^{\Imc}$, $\Imc\models q(a)$ iff $a\in C_{q}^{\Imc}$.
Conversely, we denote by $q_{C}$ the natural dtCQ corresponding to the
$\mathcal{EL}$ concept $C$ such that $\Imc\models q_{C}(a)$ iff
$a\in C^{\Imc}$ holds for all interpretations $\Imc$ and
$a\in \Delta^{\Imc}$.  It will be convenient to not always strictly
distinguish between $C$ and $q_{C}$ and denote the query $q_{C}$ by
$C$.

A CQ $q(x)$ with one answer variable $x$ is a \emph{tree CQ (tCQ)} if it
satisfies the following conditions:
\begin{enumerate}

\item $G_q$ is a tree when viewed as an \emph{undirected} graph;

\item if $r(x,y),s(x,y)$ are conjuncts of $q(x)$ then $r=s$;

\item there are no conjuncts $r(x,y),s(y,x)$ in $q(x)$.

\end{enumerate}
Similarly to dtCQs, tCQs can be regarded as concepts in the extension
$\mathcal{ELI}$ of $\mathcal{EL}$ with inverse roles, see
\cite{DBLP:books/daglib/0041477}. We use the same notation as for
dtCQs.

\subsection{TBoxes and Ontology-Mediated Queries with Closed Predicates}

As explained in the introduction, our central objects of study are
TBoxes with closed predicates in the quantified query case and 
ontology-mediated queries with closed predicates in the fixed query
case. 

A \emph{TBox with closed predicates} is a pair $(\tbox,\Sigma_\Csf)$
with \Tmc a TBox and $\Sigma_\Csf$ a set of \emph{closed
  predicates}. An \emph{ontology-mediated query with closed predicates
  (OMQC)} takes the form $Q=(\Tmc,\Sigma_\Asf,\Sigma_\Csf,q)$ where
\Tmc is a TBox, $\Sigma_\Asf$ an \emph{ABox signature} which gives
the set of predicates that can be used in ABoxes,
$\Sigma_\Csf \subseteq \Sigma_\Asf$ a set of \emph{closed predicates},
and $q$ a query (such as a UCQ). The \emph{arity} of $Q$ is defined as
the arity of $q$.  If $\Sigma_{\Asf}= \NC\cup \NR$, then we omit
$\Sigma_{\Asf}$ and write $(\Tmc,\Sigma_\Csf,q)$ for
$(\Tmc,\Sigma_\Asf,\Sigma_\Csf,q)$.  Note that when
$Q=(\Tmc,\Sigma_\Asf,\Sigma_\Csf,q)$ is an OMQC, then
$(\Tmc,\Sigma_\Csf)$ is a TBox with closed predicates. When studying
TBoxes with closed predicates (in the quantified query case), we 
generally do not restrict the ABox signature.

The semantics of OMQCs is
as follows. We say that a model \Imc of an ABox \Amc \emph{respects
  closed predicates $\Sigma_\Csf$} if the extension of these
predicates agrees with what is explicitly stated in the ABox, that is,
$$
\begin{array}{rcl@{\quad}l}
  A^{\Imc} &=& \{a \mid A(a)\in \Amc\} & \text{ for all } A\in
                                         \Sigma_\Csf \cap \NC \text{ and}\\ [0.5mm]
  r^{\Imc} &=& \{(a,b)\mid r(a,b)\in \Amc\} & \text{ for all } r\in \Sigma_\Csf \cap \NR.
\end{array}
$$
Let $Q=(\Tmc,\Sigma_\Asf,\Sigma_\Csf,q)$ be an OMQC and \Amc a
$\Sigma_\Asf$-ABox. A tuple $\vec{a}$ of elements from $\adom{\abox}$,
denoted by $\vec{a}\in \mn{Ind}(\Amc)$ for convenience, is a
\emph{certain answer to $Q$ on} \Amc, written $\Amc \models
Q(\vec{a})$, if $\Imc\models q(\vec{a})$ for all models~$\mathcal{I}$
of $\Tmc$ and $\Amc$ that respect $\Sigma_\Csf$. The \emph{evaluation
  problem} for $Q$ is the problem to decide, given a
$\Sigma_\Asf$-ABox \Amc and a tuple $\vec{a} \in \mn{Ind}(\Amc)$,
whether $\Amc \models Q(\vec{a})$. Note that this problem parallels
the evaluation problem for CQs and other standard query language,
but with CQs replaced by OMQCs. 

An OMQC $Q=(\Tmc,\Sigma_\Asf,\Sigma_\Csf,q)$ with answer variables
$\vec{x}$ is \emph{FO-rewritable} if there is a first-order formula
$p(\vec{x})$, called an \emph{FO-rewriting of} $Q$, such that for all
$\Sigma_{\Asf}$-ABoxes~$\mathcal{A}$ and all
$\vec{a}\in\mn{Ind}(\Amc)$, we have $\Imc_{\Amc} \models p(\vec{a})$
iff $\Amc \models Q(\vec{a})$.  We remind the reader that the query
evaluation problem for $Q$ is in AC$^{0}$ when $Q$ is FO-rewritable.
\begin{exa}
Consider $\Tmc=\{A \sqsubseteq \exists r.B\}$ and $q(x)= \exists y\, r(y,x)$. Let
$Q_{0}=(\Tmc,\emptyset,q(x))$ be an OMQC without closed predicates and let $Q_{1}=(\Tmc,\Sigma_{\Csf},q(x))$
be the corresponding OMQC with closed predicates $\Sigma_{\Csf}=\{B\}$. Let $\Amc=\{A(a),B(b)\}$. Then $\Amc\not\models Q_{0}(b)$
since one can define a model $\Imc$ of $\Tmc$ and $\Amc$ in which
$(a,d)\in r^{\Imc}$ and $d\in B^{\Imc}$ for a fresh
element $d$.
However, $\Amc\models Q_{1}(b)$ since $B\in \Sigma_{\Csf}$. Note that $q(x)$ is an FO-rewriting of $Q_{0}$. The
FO-rewriting of $Q_{1}$ is more complicated and given by
$$
q(x) \vee (\exists y\,A(y) \wedge B(x) \wedge \forall y\,(B(y)\rightarrow y=x)) \vee (\exists y\, A(y)\wedge \neg \exists y\, B(y))
$$
The second disjunct captures answers for ABoxes in which one has to make $x$ an $r$-successor of some $y$
because only $x$ satisfies $B$ and the third disjunct captures answers for ABoxes in which there
is no common model of $\Tmc$ and the ABox that respects $\Sigma_{\Csf}$.
\end{exa}
%

%In our One of our aims is to classify the complexity of all OMQC that emerge
%from choosing a TBox language and a query language. 
% Sometimes, we want to restrict the predicates that are admitted in the
% set of closed predicates $\Sigma_\Csf$. To formalise this, we
% introduce the notion of
An \emph{OMQC language} is a triple
$(\Lmc,\Sigma,\Qmc)$ with \Lmc a TBox language (such as \dlliter,
\EL, or \alchi), $\Sigma$ a set of predicates (such as $\NC \cup \NR$,
$\NC$, or the empty set) from which the closed predicated in OMQCs
must be taken, and \Qmc a query language
(such as UCQ or CQ). Then $(\Lmc,\Sigma,\Qmc)$ comprises all OMQCs
$(\Tmc,\Sigma_\Asf,\Sigma_\Csf,q)$ such that $\Tmc \in \Lmc$,
$\Sigma_\Csf \subseteq \Sigma$, and $q \in \Qmc$.  Note that for
$\Sigma=\emptyset$ we obtain the standard languages of
ontology-mediated queries without closed
predicates~\cite{DBLP:journals/tods/BienvenuCLW14}. 

In the quantified query case, we aim to classify the complexity of all
TBoxes with closed predicates $(\tbox,\Sigma_{\Csf})$ where \tbox is
formulated in a DL of interest.  More precisely, for a query language
\Qmc we say that
  \begin{itemize}
  \item 
    \emph{\Qmc evaluation w.r.t.~$(\Tmc,\Sigma_\Csf)$ is in \ptime} if 
    for every $q\in\Qmc$, the evaluation problem for 
    $(\Tmc,\Sigma_\Csf,q)$ is in $\ptime$;
  \item 
    \emph{\Qmc evaluation w.r.t.~$(\tbox,\Sigma_{\Csf})$ is $\conp$-hard} 
    if there exists $q\in\Qmc$ such that the evaluation problem 
    for $(\Tmc,\Sigma_\Csf,q)$ is 
    $\conp$-hard;
  \item \emph{\Qmc evaluation w.r.t.~$(\tbox,\Sigma_\Csf)$ is FO-rewritable} 
    if for every $q\in \Qmc$, the OMQC $(\Tmc,\Sigma_\Csf,q)$ is FO-rewritable. 
 \end{itemize}
In the fixed query case, we aim to classify the complexity of all
OMQCs from some OMQC language, in the standard sense. We remind the
reader that without closed predicates the complexity of query
evaluation is well understood. In fact,
\begin{itemize}
\item every OMQC in $(\dlliter,\emptyset,\text{UCQ})$ is FO-rewritable \cite{CDLLR07}; 
\item the evaluation problem for every OMQC in $(\el,\emptyset,\text{UCQ})$ is in \PTime (and there are \PTime-hard OMQCs in
$(\el,\emptyset,\text{dtCQ})$) \cite{DBLP:conf/kr/CalvaneseGLLR06,DBLP:conf/lpar/KrisnadhiL07}; and
\item the evaluation problem for every OMQC in $(\alchi,\emptyset,\text{UCQ})$ is in \conp (and there are \conp-hard OMQCs in
$(\mathcal{ALCI},\emptyset,\text{dtCQ})$) \cite{hustadt2005data,ortiz2008data,Schaerf-93,DBLP:conf/kr/CalvaneseGLLR06}.
\end{itemize}
We will often have to deal with ABoxes that contradict the TBox given
that certain predicates are closed. We say that an ABox \Amc is
\emph{consistent w.r.t.~$(\Tmc,\Sigma_\Csf)$} if there is a model of
\Tmc and \Amc that respects $\Sigma_\Csf$.  We further say that
\emph{ABox consistency is FO-rewritable} for
$(\Tmc,\Sigma_{\Asf},\Sigma_\Csf)$ if there is a Boolean FOQ $q$ such
that for all $\Sigma_{\Asf}$-ABoxes~$\mathcal{A}$,
$\Imc_{\Amc} \models q$ iff \Amc is consistent w.r.t.\
$(\Tmc,\Sigma_\Csf)$. Note that if an ABox is consistent
w.r.t.~$(\Tmc,\Sigma_{\Csf})$, then it is consistent
w.r.t.~$(\Tmc,\emptyset)$. The converse does not hold. For example, if
$\Tmc=\{A \sqsubseteq B\}$ and $\Sigma_{\Csf}=\{B\}$, then
$\Amc=\{A(a)\}$ is not consistent w.r.t.~$(\Tmc,\Sigma_{\Csf})$ but
$\Amc$ is consistent w.r.t.~$(\Tmc,\emptyset)$.

Note that a CI $C \sqsubseteq D$ that uses \emph{only} closed
predicates acts as an integrity constraint in the standard database
sense \cite{abiteboul95foundations}. As an example, consider $\Tmc =
\{ A \sqsubseteq B\}$ and $\Sigma_\Csf = \{ A,B \}$. Then
$(\Tmc,\Sigma_\Csf)$ imposes the integrity constraint that if $A(a)$
is contained in an ABox, then so must be $B(a)$. In particular, an
ABox \Amc is consistent w.r.t.\ $(\Tmc,\Sigma_\Csf)$ iff \Amc
satisfies this integrity constraint. For ABoxes \Amc that are
consistent w.r.t.\ $(\Tmc,\Sigma_\Csf)$, $(\Tmc,\Sigma_\Csf)$ has no
further effect on query answers. In a DL context, integrity
constraints are discussed in
\cite{calvanese07eql,DBLP:journals/tocl/DoniniNR02,DBLP:conf/esws/MehdiRG11,DBLP:journals/ws/MotikHS09,DBLP:journals/jacm/MotikR10}.

\subsection{Basic Observations on OMQCs}

We first show that for DLs that support role inclusions, any OMQC is equivalent
to an OMQC in which the ABox signature and the set of closed predicates coincide. 
This setup was called \emph{DBoxes} in \cite{seylan09effective,franconi11query}. 
Assume OMQCs $Q_{1}$ and $Q_{2}$ have the same arity and ABox signature $\Sigma_\Asf$.
Then $Q_{1}$ and $Q_{2}$ are \emph{equivalent} if for all $\Sigma_\Asf$-ABoxes $\Amc$ and all tuples $\vec{a}$ in
${\sf Ind}(\Amc)$, $\Amc\models Q_{1}(\vec{a})$ iff $\Amc\models Q_{2}(\vec{a})$.
A class $\mathcal{Q}$ of queries is called \emph{canonical} if it is
closed under replacing a concept or role atom in a query with an atom
of the same kind.  All classes of queries considered in this article are
canonical. 
\begin{thm}
\label{thm:basic1}
  Let $\Lmc\in \{ \dlliter, \alchi\}$ and
  $\mathcal{Q}$ be a canonical class of UCQs.
%  % \in \{\text{UCQ}$, $\text{tUCQ}$, $\text{qfUCQ}$, $\text{CQ}$,
%  % $\text{tCQ}, \text{AQ}\}$,
%  or let $\Lmc\in \{ \dllitecore, \EL \}$ and
%  $\mathcal{Q}$ be a canonical class of UCQs closed under forming
%  disjunctions of queries
%  % a\in \{\text{UCQ}$, $\text{BtUCQ}$, $\text{qfUCQ}\}$.
  Then for every OMQC $Q=(\Tmc,\Sigma_\Asf,\Sigma_\Csf,q)$ from
  $(\Lmc,\NC \cup \NR,\Qmc)$, one can construct in polynomial time an
  equivalent OMQC $Q'=(\Tmc',\Sigma_\Asf,\Sigma_\Asf,q')$ with
  $\Tmc'\in \Lmc$ and $q'\in \mathcal{Q}$.
\end{thm}
\begin{proof}
Let $\Lmc\in \{ \dlliter, \alchi\}$
and let $Q=(\Tmc,\Sigma_\Asf,\sigmaclosed,q)$ be an OMQC with $\Tmc\in \Lmc$ and $q \in \Qmc$. % Clearly, we can assume that $\Sigma_\Csf
  % \subseteq \Sigma_\Asf$ (otherwise, restrict $\Sigma_\Csf$ to
  % $\Sigma_\Asf$ and extend the TBox with $A \sqsubseteq \bot$ for
  % all
  % concept names $A \in \Sigma_\Csf \setminus \Sigma_\Asf$ and
  % $\exists
  % r . \top \sqsubseteq \bot$ for all role names $r \in \Sigma_\Csf
  % \setminus \Sigma_\Asf$).
%
  For every predicate $P\in \sigmaabox\setminus \sigmaclosed$, we take a fresh
  predicate $P'$ of the same arity (if $P$ is a concept name, then $P'$
  is a concept name, and if $P$ is a role name, then $P'$ is a role
  name). Let $\Tmc'$ be the resulting TBox when all
  $P\in\sigmaabox\setminus \sigmaclosed$ are replaced by $P'$ and the
  inclusion $P \sqsubseteq P'$ is added, for each
  $P\in \sigmaabox\setminus \sigmaclosed$.  Denote by $q'$ the
  resulting query when every $P\in \Sigma_\Asf\setminus \sigmaclosed$
  in $q$ is replaced by $P'$. We show that
  $Q'=(\Tmc',\sigmaabox,\sigmaabox,q')$ is equivalent
  to $Q$.

  \smallskip

  First let \Amc be a $\sigmaabox$-ABox with $\Amc\not\models Q(\vec{a})$. Then there is
  a model $\Imc$ of \Tmc and $\Amc$ that respects closed predicates
  $\sigmaclosed$ such that $\Imc\not\models q(\vec{a})$.  Define
  an interpretation $\mathcal{I}'$ by setting
$$
\begin{array}{r@{\,}c@{\,}l}
  \Delta^{\Imc'} &=& \Delta^{\Imc} \\[1mm]
  A^{\mathcal{I}'} &=& \{a \mid A(a)\in \Amc\}, \text { for all } A\in \sigmaabox \setminus \sigmaclosed \\[1mm]
  r^{\mathcal{I}'} &=& \{(a,b) \mid r(a,b)\in \Amc\}, \text{ for all
  } r\in \sigmaabox \setminus \sigmaclosed \\[1mm]
  {A'}^{\mathcal{I}'} &=& A^{\mathcal{I}}, \text{ for all } A\in \sigmaabox\setminus \sigmaclosed\\[1mm]
  {r'}^{\mathcal{I}'} &=& r^{\mathcal{I}}, \text{ for all } r\in \sigmaabox\setminus \sigmaclosed
\end{array}
$$
and leaving the interpretation of the remaining predicates unchanged.  It
can be verified that $\Imc'$ is a model of $\Tmc'$ and \Amc that
respects closed predicates $\Sigma_\Asf$ such that
$\Imc'\not\models q'(\vec{a})$. Thus, $\Amc\not\models Q'(\vec{a})$.

\smallskip
\noindent
Conversely, let \Amc be a $\sigmaabox$-ABox such that
$\Amc\not\models Q'(\vec{a})$.  Let $\Imc'$ be
a model of $\Tmc'$ and $\Amc$ that respects closed predicates
$\Sigma_\Asf$ and such that $\Imc'\not\models q'(\vec{a})$.  Define an
interpretation $\mathcal{I}$ by setting
$$
\begin{array}{r@{\;}c@{\;}l}
  \Delta^{\Imc} &=& \Delta^{\Imc'} \\[1mm]
  A^{\mathcal{I}} &=& A'^{\mathcal{I}'}, \text{ for all }
  A\in \sigmaabox\setminus \sigmaclosed \\[1mm]
  r^{\mathcal{I}} &=& r'^{\mathcal{I}'}, \text{ for all } r\in \sigmaabox\setminus \sigmaclosed
\end{array}
$$
and leaving the interpretation of the remaining predicates unchanged. It 
is readily checked that $\Imc$ is a model of $\Tmc$ and $\Amc$ that 
respects closed predicates $\Sigma_\Csf$ and such that $\Imc\not\models 
q(\vec{a})$. Thus, $\Amc\not\models Q(\vec{a})$.

\end{proof}
As observed in \cite{reiter92what,calvanese07eql}, a partial CWA
enables the use of more expressive query languages without increasing
the complexity of query evaluation. This is particularly useful when
many predicates are closed---recall that it can even be useful to
close all predicates that can occur in the data. We next make this more
precise for our particular framework by introducing a concrete class of
OMQCs that combine FOQs for closed predicates with UCQs for open
predicates. As in the relational database setting, we admit only FOQs
that are \emph{domain-independent} and thus correspond to expressions
of relational algebra (and SQL queries), see
\cite{abiteboul95foundations} for a formal
definition. % Formally, a FOQ $q$
% is \emph{domain-independent} if for all interpretations $\inter$
% and $\interj$
% such that $\ext{P}=\extj{P}$
% for all $P\in
% {\sf sig}(q)$, we have that a tuple $\vec{d}$ is in
% $\Delta^{\Jmc}$
% and $\Jmc\models
% q(\vec{d})$ if $\vec{d}$ is in $\Delta^{\Imc}$ and $\Imc \models
% q(\vec{d})$, and vice versa. Intuitively, the truth value of a
% domain-independent FOQ depends only on the interpretation of the
% predicates, but not on the actual domain of the interpretation. For
% example, $\neg
% A(x)$, is not domain-independent whereas $B(x)\wedge \neg
% A(x)$ is domain-independent.

% Let $\Tmc$ be a TBox, $\Sigma_{\Asf}$ an ABox signature, and $\Sigma_{\Csf}$ a set of closed predicates. 
% We call a $\Sigma_{\Asf}$-ABox \Amc \emph{consistent w.r.t.~$(\Tmc,\Sigma_\Csf)$} if 
% there is a model of \Tmc and \Amc that respects $\Sigma_\Csf$. We say that \emph{ABox consistency is in \ptime} for $(\Tmc,\Sigma_{\Asf},\Sigma_\Csf)$
% if it is decidable in \ptime whether a $\Sigma_{\Asf}$-ABox is consistent w.r.t.~$(\Tmc,\Sigma_\Csf)$.
% We say that \emph{ABox consistency is FO-rewritable} for $(\Tmc,\Sigma_{\Asf},\Sigma_\Csf)$ if there is a 
% first-order sentence $\varphi$ such that for all $\Sigma_{\Asf}$-ABoxes~$\mathcal{A}$, we have
% $\Imc_{\Amc} \models \vp$ iff \Amc is consistent w.r.t.\ $(\Tmc,\Sigma_{\Asf},\Sigma_\Csf)$.
% or the appendix for details.
%
%   In the query language CQ$^{\text{FO}(\Sigma)}$, which we
% define next, (domain-independent) FOQs over closed predicates are
% plugged into CQs in place of atoms.

\begin{thm}
\label{thm:FOQs}
Let $Q=(\Tmc,\Sigma_{\Asf},\Sigma_{\Csf},q(\vec{x}))$ be an OMQC from
$(\alchi,\NC \cup \NR,\text{CQ})$ and $q'(\vec{x})$ a
domain-independent FOQ with $\mn{sig}(q') \subseteq \Sigma_\Csf$.
If $Q$ is FO-rewritable (evaluating $Q$ is in \ptime) and
ABox-consistency is FO-rewritable (in \ptime, respectively) for
$(\Tmc,\Sigma_{\Asf},\Sigma_{\Csf})$, then 
the OMQC $Q'=(\Tmc,\Sigma_{\Asf},\Sigma_{\Csf},q \wedge q')$ is FO-rewritable
(evaluating $Q'$ is in \ptime, respectively).
\end{thm}
\begin{proof}
  Assume that $p$ is an FO-rewriting of $Q$ and that $p'$ is a Boolean
  FOQ such that for all $\Sigma_{\Asf}$-ABoxes~$\mathcal{A}$,
  $\Imc_{\Amc} \models p'$ iff \Amc is consistent w.r.t.\
  $(\Tmc,\Sigma_\Csf)$. Then $\neg p' \vee (p \wedge q')$ is an
  FO-rewriting of $Q'$.  Next assume that evaluating $Q$ is in \PTime
  and that ABox consistency w.r.t~$(\Tmc,\Sigma_{\Asf},\Sigma_{\Csf})$
  is in \ptime.  To show that evaluating $Q'$ is in \ptime, let $\Amc$
  be a $\Sigma_{\Asf}$-ABox and $\vec{a}$ a tuple in $\Amc$. Then
  $\Amc\models Q'(\vec{a})$ iff $\Amc$ is not consistent
  w.r.t.~$(\Tmc,\Sigma_{\Csf})$ or $\Amc\models Q(\vec{a})$ and
  $\Imc_{\Amc}\models q'(\vec{a})$. As both can be checked in
  polynomial time, one can decide $\Amc\models Q'(\vec{a})$ in \ptime.
\end{proof}
%  
% Query languages between UCQs and FOQs have been studied before. In the
% standard setup where all predicates are open, it was shown in
% \cite{rosati07limits} that extending CQs with union and atomic
% negation results in \conp-hardness both in DL-Lite and in \EL,
% and that extending CQs with union and inequality results in
% undecidability in \EL. 
%In our query language CQ$^{\text{FO}(\Sigma)}$,
%we avoid these problems by allowing only CQs for the open predicates
%while restricting the expressive power of FOQs (which admits
% disjunction, full negation, and (in)equality) to closed predicates.
%
% While the proof of Theorem~\ref{thm:FOQs} is not intricate, we believe that the queries considered
% can be very useful for applications. 
%Note that the query language EQL-Lite(CQ) from
%\cite{calvanese07eql} can be viewed as a fragment of in which \emph{only} closed predicates are admitted.

\section{A \conp-Upper Bound for Query Evaluation}
\label{sect:basicres}

We show that for our most expressive DL, $\alchi$, UCQ evaluation for OMQCs is in \conp.
Recall from the introduction that this bound is not a consequence of results on ontology-mediated 
querying in description logics with nominals because nominals are part of the TBox and thus their number is a
constant. The proof uses a decomposition of countermodels (models that
demonstrate query non-entailment) into mosaics and then relies on a
guess-and-check algorithm for finding such decompositions.

\begin{thm}\label{thm:alchi_closed_datacomplexity}
  The evaluation problem for OMQCs in $(\alchi,\NC \cup 
  \NR,\text{UCQ})$ is in \conp.
\end{thm}
The proof is given by a sequence of lemmas. We first show that it
suffices to consider 
interpretations that are (essentially) forest-shaped 
when evaluating UCQs and then introduce mosaics as small forest-shaped interpretations. 
% Observe first that it suffices to
% consider OMQCs with Boolean UCQs. To give a polynomial reduction of evaluating OMQCs with arbitrary UCQs to 
% evaluating OMQCs with Boolean UCQs, assume an OMQC $Q=(\Tmc,\NC\cup \NR,\Sigma_{\Csf},q(\vec{x}))$ with $\vec{x}=x_{1},\ldots,x_{n}$
% is given. Assume w.l.o.g. that the $x_{i}$, $1\leq i \leq n$, are mutually distinct. Introduce fresh concept names 
% $A_{1},\ldots,A_{n}$, let $\Sigma_{\Csf}'= \Sigma_{\Csf}\cup \{A_{1},\ldots,A_{n}\}$,
% and obtain a Boolean UCQ $q'$ by replacing every CQ $q_{j}$ in $q$ by $\exists \vec{x}q_{j}'$, where $q_{j}'$ is obtained from
% $q_{j}$ by adding $A_{i}(x_{i})$ as a conjunct to $q_{j}$, for $1\leq i \leq n$. Then $\Amc\models Q(\vec{a})$ iff $\Amc'\models Q'$, where
% $\Amc'=\Amc \cup \{A_{i}(a_{i}) \mid 1\leq i \leq n\}$ and $Q'= (\Tmc,\NC\cup \NR,\Sigma_{\Csf}',q')$. In the remainder of this section we
% only consider OMQCs with Boolean UCQs.
A \emph{forest over an alphabet} $S$ is a prefix-closed set of words 
over $S^*\setminus\{\varepsilon\}$, where $\varepsilon$ denotes the 
empty word. Let $F$ be a forest over $S$. A \emph{root 
of} $F$ is a word in $F$ of length one.  A \emph{successor of} $w$ in 
$F$ is a $v\in F$ of the form $v=w\cdot x$, where $x\in S$. For a 
$k\in\natno$, $F$ is called \emph{$k$-ary}, if for all $w\in F$, we 
have that the number of successors of $w$ is at most $k$. The \emph{depth} of $w\in F$
is $|w|-1$, where $|w|$ is the length of $w$. The \emph{depth} of
a finite forest $F$ is the maximum of the depths of all $w\in F$. 
A \emph{tree} is a forest that 
has exactly one root. We do not mention the alphabet of a forest if it 
is not important.
\begin{defi}
  An interpretation $\inter=(\domain,\interf)$ is \emph{forest-shaped} 
  if $\domain$ is a forest and for all $(d,e)\in\domain\times\domain$ 
  and $r\in\rolenames$, if $(d,e)\in\ext{r}$, then
  \begin{itemize}
  \item $d$ or $e$ is a root of $\domain$, or
  \item $e$ is a successor of $d$ or $d$ is a successor of $e$.
  \end{itemize}
  $\Imc$ is of arity $k$ if the forest $\Delta^{\Imc}$ is of arity
  $k$.
  \hfill$\triangle$
\end{defi}

\smallskip
\noindent
Note that a forest-shaped interpretation is forest-shaped only in a
loose sense since it admits edges from any node to the root.  We
remind the reader of the following easily proved fact.
\begin{lem}
\label{homo}
Let $h$ be a homomorphism from $\Imc$ to $\Jmc$ preserving $\NI$ and let $q(\vec{x})$ be a UCQ 
and $\vec{a}$ a tuple of individual names. Then $\Jmc\models q(\vec{a})$ if $\Imc\models q(\vec{a})$. 
\end{lem}
As announced, the next lemma shows that it suffices to consider forest-shaped interpretations
when evaluating UCQs. We use 
${\sf cl}(\Tmc)$ to denote the closure of ${\sf sub}(\Tmc)$ under 
single negation.  
\begin{lem}\label{lem:alchi_forest_model}
  Let $\abox$ be a $\Sigma_{\Asf}$-ABox, $\vec{a}$ a tuple in ${\sf Ind}(\Amc)$, and
  $Q=(\Tmc,\Sigma_\Asf,\Sigma_\Csf,q)$ a OMQC from $(\alchi,\NC \cup
  \NR,\text{UCQ})$. Then the following are equivalent:
  \begin{enumerate}
  \item $\abox\models Q(\vec{a})$;
  \item $\inter\models q(\vec{a})$ for all forest-shaped models $\inter$ of $\tbox$ and $\abox$ that respect $\Sigma_\Csf$ and such that
     \begin{itemize}
            \item the arity of $\Delta^{\Imc}$ is $|\Tmc|$, 
            \item ${\sf Ind}(\Amc)$ is the set of roots of $\Delta^{\Imc}$,
            \item for every $d\in \Delta^{\Imc}\setminus {\sf Ind}(\Amc)$ and $\exists r.C\in {\sf cl}(\Tmc)$ with $d\in (\exists r.C)^{\Imc}$, there exists $a\in {\sf Ind}(\Amc)$ with $(d,a)\in r^{\Imc}$ and $a\in C^{\Imc}$ or 
there exists a successor $d'$ of $d$ in $\Delta^{\Imc}$
                  such that $(d,d')\in r^{\Imc}$ and $d'\in C^{\Imc}$.
     \end{itemize} 
  \end{enumerate}
\end{lem}
The proof is given in the appendix. (1) $\Rightarrow$ (2) is trivial and the proof of (2) $\Rightarrow$ (1) 
is by unravelling a model $\Imc$ of $\Tmc$ and $\Amc$ with $\Imc\not\models q(\vec{a})$ into a forest-shaped model 
of $\Tmc$ and $\Amc$ from which there is a homomorphism preserving $\NI$ to the original model $\Imc$ and
then applying Lemma~\ref{homo}.

Let $\tbox$ be an \alchi TBox. For an interpretation
$\inter$ and $d\in\domain$, let the \emph{$\tbox$-type}
of $d$ in \Imc be
 $$
 \mn{tp}_\inter(d)=\{C\in\cclos{\tbox}{}\mid d\in\ext{C}\}.
$$
In general, a \emph{$\tbox$-type} is a set
$t\subseteq \cclos{\tbox}{}$ such that for some model $\inter$ of
$\tbox$ and some $d\in\domain$, we have $t=\mn{tp}_\inter(d)$. We use
$\mathsf{TP}(\tbox)$ to denote the set of all $\Tmc$-types.  For
$\tbox$-types $t,t'$ and a role $r$, we write $t\rightsquigarrow_r t'$
if there is some model $\inter$ of $\tbox$ and $d,e\in\domain$ such
that $(d,e)\in\ext{r}$, $t=\mn{tp}_\inter(d)$, and
$t'=\mn{tp}_\inter(e)$.

We now define the notion of a \emph{mosaic} for an ABox $\Amc$ and an OMQC $Q=(\Tmc,\Sigma_\Asf,\Sigma_\Csf,q)$. 
Mosaics are abstract representations of interpretations which add to the ABox $\Amc$
a tree-shaped interpretation of outdegree bounded by $|\Tmc|$ and depth at most $|q|$. The tree-shaped
part is linked to the ABox via roles, where the number of ABox individuals linked to an element of
the tree-shaped interpretation is bounded by $|\Tmc|$. We ensure that a mosaic can be extended to a proper 
model of $\Tmc$ and $\Amc$ by hooking fresh interpretations to its ABox individuals and the leaves of its
tree-shaped interpretation. \emph{Coherent} sets of mosaics will correspond to forest-shaped
models of $\Tmc$ and $\Amc$. We ensure that is can be checked in polynomial time in $|\Amc|$ whether a set of 
mosaics is coherent and whether $q$ is satisfied in the interpretation to which is corresponds.
A standard guess and check algorithm (which guesses a set of mosaics and checks its coherence and satisfaction of $q$)
then shows that it is \np to decide $\Amc\not\models Q$.
%that these notions are also used in the investigation of the
%relationship between OMQCs and surjective CSPs in the main paper.

\begin{defi}
  Let $\abox$ be a $\Sigma_{\Asf}$-ABox and
  $Q=(\Tmc,\Sigma_\Asf,\Sigma_\Csf,q)$ from $(\alchi,\NC \cup
  \NR,\text{UCQ})$. A \emph{mosaic for} $Q$ and $\abox$ is a pair
  $(\inter,\tau)$, where $\inter$ is a forest-shaped interpretation
  and $\tau:\domain\rightarrow \mathsf{TP}(\tbox)$, satisfying the
  following properties:
  \begin{enumerate}
  \item $\domain\cap\indvnames=\adom{\abox}$;
  \item $\domain\setminus\adom{\abox}$ is a $\length{\tbox}$-ary tree
    of depth at most $\length{q}$;
  \item for all $d\in\domain\setminus\adom{\abox}$, the cardinality of
    $\{a\in\adom{\abox}\mid (d,a)\in\ext{r}\text{ for some role }r\}$
    is at most $\length{\tbox}$;
  \item for all $d\in\domain$ and
    $A\in\conceptnames\cap\cclos{\tbox}{}$, $d\in\ext{A}$ iff
    $A\in\tau(d)$;
  \item for all $(d,e)\in\domain\times\domain$ and roles $r$, if
    $(d,e)\in\ext{r}$ then $\tau(d)\rightsquigarrow_r \tau(e)$;
  \item for all $d\in\domain\setminus\adom{\abox}$ of depth at most
    $|q|-1$, if $\existsr{r}{C}\in\tau(d)$, then there is some
    $e\in\domain$ such that $(d,e)\in\ext{r}$ and $C\in\tau(e)$;
  \item $\inter\models\abox$
  \item for all $r\sqsubseteq s\in\tbox$, $\ext{r}\subseteq\ext{s}$;
  \item for all $A\in\Sigma_\Csf$ and all $A$ that do not occur in
    $\Tmc$, $\ext{A}=\{a\mid A(a)\in\abox\}$ and for all
    $r\in\Sigma_\Csf$ and all $r$ that do not occur in $\Tmc$,
    $\ext{r}=\{(a,b)\mid r(a,b)\in\abox\}$. \hfill$\triangle$
  \end{enumerate}
\end{defi}

\smallskip
\noindent
Let $(\Imc,\tau)$ and $(\Imc',\tau')$ be mosaics. A bijective function
$f:\Delta^{\Imc}\rightarrow \Delta^{\Imc'}$ is an \emph{isomorphism}
between $(\Imc,\tau)$ and $(\Imc',\tau')$ if both $f$ and its inverse
$f^{-1}$ are homomorphisms preserving $\NI$ and $\tau(d)=\tau'(f(d))$,
for all $d\in \Delta^{\Imc}$. We call $(\Imc,\tau)$ and
$(\Imc',\tau')$ \emph{isomorphic} if there is an isomorphism
between $(\Imc,\tau)$ and $(\Imc',\tau')$.
% $\Delta\subseteq\domain$, denoted by $\inter|_\Delta$, is defined as
% the following interpretation:
% \begin{itemize}
% \item $\Delta^{\inter|_\Delta}=\Delta$;
% \item $A^{\inter|_\Delta}=\ext{A}\cap\Delta$, for all
%   $A\in\conceptnames$;
% \item $r^{\inter|_\Delta}=\ext{r}\cap (\Delta\times\Delta)$, for all
%   $r\in\rolenames$.
% \end{itemize}
% We say that $\interj$ is a \emph{subinterpretation} of $\inter$ if
% $\domainj\subseteq\domain$ and $\interj=\inter|_{\domainj}$.
%
% These standard notions can be extended to mosaics in a
% straightforward fashion.
% For a mosaic
% $(\inter,\tau)$ and $\Delta\subseteq\domain$, we denote by
% $(\inter,\tau)|_\Delta$ the \emph{restriction} of $(\inter,\tau)$ to
% $\Delta$, i.e.,
% $(\inter,\tau)|_\Delta=(\inter|_\Delta,\tau|_\Delta)$. A mosaic
% $(\inter',\tau')$ is a \emph{submosaic} of $(\inter,\tau)$ if
% $\domainp{'}\subseteq\domain$ and
% $(\inter',\tau')=(\inter,\tau)|_{\domainp{'}}$.

For a forest $F$, $w\in F$, and $n\geq 0$, we denote by $F_{w,n}$
the set of all words $w'\in F$ such that $w'$ begins with $w$ and
$\length{w'}\leq\length{w}+n$.

\begin{defi}
  A set $M$ of mosaics for $(\Tmc,\Sigma_\Asf,\Sigma_\Csf,q)$ and
  $\abox$ is \emph{coherent} if the following conditions are
  satisfied:
  \begin{itemize}
  \item for all $(\inter,\tau),(\inter',\tau')\in M$,
$(\inter,\tau)|_{\adom{\abox}}=(\inter',\tau')|_{\adom{\abox}}$.
\item for all $(\inter,\tau)\in M$, $a\in\adom{\abox}$, and
  $\existsr{r}{C}\in\cclos{\tbox}{}$, if $\existsr{r}{C}\in\tau(a)$,
  then there exists $(\inter',\tau')\in M$ and $d\in\domainp{'}$ such
  that $(a,d)\in\extp{r}{'}$ and $C\in\tau'(d)$, where $d$ is either
  the root of $\Delta^{\Imc'}\setminus \mn{Ind}(\Amc)$ or $d\in
  \mn{Ind}(\Amc)$;
\item for all $(\inter,\tau)\in M$ and all successors $d\in\domain$ of the root of $\domain\setminus\adom{\abox}$, 
  there exist $(\inter',\tau')\in M$ and an isomorphism $f$ from $(\inter,\tau)|_{\domain_{d,\length{q}-1}\cup\adom{\abox}}$ 
  to $(\inter',\tau')|_{\domainp{'}_{e,\length{q}-1}\cup\adom{\abox}}$ such that $f(d)=e$, where
  $e\in\domainp{'}$ is the root of $\domainp{'}\setminus\adom{\abox}$.
\end{itemize}
We write $M\vdash q(\vec{a})$ if $\biguplus_{(\inter,\tau)\in M}\inter\models
q(\vec{a})$, where here and in what follows $\biguplus$ denotes a disjoint
union that only makes the elements that are not in $\adom{\abox}$
disjoint.
  \hfill$\triangle$
\end{defi}

\begin{lem}\label{lem:alchi_mosaic}
  Let $\abox$ be a $\Sigma_{\Asf}$-ABox, $\vec{a}$ a tuple in ${\sf Ind}(\Amc)$, and
  $Q=(\Tmc,\Sigma_\Asf,\Sigma_\Csf,q)$ a OMQC from $(\alchi,\NC \cup
  \NR,\text{UCQ})$. Then the following are equivalent:
  \begin{enumerate}
  \item $\abox\models Q(\vec{a})$;
  \item $M \vdash q(\vec{a})$, for all coherent sets $M$ of mosaics for $Q$ and
    $\abox$.
  \end{enumerate}
\end{lem}
\begin{proof}
  (2) $\Rightarrow$ (1). Suppose
  $\abox\not\models Q(\vec{a})$. Let $\Imc$ be a forest-shaped model with $\Imc\not\models q(\vec{a})$
  and satisfying the conditions of Lemma~\ref{lem:alchi_forest_model} (2).
  For each $d\in\domain\setminus\adom{\abox}$, let
  $\inter_d=\inter|_{\domain_{d,\length{q}}\cup\adom{\abox}}$ and
  $\tau_d=\bigcup_{e\in\domainp{_d}}\{e\mapsto
  \mn{tp}_{\inter}(e)\}$. Now set $M=\{(\inter_d,\tau_d)\mid
  d\in\domain\setminus\adom{\abox}\}$ if $\domain\neq\adom{\abox}$;
  and set $M=\{(\inter,\tau)\}$ with
  $\tau=\bigcup_{a\in\adom{\abox}}a\mapsto \mn{tp}_\inter(a)$ if
  $\domain=\adom{\abox}$. It is not hard to see that $M$ is a coherent
  set of mosaics for $Q$ and $\abox$
  (to satisfy Condition~9 for mosaics for concept names $A$ and role names $r$ that do not occur in $\Tmc$,
  we can clearly assume that $\ext{A}=\{a\mid A(a)\in\abox\}$ for all $A$
  that do not occur in $\Tmc$, and $\ext{r}=\{(a,b)\mid
  r(a,b)\in\abox\}$ for all $r$ that do not occur in $\Tmc$).  It
  remains to show that $M\not\vdash q(\vec{a})$. 
  But this follows from Lemma~\ref{homo} and the fact that the function
  $h$ from $\inter'=\biguplus_{(\interj,\tau)\in M}\interj$ to $\inter$ mapping every $a\in \NI$ to itself
  and every copy $d'\in \Delta^{\inter'}$ of some $d\in \Delta^{\Imc}$ to $d$ is a homomorphism from $\Imc'$ to $\Imc$
  preserving $\NI$.

  (1) $\Rightarrow$ (2). Suppose there is a coherent set $M$ of mosaics
  for $Q$ and $\abox$ with $M\not\vdash q$. We construct, by
  induction, a sequence of pairs
  $(\inter_0,\tau_{0}),(\inter_1,,\tau_{1}),\ldots$, where every
  $\Imc_{i}$ is a forest-shaped interpretation and
  $\tau_{i}:\Delta^{\Imc_{i}}\rightarrow \mathsf{TP}(\Tmc)$ such that
  every $d\in\domainp{_i}\setminus\adom{\abox}$ of depth
  $\leq i$ is associated with a mosaic $(\Imc_{d},\tau_{d})=
  (\Imc_{i},\tau_{i})|_{\Delta^{\Imc_{i}}_{d,|q|}\cup \mn{Ind}(\Amc)}$
  that is isomorphic to a mosaic in $M$.
  % , and where $(\Imc_{d},\tau_{d})=
  % (\Imc_{i},\tau_{i})|_{\Delta^{\Imc_{i}}_{d,|q|}\cup
  % \mn{Ind}(\Amc)}$d$ is the root of
  % $\domainp{_d}\setminus\adom{\abox}$,
  % $\inter_d=\inter_i|_{\domainp{_d}}$, and
  % $\tau_{d}=\tau_{i}|_{\domainp{_d}}$.

  For $i=0$, let $M_0$ be the set of all $(\interj,\tau)\in M$ such
  that there are $a\in\adom{\abox}$, $d\in\domainj$, and
  $\existsr{r}{C}\in\cclos{\tbox}{}$ with $\existsr{r}{C}\in\tau(a)$,
  $C\in\tau(d)$, $(a,d)\in\extj{r}$, and $d$ is either the root of
  $\domainj\setminus\adom{\abox}$ or $d\in \mn{Ind}(\Amc)$. Define
  $$
  \inter_0=\biguplus_{(\interj,\tau)\in M_0}\interj, \quad
  \tau_{0}=\biguplus_{(\interj,\tau)\in M_{0}}\tau
  $$
  It is easy to see that $(\inter_0,\tau_{0})$ satisfies the
  conditions above.

  For $i> 0$, let $d'\in\domainp{_i}\setminus\adom{\abox}$ be of depth
  $i$ and let $d$ be the unique element of 
  $\domainp{_i}\setminus\adom{\abox}$ of depth $i-1$ such that $d'$ is
  the successor of $d$. By the induction hypothesis and coherency of
  $M$, there is some $(\interj,\tau)\in M$ with $e\in\domainj$ the
  root of $\domainj\setminus\adom{\abox}$ such that
  $(\inter_d,\tau_d)|_{\domainp{_d}_{d',\length{q}-1}\cup\adom{\abox}}$
  is isomorphic to
  $(\interj,\tau)|_{\domainj_{e,\length{q}-1}\cup\adom{\abox}}$. W.l.o.g.\
  we assume that
  $\domainp{_d}_{d',\length{q}-1}=\domainj_{e,\length{q}-1}$; if this
  is not the case, we can always rename the elements in the latter
  without destroying the isomorphism. Set
  $(\inter_{d'},\tau_{d'})=(\interj,\tau)$ and assume that the points
  in $\Delta^{\Imc_{d'}}\setminus \domainp{_d}_{d',\length{q}-1}$ are
  fresh.  Set
  $$
  (\inter_{i+1},\tau_{i+1})=(\inter_{i}, \tau_{i})\cup
  \bigcup_{d'\in\domainp{_i}\setminus\adom{\abox}\text{ of depth
    }i}(\inter_{d'},\tau_{d'})
$$
Now define the interpretation $\inter$ as the limit of the 
sequence $\Imc_{0},\Imc_{1},\ldots$ (cf. proof of Lemma~\ref{lem:alchi_forest_model}).
It is shown in the appendix that $\inter$ is a model of $\tbox$ and $\abox$ that respects closed
predicates $\Sigma_\Csf$ such that $\inter\not\models q(\vec{a})$.
\end{proof}

\begin{lem}\label{lem:mosaic_set_bound}
  Let $\abox$ be a $\Sigma_{\Asf}$-ABox and
  $Q=(\Tmc,\Sigma_\Asf,\Sigma_\Csf,q)$ in $(\alchi,\NC \cup
  \NR,\text{UCQ})$. Then, up to isomorphisms, the size of any
  coherent set $M$ of mosaics for $Q$ and $\Amc$ is bounded by $(2|\Amc|)^{|\Tmc|^{f(|q|)}}$, for
  a linear polynomial $f$.
\end{lem}
\begin{proof}
 % \todo[inline]{This proof needs to have more details on why we have this bound.}
  The bound follows from Conditions~ 1, 2, 3, and 9 on mosaics and the
  first condition on coherent sets of mosaics. Note, in particular,
  that by the first condition on coherent sets $M$ of mosaics the
  restriction to $\mn{Ind}(\Amc)$ coincides for all mosaics in $M$ and
  that by Condition~3 on mosaics for any $d\in \Delta^{\Imc}\setminus
  \mn{Ind}(\Amc)$ the number of distinct $a\in \mn{Ind}(\Amc)$ with
  $(d,a)\in r^{\Imc}$ for some role $r$ is bounded by $|\Tmc|$ for any
  mosaic $(\Imc,\tau)$.
\end{proof}
We are now in the position to prove
Theorem~\ref{thm:alchi_closed_datacomplexity}.  Fix an OMQC
$Q=(\tbox,\Sigma_\Asf,\Sigma_\Csf,q)$ in
$(\alchi,\NC \cup \NR,\text{UCQ})$.  We show that given a
$\Sigma_{\Asf}$-ABox $\abox$ and tuple $\vec{a}$ in ${\sf Ind}(\Amc)$,
deciding $\abox\not\models Q(\vec{a})$ is in \np.  Assume $\Amc$ and
$\vec{a}$ are given. By Lemmas~\ref{lem:alchi_mosaic}
and~\ref{lem:mosaic_set_bound}, $\abox\not\models Q(\vec{a})$ iff
there exists a coherent set $M$ of mosaics for $Q$ and $\Amc$ such
that $|M| \leq (2|\Amc|)^{|\Tmc|^{f(|q|)}}$ ($f$ a linear polynomial)
and $M\not\vdash q(\vec{a})$. Thus, it is sufficient to show that it
can be decided in time polynomial in the size $|\Amc|$ of $\Amc$
whether $M$ is a coherent set of mosaics for $Q$ and $\Amc$ and
whether $M\not\vdash q(\vec{a})$. The first condition is clear. For
the second condition, observe that
$\Jmc = \biguplus_{(\Imc,\tau)\in M}\Imc$ can be constructed in
time polynomial in $|\Amc|$ and that checking if
$\Jmc\models q(\vec{a})$ is again possible in time polynomial in $|\Amc|$.

\section{Quantified Query Case: Dichotomies for \dlliter and \el}
\label{sec:dichotomytbox}

We consider the quantified query case and show two dichotomy results:
for every \dlliter TBox with closed predicates $(\Tmc,\Sigma_\Csf)$,
UCQ evaluation is FO-rewritable or \conp-complete. In the latter case,
there is even a tCQ $q$ such that evaluating the OMQ
$(\Tmc,\Sigma_\Csf,q)$ is \conp-hard. It thus follows that a TBox with
closed predicates is FO-rewritable for tCQs iff it is FO-rewritable
for CQs iff it is FO-rewritable for UCQs, and likewise for
\conp-completeness. It also follows that FO-rewritability coincides
with tractability, that is, query evaluation in \ptime.  We obtain the
same results for \EL TBoxes with closed predicates except that tCQs
are replaced with dtCQs and FO-rewritability is replaced with \ptime.
In both the \dlliter case and the \EL case, tractability also implies
that query evaluation with closed predicates coincides with query
evaluation without closed predicates, unless the data is inconsistent with the TBox. 
The proof strategy is similar in both cases, but the details are more involved for \EL. We first
consider the notion of convexity which formalizes the absence of
implicit disjunctions in answering tree-shaped queries and show that
for $\mathcal{ALCHI}$ TBoxes with closed predicates, non-convexity implies
\conp-hardness. We then introduce a syntactic condition for \dlliter
TBoxes (and later also for \EL TBoxes) with closed predicates called
safeness and show that non-safeness implies non-convexity while safeness
implies tractability.
%
% start with introducing the notion of convexity of a TBox
% with closed predicates that 
%
%
% To prove
% these dichotomies, we establish two characterisations of tractable query
% evaluation that are of independent interest: first, tractability is
% captured by rather natural syntactic constraints on the interaction
% between open and closed predicates in TBoxes. We call TBoxes with
% closed predicates satisfying these constraints \emph{safe}. Second,
% tractability is captured by the non-existence of proper disjunctive
% consequences of a TBox with closed predicates. TBoxes satisfying this
% condition are called \emph{convex}.
%

\subsection{Non-Convexity Implies \conp-hardness}

It is well-known that the notion of convexity is closely related to
the complexity of query evaluation, see for example
\cite{DBLP:conf/lpar/KrisnadhiL07,DBLP:journals/lmcs/LutzW17}.
Recall that we omit $\Sigma_{\Asf}$ from the OMQC $(\tbox,\Sigma_{\Asf},\Sigma_{\Csf},q)$ 
and write $(\tbox,\Sigma_{\Csf},q)$ if $\Sigma_{\Asf}=\NC\cup \NR$.
%For CQs $q_{1}(x),q_{2}(x)\in =\varphi_{1}(x)$ and $q_{2}(x)=\varphi_{2}(x)$ with one answer variable $x$,
%we denote by $q_{1}\vee q_{2}(x)$ the FOQ with answer variable $x$ defined by the disjunction $\varphi_{1}(x)\vee% \varphi_{2}(x)$.
% 
\begin{defi}\label{def:convex_tbox}
  Let $\mathcal{Q}\in \{\text{tCQ},\text{dtCQ}\}$. A TBox with closed predicates $(\Tmc,\Sigma_\Csf)$ is 
  \emph{convex} for $\Qmc$ if for all ABoxes $\abox$, $a\in\adom{\abox}$, and $q_1(x), q_2(x)\in \Qmc$ 
  the following holds: 
  if $\Amc\models (\tbox,\Sigma_\Csf,q_1\vee q_2)(a)$, then $\Amc\models (\tbox,\Sigma_\Csf,q_{i})(a)$ 
  for some $i \in \{1,2\}$.
  \hfill$\triangle$
\end{defi}

\smallskip
\noindent
Without closed predicates, every \dlliter and $\mathcal{EL}$ TBox is
convex for tCQs. In fact, it is shown in
\cite{DBLP:journals/lmcs/LutzW17,DBLP:conf/pods/HernichLPW17} that for
TBoxes in $\mathcal{ALCHI}$ (and even more expressive languages)
without closed predicates, convexity for tCQs is a necessary condition
for UCQ evaluation to be in \ptime (unless \ptime = \conp).  The
following is an example of a \dlliter TBox with closed predicates that
is not convex for tCQs.
\begin{exa}\label{ex1}
Let
$\Tmc = \{ A \sqsubseteq \exists r.\top, \exists
    r^{-}.\top\sqsubseteq B\}$ and $\Sigma_\Csf = \{ B \}$.
We show that $(\Tmc,\Sigma_\Csf)$ is not convex for tCQs. To this end, let
$$
\begin{array}{rcl}
    \Amc &=&
    \{A(a),B(b_{1}),A_{1}(b_{1}), B(b_{2}), A_{2}(b_{2})\} \\[1mm]
    q_i &=& \exists
    y \, r(x,y) \wedge A_i(y) %\wedge B(y) 
    \text{ for } i \in \{1,2\}.
\end{array}
$$
Then $\Amc\models(\Tmc,\Sigma_\Csf,q_1 \vee q_2)(a)$, whereas
$\Amc \not\models (\Tmc,\Sigma_\Csf, q_{i})(a)$ for any $i \in \{1,2\}$.
\end{exa}
We next show that non-convexity implies that query evaluation is
\conp-hard. The result is formulated for $\mathcal{ALCHI}$ as our
maximal description logics and comes in a directed version (used lated
for \EL) and a non-directed one (used for \dlliter).
%this condition is necessary for \conp-hardness as well.
%
\begin{lem}\label{lem:nonconvex_imp_conp}
  Let $(\Tmc,\Sigma_\Csf)$ be an $\mathcal{ALCHI}$ TBox with closed predicates that is not 
  convex for tCQs (resp.\ dtCQs). Then there exists a tCQ $q$ (resp.\ dtCQ $q$) such that the  
  evaluation problem for $(\tbox,\Sigma_\Csf,q)$ is \conp-hard.
\end{lem}
\begin{proof}
The proof is by a reduction of 2+2-SAT inspired by~\cite{Schaerf-93}. 2+2-SAT is a variant of 
propositional satisfiability where each clause contains precisely two 
positive literals and two negative literals. The 
queries $q_1$ and $q_2$ that witness non-convexity from 
Definition~\ref{def:convex_tbox} are used as subqueries of the query 
constructed in the reduction, where they serve the purpose of 
distinguishing truth values of propositional variables. We give a sketch of the reduction only 
as it is very similar to a corresponding reduction for TBoxes without closed predicates~\cite{DBLP:journals/lmcs/LutzW17}. 
Let $(\tbox,\Sigma_\Csf)$ be not convex for tCQs. Then there are an ABox 
$\abox$ with $a\in\adom{\abox}$ and tCQs $q_1(x), q_2(x)$ such that 
$\abox\models(\tbox,\Sigma_\Csf,q_1\vee q_2)(a)$ and 
$\abox\not\models(\tbox,\Sigma_\Csf,q_i)(a)$ for all $i\in\{1,2\}$.
  
We define 2+2-SAT. A \emph{2+2 clause} is of the form $(p_1
\vee p_2 \vee \neg n_1 \vee \neg n_2)$, where each of
$p_1,p_2,n_1,n_2$ is a propositional letter or a truth constant $0$, $1$. 
A \emph{2+2 formula} is a finite conjunction of 2+2 clauses. Now,
2+2-SAT is the problem of deciding whether a given 2+2 formula is
satisfiable. It is shown in \cite{Schaerf-93} that 2+2-SAT is
\NP-complete.

Let $\vp=c_0 \wedge \cdots \wedge c_{n}$ be a 2+2 formula in 
propositional letters $w_0,\dots,w_{m}$, and let $c_i=p_{i,1} \vee 
p_{i,2} \vee \neg n_{i,1} \vee \neg n_{i,2}$ for all $i \leq n$. Our 
aim is to define an ABox $\Amc_\vp$ and a tCQ $q_0$ such that $\vp$ 
is unsatisfiable iff $\Amc_\vp \models 
(\tbox,\Sigma_\Csf, q_0)(f)$, for an individual name $f$ we define shortly.
To start, we represent the formula $\vp$ in the ABox $\Amc_\vp$ as follows:
    
  \begin{itemize}
    \item the individual name $f$ represents the formula $\vp$;
    \item the individual names $c_0,\dots,c_n$ represent the clauses of
      $\vp$;
    \item the assertions $c(f,c_0), \dots, c(f,c_{n})$, associate $f$ with
      its clauses, where $c$ is a role name that does not occur in \Tmc;
    \item the individual names $w_0,\dots,w_m$ represent propositional 
      letters, and the individual names $0,1$ represent truth constants;
    \item the assertions
        $$
        \bigcup_{i \leq n} \{ p_1(c_i,p_{i,1}), p_2(c_i,p_{i,2}),
        n_1(c_i,n_{i,1}), n_2(c_i,n_{i,2}) \}
        $$
        associate each clause with the four variables/truth constants that
        occur in it, where $p_1$, $p_2$, $n_1$, and $n_2$ are role names that do not
        occur in \Tmc.
  \end{itemize}
  We further extend $\Amc_\vp$ to enforce a truth value for each of
  the variables $w_i$ and the truth-constants $0,1$. To this end, add
  to $\Amc_\vp$ copies $\Amc_0,\dots,\Amc_{m}$ of the ABox $\Amc$ 
  obtained by renaming individual names such that $\mn{Ind}(\Amc_i) 
  \cap \mn{Ind}(\Amc_j) = \emptyset$ whenever $i \neq j$. Moreover, 
  assume that $a_{i}$ coincides with the $i$th copy
  of $a$.  Intuitively, the copy $\Amc_i$ of \Amc is used to generate
  a truth value for the variable $w_i$, where we want to interpret
  $w_i$ as true in an interpretation \Imc if $\Imc \models q_1(a_i)$ 
  and as false if $\Imc \models q_2(a_i)$. To actually relate each 
  individual name $w_i$ to the associated ABox $\abox_i$, we use the 
  role name $r$ that does not occur in $\tbox$. More 
  specifically, we extend $\abox_\varphi$ as follows:
  \begin{enumerate}
    \item link variable $w_i$ to the ABox $\abox_i$ by adding the 
      assertion $r(w_i,a_i)$, for all $i\leq m$; thus, the truth of 
      $w_i$ means that ${\sf tt}(x):=\exists y\,(r(x,y) \wedge q_{1}(y))$ is satisfied and falsity 
      means that ${\sf ff}(x):=\exists y\,(r(x,y) \wedge q_2(y))$ is satisfied;
    
    \item to ensure that $0$ and $1$ have the expected truth values, 
      add a copy of $q_1$ viewed as an ABox $\Amc_{q_{1}}$ with root $1'$ and a copy 
      of $q_2$ viewed as an ABox $\Amc_{q_{2}}$ with root $0'$; add $r(0,0')$ and 
      $r(1,1')$.
  \end{enumerate}
  Let $\Bmc$ be the resulting ABox. Consider the tCQ
   
   \begin{eqnarray*}
    q_0(x) &  = & \exists y, y_{1},y_{2},y_{3},y_{4}\,\big(c(x,y) \wedge p_{1}(y,y_{1}) \wedge \mn{ff}(y_{1})\wedge
                      p_{1}(y,y_{2}) \wedge \mn{ff}(y_{2})\wedge\\
           &&         \;\;\;\;\;   n_{1}(y,y_{3}) \wedge \mn{tt}(y_{3})\wedge
                         n_{2}(y,y_{4}) \wedge \mn{tt}(y_{4})\big)
   \end{eqnarray*}
  which describes the existence of a clause with only false literals 
  and thus captures falsity of $\vp$. It is straightforward to show that $\vp$ is 
  unsatisfiable iff $\Bmc \models (\Tmc,\Sigma_\Csf,q_0)(f)$. Finally observe that $q_{0}$ is a dtCQ
  if $q_{1}$ and $q_{2}$ are dtCQs.
\end{proof}

\subsection{Dichotomy for $\dlliter$}
The next definition gives a syntactic safety condition for \dlliter
TBoxes with closed predicates that turns out to characterize tractability.
\begin{defi}[Safe \dlliter TBox]
\label{def:safe_dllite_tbox}
  Let $(\Tmc,\Sigma_\Csf)$ be a $\dlliter$ TBox with closed predicates.
  Then $(\Tmc,\Sigma_\Csf)$ 
  is \emph{safe} if there are no basic concepts $B_{1},B_{2}$ and
  role $r$ such that the following conditions are satisfied:
  \begin{enumerate}
  \item $B_1$ is satisfiable w.r.t. $\tbox$;
  \item $\tbox\models B_1\sqsubseteq \exists r$ and
    $\tbox\models\exists r^{-}\sqsubseteq B_2$;
  \item $B_1\neq \exists r'$, for every role $r'$ with
    $\Tmc\models r'\sqsubseteq r$;
  \item ${\sf sig}(B_{2})\subseteq \Sigma_\Csf$ and $\mn{sig}(r') \cap
    \Sigma_\Csf = \emptyset$ for every role $r'$ with $\Tmc \models B_1
    \sqsubseteq \exists r'$ and $\Tmc\models r'\sqsubseteq r$.
  \end{enumerate}
  \hfill$\triangle$
\end{defi}

  % Point~1 should be clear since in any model $\Imc$ of $(\Tmc,\Sigma)$
  % and $\Amc$ one has to link $a$ with $r$ to $b_{1}$ or to $b_{2}$ to
  % satisfy $\Tmc$.  For Point~2 and $i \in \{1,2\}$, consider the model
  % $\Imc_i$ that corresponds to $\Amc$ expanded with $r(a,b_{i})$. Then
  % $\Imc_i$ is a model of $(\Tmc,\Sigma)$ and $\Amc$ (note that
  % $r\not\in \Sigma$) but $a\not\in (\exists r.(A_{\overline{i}}\sqcap
  % B))^{\Imc_i}$ where $\overline{1}=2$ and $\overline{2}=1$. Thus
  % $\Tmc,\Amc\not\models_{c(\Sigma)} \exists r.(A_{i}\sqcap
  % B)(a)$. When we add any of $A,A_1,A_2$ to $\Sigma$, all statements
  % are still true.
  % To show $\Tmc,\Amc\not\models_{c(\Sigma)} \exists r.(A_{1}\sqcap
  % B)(a)$ one can take the expansion of $\Imc_{\Amc}$ by
  % $\{r(a,b_{2})\}$. Note that it is not relevant whether
  % $A\in\Sigma$.
%
\smallskip
\noindent
The following example illustrates safeness.
\begin{exa}\label{ex:rolehierarchies}
It is easy to see that the TBox with closed predicates from Example~\ref{ex1} is not safe.
As an additional example, consider
$$
\Tmc = \{ A \sqsubseteq \exists r. \top, r \sqsubseteq s\} \quad \text{and} \quad \Sigma_\Csf = \{ s \}
$$
Then $(\Tmc,\Sigma_\Csf)$ is not safe, which is
  witnessed by the concepts $B_1 = A$, $B_2 =\exists s^- . \top$, and
  the role~$r$. Indeed, $(\Tmc,\Sigma_\Csf)$ is not convex for tCQs. This can be proved using the ABox  
$$
\{ A(a), s(a,b_1), A_1(b_1), s(a,b_2), A_2(b_2)\}
$$
and the tCQs   $q_i = \exists y\, (r(x,y) \wedge A_i(y))$, for $i \in \{1,2\}$.
\end{exa}
We now establish the dichotomy result for \dlliter TBoxes with closed predicates.
% Recall that we call an ABox \Amc \emph{consistent w.r.t.\ $(\Tmc,\Sigma_\Csf)$} if there is a model of \Tmc and \Amc that
% respects $\Sigma_\Csf$.
% Clearly, if an ABox is consistent w.r.t.~$(\Tmc,\Sigma_{\Csf})$, then it is consistent
% w.r.t.~$(\Tmc,\emptyset)$. The converse does not hold. For example, if $\Tmc=\{A \sqsubseteq B\}$ and $\Sigma_{\Csf}=\{B\}$,
% then $\Amc=\{A(a)\}$ is not consistent w.r.t.~$(\Tmc,\Sigma_{\Csf})$ but $\Amc$ is consistent w.r.t.~$(\Tmc,\emptyset)$.
% Recall that we say that \emph{ABox consistency is FO-rewritable} for $(\Tmc,\Sigma_\Csf)$ if there is a 
% first-order sentences $\varphi$ such that for all ABoxes~$\mathcal{A}$, we have
% $\Imc_{\Amc} \models \vp$ iff \Amc is consistent w.r.t.\ $(\Tmc,\Sigma_\Csf)$.
Let $(\Tmc_{1},\Sigma_{1})$ and $(\Tmc_{2},\Sigma_{2})$ be TBoxes with closed predicates.
Then we say that $(\Tmc_{1},\Sigma_{1})$ and $(\Tmc_{2},\Sigma_{2})$ are \emph{UCQ-inseparable on consistent ABoxes}~\cite{DBLP:conf/rweb/BotoevaKLRWZ16,DBLP:journals/ai/BotoevaKRWZ16}
if 
  \[\abox\models 
  (\tbox_{1},\Sigma_{1},q)(\vec{a})\quad \text{ iff 
  }  \quad
  \abox\models(\tbox_{2},\Sigma_{2},q)(\vec{a})\]
  holds for all UCQs $q$, all ABoxes $\abox$ consistent w.r.t.\ both
  $(\tbox_{1},\Sigma_{1})$ and $(\tbox_{2},\Sigma_{2})$, and all
  tuples $\vec{a}$ in ${\sf Ind}(\Amc)$. The notion of
  UCQ-inseparability is used in Condition~2(a) of the following
  dichotomy theorem. Informally, it says that tractable query
  evaluation implies that query evaluation with closed predicates
  coincides with query evaluation without closed predicates.
\begin{thm}\label{thm:dlliter_dichotomy}
Let $(\Tmc,\Sigma_{\Csf})$ be a \dlliter TBox with closed predicates.
Then
\begin{enumerate}
\item If $(\Tmc,\Sigma_{\Csf})$ is not safe, then $(\Tmc,\Sigma_{\Csf})$ is not convex for tCQs and tCQ evaluation
w.r.t.~$(\Tmc,\Sigma_{\Csf})$ is \conp-hard.
\item If $(\Tmc,\Sigma_{\Csf})$ is safe, then $(\Tmc,\Sigma_{\Csf})$ is convex for tCQs and
\begin{itemize}
\item[(a)] $(\Tmc,\Sigma_{\Csf})$ and $(\Tmc,\emptyset)$ are UCQ-inseparable on consistent ABoxes.
\item[(b)] UCQ evaluation w.r.t.~$(\Tmc,\Sigma_{\Csf})$ is FO-rewritable.
\end{itemize}
\end{enumerate}
\end{thm}
\begin{proof}
We start with the proof of Point~(1). Assume that $(\Tmc,\Sigma_{\Csf})$ is not safe.
Consider basic concepts $B_{1},B_{2}$ and a role $r$ satisfying Points~(1) to (4) of Definition~\ref{def:safe_dllite_tbox}. 
By Points~(1) and (4) of Definition~\ref{def:safe_dllite_tbox}, $B_{1}$ is satisfiable w.r.t.~$\Tmc$
and $\Tmc\not\models B_{1}\sqsubseteq \exists r'$ for any role $r'$ with ${\sf sig}(r')\subseteq \Sigma_{\Csf}$
and $\Tmc\models r'\sqsubseteq r$. We obtain
$$
\Tmc\not\models B_{1}\sqsubseteq \bigsqcup_{\Tmc\models r'\sqsubseteq r,{\sf sig}(r')\subseteq \Sigma_{\Csf}}\exists r'
$$ 
since $(\Tmc,\emptyset)$ is convex. Observe that the CI to the right is a $\mathcal{ALCI}$ CI and $\Tmc$ is an $\mathcal{ALCHI}$ TBox. It is well known that 
$\mathcal{ALCHI}$ has the finite model property in the sense that any CI that does not follow from an 
$\mathcal{ALCHI}$ TBox is refuted in a finite model of the TBox \cite{Baader-et-al-03b}. Thus, 
we can take a finite model $\Imc$ of $\Tmc$ and some $a_{0}\in B_{1}^{\Imc}$ such that $a_{0}\not\in (\exists r'.\top)^{\Imc}$ 
for any role $r'$ with ${\sf sig}(r')\subseteq \Sigma_\Csf$ and $\Tmc \models r'\sqsubseteq r$.  Let $\Imc_{r}$ be the
interpretation obtained from $\Imc$ by removing all pairs $(a_{0},b)$ from any $r'^{\Imc}$ with $\Tmc\models r'\sqsubseteq r$.
Take the ABox $\Amc_{r}$ corresponding to $\Imc_{r}$ and let $\Amc$ be the disjoint union of two copies of $\Amc_{r}$. 
We denote the individual names of the first copy by $(b,1)$, $b\in \Delta^{\Imc_{r}}$, and the individual names of the
second copy by $(b,2)$, $b\in \Delta^{\Imc_{r}}$. Let $\Amc'$ be defined as
$$
\begin{array}{l}
  \Amc \; \cup \\[1mm]
  \{ A_{1}(b,1)\mid b\in B_{2}^{\Imc}\} 
  \cup \{ A_{2}(b,2) \mid b\in B_{2}^{\Imc}\} \; \cup\\[1mm]
  \{r'((a_{0},i),(b,j))\mid (a_{0},b)\in r'^{\Imc},
  \Tmc \not\models r' \sqsubseteq r,
  {\sf sig}(r')\subseteq \Sigma_\Csf, i,j\in \{1,2\}\}
\end{array}
$$
where $A_{1}$ and $A_{2}$ are fresh concept names. Define, for 
$i\in\{1,2\}$, the tCQs
$$
  q_i(x) = \exists y \,(r(x,y) \wedge A_i(y) \wedge B_2 (y)),
$$
if $B_{2}$ is a concept name. If $B_{2}=\exists s$ (or $B_{2}=\exists s^{-}$), for a role name $s$, then 
set $q_{i}(x)=\exists y,z \,(r(x,y) \wedge A_i(y) \wedge s(y,z))$ (or $q_{i}(x)=\exists y,z \,(r(x,y) \wedge A_i(y) \wedge s(z,y))$, respectively).
We use $\Amc'$ and $q_{i}(x)$ to prove that $(\tbox,\Sigma_{\Csf})$ is not convex for tCQs.
\begin{claim}
$\Amc'\models(\tbox,\Sigma_\Csf,q_1\vee q_{2})(a_{0},1)$.
\end{claim}
\begin{clmproof}
Let $\Jmc$ be a model of $\Tmc$ and $\Amc'$ that respects
$\Sigma_\Csf$. We have $(a_{0},1)\in B_{1}^{\Jmc}$ (since, by
Point~(3) of Definition~\ref{def:safe_dllite_tbox}, $B_{1}\not=\exists r'$ for every $r'$ with $\Tmc\models r' \sqsubseteq r$).  
It follows from the conditions that $\Jmc$ is a model of $\Tmc$, $\Tmc\models B_{1}\sqsubseteq \exists r$, and 
$\Tmc\models \exists r^{-}\sqsubseteq B_{2}$, that there exists $e\in \Delta^{\Jmc}$ with
$((a_{0},1),e)\in r^{\Jmc}$ and $e\in B_{2}^{\Jmc}$. Using the condition that ${\sf sig}(B_{2})\subseteq \Sigma_\Csf$
it follows from the definition of $\Amc'$ that $e$ is of the form $(e',i)$ with $e'\in B_{2}^{\Imc}$ and $i\in\{1,2\}$.
If $i=1$, we have $A_{1}(e',1)\in \Amc'$ and so $(a_{0},1)\in \exists r.(A_{1} \sqcap B_{2})^{\Jmc}$, as required. 
If $i=2$, we have $A_{2}(e',2)\in \Amc'$ and so $(a_{0},1)\in \exists r.(A_{2} \sqcap
B_{2})^{\Jmc}$, as required.
\end{clmproof}
\begin{claim}
$\Amc'\not\models(\tbox,\Sigma_\Csf,q_i)(a_{0},1)$, for $i\in \{1,2\}$.
\end{claim}
\begin{clmproof}
Let $i=1$ (the case $i=2$ is similar and omitted). We construct a model $\Jmc$ of $\Tmc$ and $\Amc'$ that respects
$\Sigma_\Csf$ such that $(a_{0},1)\not\in (\exists r.(A_{1} \sqcap B_{2}))^{\Jmc}$. 
$\Jmc$ is defined as the interpretation corresponding to the ABox $\Amc'$ extended by
$$
 \{r'((a_{0},1),(e,2)) \mid (a_{0},e)\in r'^{\Imc}\} \cup
 \{r'((a_{0},2)),(e,1))\mid (a_{0},e)\in r'^{\Imc}\},
$$
for all roles $r'$ such that ${\sf sig}(r')\cap \Sigma_\Csf=\emptyset$ and
$\Tmc\models r'\sqsubseteq r$, and
$$
\{r'((a_{0},i),(e,j)) \mid (a_{0},e)\in r'^{\Imc}, i,j\in \{1,2\}\},
$$
for all roles $r'$ with ${\sf sig}(r')\cap \Sigma_\Csf=\emptyset$ and
$\Tmc\not\models r'\sqsubseteq r$.

Clearly $(a_{0},1)\not\in (\exists r.(A_{1} \sqcap B_{2}))^{\Jmc}$.
Thus it remains to show that $\Jmc$ is a model of $\Tmc$ and
$\Amc'$ that respects $\Sigma_\Csf$.  Since no symbol from
$\Sigma_\Csf$ has changed its interpretation, it is sufficient to show
that $\Jmc$ satisfies all inclusions in $\Tmc$.

Let $s\sqsubseteq s'$ be an RI in $\Tmc$. Since $\Imc$ is a
model of $\Tmc$, the only pairs where $s\sqsubseteq s'$ can possibly
be refuted are of the form $((a_{0},i),(b,j))$ with $i,j\in \{1,2\}$.
Assume $((a_{0},i),(b,j)) \in s^{\Jmc}$. Then, by definition,
$(a_{0},b)\in s^{\Imc}$ and so $(a_{0},b)\in s'^{\Imc}$ because $\Imc$
is a model of $\Tmc$. We distinguish the following cases:
\begin{itemize}
\item $\Tmc \not\models s'\sqsubseteq r$. Then, by definition of
  $\Jmc$, $((a_{0},i),(b,j))\in s'^{\Jmc}$ since
  $((a_{0},i'),(b,j'))\in s'^{\Jmc}$ for all $i',j'\in \{1,2\}$.
\item $\Tmc\models s'\sqsubseteq r$. Then $\Tmc\models s \sqsubseteq
  r$. Note that, by construction of $\Imc$, ${\sf sig}(s)\cap
  \Sigma_\Csf=\emptyset$ and ${\sf sig}(s')\cap
  \Sigma_\Csf=\emptyset$.  Hence, by construction of $\Jmc$, $(i,j) =
  (1,2)$ or $(i,j)=(2,1)$.  In both cases we have
  $((a_{0},i),(b,j))\in s'^{\Jmc}$ as well.
\end{itemize}
To prove that all CIs of $\Tmc$ are satisfied in $\Jmc$
observe that $B^{\Jmc}= (B^{\Imc} \times \{1\})\cup (B^{\Imc}\times \{2\})$ holds for all 
basic concepts $B$. Thus, $\Jmc$ satisfies all CIs satisfied in 
$\Imc$ and, therefore, is a model of any CI in $\Tmc$, as required.
\end{clmproof}

It follows from Claims~1 and 2 that $(\tbox,\Sigma_{\Csf})$ is not convex for tCQs. The coNP-hardness of
tCQ evaluation follows from Lemma~\ref{lem:nonconvex_imp_conp}. This finishes the proof of Point~(1).
% : in this case, $\Jmc$ is defined as the interpretation
% corresponding to the ABox $\Amc'$ extended by
%$$
%\{r((a_{0},1),(e,1)) \mid (a_{0},e)\in r^{\Imc}\} \cup
% \{r((a_{0},2)),(e,2))\mid (a_{0},e)\in r^{\Imc}%\}
%$$
%Again $\Jmc$ is a model of $(\Tmc,\Sigma)$ and $\Amc'$, and
% $(a_{0},1)\not\in (\exists r.(A_{2} \sqcap B_{2}))^{\Jmc}$.

We come to the proof of Point~(2). The 
proof relies on the canonical model associated with an ABox 
and a \dlliter TBox~\cite{CDLLR07,DBLP:conf/kr/KontchakovLTWZ10}. 
Specifically, for every ABox \Amc 
that is consistent w.r.t.~a $\dlliter$ TBox $\Tmc$ without closed predicates, there is a model 
\Imc of \Amc and \Tmc which is minimal in the sense that for all UCQs $q$ and tuples 
$\vec{a}$ in ${\sf Ind}(\Amc)$:
$$
\Amc\models(\tbox,\emptyset, q)(\vec{a}) \quad \text{iff} \quad \Imc \models q(\vec{a}).
$$
We show that if $\Imc$ is constructed in a careful way and 
$(\Tmc,\Sigma_\Csf)$ is safe, then \Imc respects $\Sigma_\Csf$. This means 
that a tuple $\vec{a}\in\adom{\abox}$ is a certain answer to 
$(\tbox,\emptyset, q)$ on $\abox$ iff it is 
a certain answer to $(\tbox,\Sigma_\Csf,q)$ on $\abox$ since 
closed predicates can only result in additional answers, but not in invalidating answers. It follows that $(\Tmc,\Sigma_{\Csf})$ is
convex for tCQs and that $(\Tmc,\Sigma_{\Csf})$ and 
$(\Tmc,\emptyset)$ are UCQ-inseparable on consistent ABoxes (Point~(a)). 
Additionally, we show that ABox consistency w.r.t.\ $(\Tmc,\Sigma_\Csf)$ is FO-rewritable when 
$(\Tmc,\Sigma_\Csf)$ is safe, which together with the first observation implies that UCQ evaluation w.r.t.\ a safe $(\tbox,\Sigma_{\Csf})$ is FO-rewritable: 
if $p_{c}$ is an FO-rewriting of ABox consistency w.r.t.~$(\Tmc,\Sigma_{\Csf})$ and $q'(\vec{x})$ is an 
FO-rewriting of 
$(\Tmc,\emptyset,q(\vec{x}))$, then $\neg p_{c}\vee q'(\vec{x})$ is an FO-rewriting 
$(\Tmc,\Sigma_{\Csf},q(\vec{x}))$.
Thus, Point~(b) follows. It thus remains to prove Claims~3 and 4 below.

\begin{claim} Let $(\tbox,\Sigma_\Csf)$ be a safe \dlliter TBox with closed
  predicates. Then for every UCQ $q$, we have 
  \[\abox\models 
  (\tbox,\Sigma_\Csf,q)(\vec{a})\quad \text{ iff 
  }  \quad
  \abox\models(\tbox,\emptyset,q)(\vec{a})\]
  for all ABoxes $\abox$ that are consistent w.r.t.\ $(\tbox,\Sigma_\Csf)$ and 
  all $\vec{a}\in\adom{\abox}$.
\end{claim}
\begin{clmproof}
  Let $(\Tmc,\Sigma_\Csf)$ be safe and assume that $\Amc$ is consistent
  w.r.t.~$(\Tmc,\Sigma_\Csf)$. We construct a canonical model of $\Amc$ and $\Tmc$
  as the interpretation corresponding to a (possibly infinite) ABox $\Amc_{c}$ that is the limit
  of a sequence of ABoxes $\Amc_{0},\Amc_{1},\ldots$. Let
  $\Amc_{0}=\Amc$ and assume $\Amc_{j}$ has been defined already. Then $\Amc_{j+1}$ is obtained from $\Amc_{j}$ by applying
  the following two rules:
  % and is satisfiable w.r.t.~$(\Tmc,\Sigma)$.
\begin{itemize}
\item[(R1)] if there exist roles $r,s$ and $a,b\in \NI$ with $\Tmc\models r \sqsubseteq s$, $r(a,b)\in \Amc_{j}$, and $s(a,b)\not\in \Amc_{j}$,
            then add $s(a,b)$ to $\Amc_{j}$;
\item[(R2)] if (R1) does not apply and there are basic concepts $B_{1},B_{2}$ and $a\in \NI$ such that $\Tmc \models B_{1}\sqsubseteq B_{2}$,
      $a\in B_{1}^{\inter_{\Amc_{j}}}$, and $a\not\in B_{2}^{\inter_{\Amc_{j}}}$, then add $B_{2}(a)$ to $\Amc_{j}$ if $B_{2}$ is a concept name
      and add $r(a,b)$ for some fresh $b\in \NI$ to $\Amc_{j}$ if $B_{2}=\exists r$ for some role $r$.
\end{itemize}
% (if no such $i$ exists, then set $\Amc_{c}:=\Amc_{j}$). Then
%   \begin{itemize}
%   \item if $B_{2}$ is a concept name, let $\Amc_{j+1}= \Amc_{j}\cup
%     \{B_{2}(a_{i})\}$;
%   \item if $B_{2}= \exists s$, then take a fresh individual
%     $b_{a_{i},s}$ and set $\Amc_{j+1}$ to $\Amc_{j} \cup
%     \{s(a_{i},b_{a_{i},s})\}$ if $s$ is a role name, or to $\Amc_{j}\cup
%     \{s(b_{a_{i},s},a_{i})\}$ if $s$ is an inverse role.
%   \end{itemize}
We assume that (R1) and (R2) are applied in a fair way. Now let $\Imc_{\Tmc,\Amc}=\inter_{\abox_{c}}$, where $\Amc_{c} = \bigcup_{i\geq 0}\Amc_{i}$.
  % in which exactly the individual names in ${\sf Ind}(\Amc)$ are
  % interpreted.
It is known \cite{DBLP:conf/kr/KontchakovLTWZ10} (and easy to prove) that $\Imc_{\Tmc,\Amc}$ is a model of $\Tmc$ and 
$\Amc$ with the following properties:
  \begin{itemize}
 \item[(p1)] For all UCQs $q(\vec{x})$ and $\vec{a}\in {\sf Ind}(\Amc)$: $\Amc\models (\tbox,\emptyset,q)(\vec{a})$ iff $\Imc_{\Tmc,\Amc}\models q(\vec{a})$.
%
    % for any model $\Imc$ of $\Tmc$ and $\Amc$ there is a
    % homomorphism $h_{i}$ from $\Imc_{\Amc'}$ to $\Imc$ with
    % $h(a)=a^{\Imc}$ for all $a\in \Amc_{c}$
 \item[(p2)] For any individual name $b\in {\sf
      Ind}(\Amc_{c})\setminus {\sf Ind}(\Amc)$ introduced as a witness
    for Rule (R2) for some CI of the form $B_{1}\sqsubseteq \exists s$ and every basic concept $B$
    : $b\in B^{\Imc_{\Tmc,\Amc}}$ iff $\Tmc\models \exists s^{-} \sqsubseteq B$.
  \end{itemize}
  To show that $\Imc_{\Tmc,\Amc}$ is a model of $\Tmc$ and $\Amc$ that respects
  $\Sigma_\Csf$ it is sufficient to prove that every assertion using
  predicates from $\Sigma_\Csf$ in $\Amc_{c}$ is contained in $\Amc$. We first show that for all $a,b\in {\sf Ind}(\Amc)$,
    \begin{itemize}
    \item if $A(a)\in \Amc_{c}$ and $A\in\Sigma_\Csf$, then
      $A(a)\in \Amc$; and
    \item if $r(a,b)\in \Amc_{c}$ and $r\in \Sigma_\Csf$,
      then $r(a,b)\in \Amc$.
    \end{itemize}
    For a proof by contradiction assume that $A(a)\in \Amc_{c}$ but 
    $A(a)\not\in \Amc$ for some concept name $A\in\Sigma_\Csf$. 
    By Point~(p1), the former implies 
    $\Amc\models(\tbox,\emptyset,A(x))(a)$. Thus, $\Amc\models 
    (\tbox,\Sigma_\Csf, A(x))(a)$ which contradicts the assumption that 
    $\Amc$ is consistent w.r.t.~$(\Tmc,\Sigma_\Csf)$. The argument 
    for role assertions $r(a,b)$ is similar and omitted.

It remains to show there are no $a\in {\sf Ind}(\Amc_{c})\setminus{\sf Ind}(\Amc_{0})$ 
and basic concept $B$ with ${\sf sig}(B)\subseteq \Sigma_\Csf$ such 
that $a\in B^{\Imc_{\Tmc,\Amc}}$. For a proof by contradiction, assume that there exist an $a\in {\sf Ind}(\Amc_{c})\setminus{\sf Ind}(\Amc_{0})$ 
and basic concept $B$ with ${\sf sig}(B)\subseteq \Sigma_\Csf$ such that $a\in B^{\Imc_{\Tmc,\Amc}}$. 
Let $a$ be the first such individual name introduced using Rule (R2) in the construction 
of $\Amc_{c}$. By Point~(p2) and the construction of $\Amc_{c}$ there exist $B_{1}$, $r$, $a_{0}$ and $j\geq 0$ such that
$\Tmc\models B_{1} \sqsubseteq \exists r$, $a_{0}\in B_{1}^{\inter_{\abox_j}}$,
    $a_{0}\not\in(\exists r)^{\inter_{\Amc_{j}}}$,
    $(a_{0},a)\in r^{\inter_{\Amc_{j+1}}}$,
    $\Tmc\models \exists r^{-}\sqsubseteq B$.
    We show that $B_{1}$, $B_{2}$, and $r$ satisfy Conditions~1 to 4 from 
    Definition~\ref{def:safe_dllite_tbox} for $B_{2}:=B$ and thus derive a contradiction 
    to the assumption that $(\Tmc,\Sigma_\Csf)$ is safe. Conditions~1 and 
    2 are clear. For Condition~3, assume that $B_{1}= \exists r'$ for 
    some $r'$ such that $\Tmc\models r'\sqsubseteq r$. Then $(a_{0},e)\in 
    (r')^{\inter_{\Amc_{j}}}$ for some $e$. But then, since Rule (R1) is exhaustively
    applied before Rule (R2) is applied, we have $(a_{0},e)\in 
    r^{\inter_{\Amc_{j}}}$ which contradicts $a_{0}\not\in (\exists 
    r)^{\inter_{\Amc_{j}}}$. For Condition~4 assume that $\Tmc\models 
    B_{1} \sqsubseteq \exists r'$ for some role $r'$ such that ${\sf 
    sig}(r')\subseteq \Sigma_\Csf$ and $\Tmc\models r'\sqsubseteq r$. 
    Then $a_{0} \in {\sf Ind}(\Amc)$ because otherwise $a_{0}$ is an individual name 
    introduced before $a$ such that $a_{0}\in(\existsr{r'}{})^{\inter_{\Tmc,\Amc}}$ and 
    ${\sf sig}(\exists r')\subseteq \Sigma_\Csf$, which contradicts our 
    assumption about $a$. By Point~(p1) and the consistency of $\abox$ 
    w.r.t.\ $(\tbox,\Sigma_\Csf)$, there is some $b\in\adom{\abox}$ such 
    that $(a_{0},b)\in (r')^{\inter_\Amc}$. But then, again since Rule (R1) is exhaustively
    applied before Rule (R2) is applied, $(a_{0},b)\in 
    r^{\inter_{\Amc_{j}}}$ which contradicts $a_{0}\not\in (\exists 
    r)^{\inter_{\Amc_{j}}}$. 

Observe that we have also proved that if $\Amc$ is consistent w.r.t.~$(\Tmc,\emptyset)$, then
there do not exist $a\in {\sf Ind}(\Amc_{c})\setminus{\sf Ind}(\Amc_{0})$ 
and a basic concept $B$ with ${\sf sig}(B)\subseteq \Sigma_\Csf$ such that $a\in B^{\Imc_{\Tmc,\Amc}}$. 
\end{clmproof}

As the final step in the proof of Point~(2), we show that ABox consistency w.r.t.\ a safe \dlliter TBox 
with closed predicates is FO-rewritable. 
\begin{claim}
  Let $(\Tmc,\Sigma_\Csf)$ be a safe \dlliter TBox with closed
  predicates. Then ABox consistency w.r.t.~$(\Tmc,\Sigma_{\Csf})$ is FO-rewritable.
\end{claim}
\begin{clmproof}
  It follows immediately from the final remark in the proof of Claim~3 above
  that an ABox $\Amc$ is consistent w.r.t.~a safe 
  $(\Tmc,\Sigma_\Csf)$ if, and only if, (i) $\Amc$ is consistent w.r.t.\ 
  $(\Tmc,\emptyset)$, (ii) $\Amc\models 
  (\Tmc,\emptyset, A(x))(a)$ implies $a\in 
  A^{\inter_\abox}$ for all concept names $A\in \Sigma_\Csf$, (iii) $\Amc\models 
  (\Tmc,\emptyset, \exists y\, r(x,y))(a)$ implies $a\in 
  (\exists r)^{\inter_\abox}$ for all roles $r$ with ${\sf sig}(r)\subseteq \Sigma_\Csf$,
  and (iv) $\Amc\models 
  (\Tmc,\emptyset,r(x,y))(a,b)$ implies 
  $(a,b)\in r^{\inter_\Amc}$ for all role names $r\in \Sigma_\Csf$. 

  To obtain an FO-rewriting of ABox consistency w.r.t.~$(\Tmc,\Sigma_{\Csf})$, let 
  $p_{c}$ be an FO-rewriting of
  ABox consistency w.r.t.~$(\Tmc,\emptyset)$, let $q_{A}(x)$ be an FO-rewriting of 
  $(\Tmc,\emptyset,A(x))$, for $A\in \Sigma_{\Csf}$,
  let $q_{\exists r}(x)$ be an FO-rewriting of $(\Tmc,\emptyset,\exists y\, r(x,y))$ for $r\in \Sigma_{\Csf}$, 
  let $q_{\exists r^{-}}(x)$ be an FO-rewriting of $(\Tmc,\emptyset,\exists y\, r(y,x))$ for $r\in \Sigma_{\Csf}$ and let
  $q_{r}(x,y)= \bigvee_{\Tmc\models s \sqsubseteq r}s(x,y)$, for $r\in \Sigma_{\Csf}$. 
  Then $p_{c} \wedge q_{1} \wedge q_{2} \wedge q_{3} \wedge q_{4}$
with 
\begin{eqnarray*}
q_{1} & = &\forall x \bigwedge_{A\in \Sigma_{\Csf}}(q_{A}(x)\rightarrow A(x))\\
q_{2} & = &\forall x \bigwedge_{r\in \Sigma_{\Csf}}(q_{\exists r}(x)\rightarrow \exists y \,r(x,y))\\
q_{3} & = &\forall x \bigwedge_{r\in \Sigma_{\Csf}}(q_{\exists r^{-}}(x)\rightarrow \exists y \,r(y,x))\\
q_{4} & = &\forall x\forall y \bigwedge_{r\in \Sigma_{\Csf}}(q_{r}(x,y)\rightarrow r(x,y))
\end{eqnarray*}
is an FO-rewriting of ABox consistency w.r.t.~$(\Tmc,\Sigma_{\Csf})$. \hfill $\dashv$ \hfill
$\square$
\phantom{\qedhere}
\end{clmproof}
\phantom{\qedhere}
\end{proof}

% In a sense, Lemmas~\ref{lem:nonconvex_imp_conp}, \ref{dllite1}, and 
% \ref{lem:cqcoincides_dllite} show that CQ evaluation w.r.t.\ \dlliter 
% TBoxes with closed predicates  is inherently intractable: in all cases 
% where closing predicates results in additional answers to queries on 
% satisfiable ABoxes, CQ answering is \conp-hard.  In all tractable 
% cases, the only effect that closing predicates can thus have is to act 
% as integrity constraints on the ABox (but see Section~\ref{sect:foqs} 
% for another virtue of closing predicates) \todo{reference FO 
% closed predicate query lemma}. Note that all TBoxes that refer 
% \emph{only} to closed predicates (thus express only integrity 
% constaints) are safe.

\subsection{Dichotomy for $\mathcal{EL}$}

We show the announced dichotomy for $\EL$ TBoxes with closed
predicates.  While we follow the same strategy as in the DL-Lite case,
there are some interesting new aspects. In particular, we identify an
additional reason for \conp-hardness that we treat by using a variant
of the Craig interpolation property for $\EL$.  We call a concept $E$
a \emph{top-level conjunct (tlc) of} an \EL concept $C$ if $C$ is of
the form $D_{1}\sqcap \cdots \sqcap D_{n}$ and $E=D_{i}$ for some~$i$.
We use the following version of safeness.

\begin{defi}[Safe \EL TBox]
\label{def:elsafe}
  An $\EL$ TBox with closed predicates $(\Tmc,\Sigma_\Csf)$ is \emph{safe}
  if there exists no $\EL$ concept inclusion $C \sqsubseteq \exists r.D$ such
  that
  \begin{enumerate}
  \item\label{it:elsafe_ent} $\Tmc\models C \sqsubseteq \exists r.D$;
  \item\label{it:elsafe_tlc} there is no tlc $\exists r.C'$
    of $C$ with $\Tmc\models C' \sqsubseteq D$;
  \item one of the following is true:
    \begin{enumerate}
      \renewcommand{\theenumii}{(s\arabic{enumii})}
      \renewcommand{\labelenumii}{\theenumii}
    \item\label{it:el_safeness_cond1} $r\not\in\Sigma_\Csf$ and ${\sf
        sig}(D) \cap \Sigma_\Csf\not=\emptyset$;
    \item \label{it:el_safeness_cond2} $r\in \Sigma_\Csf$, ${\sf sig}(D)
      \not\subseteq\Sigma_\Csf$,
% \todo[inline]{is it not redundant to say
%         that $\sig{}{D}\not\subseteq\Sigma$ because of the
%         non-existence of a concept $E$ defined next?} 
      and there is no $\el$ concept $E$ with 
      $\sig{E}{}\subseteq\Sigma_\Csf$, $\Tmc\models C \sqsubseteq 
      \exists r.E$, and $\Tmc\models E \sqsubseteq D$. \hfill$\triangle$
    \end{enumerate}
  \end{enumerate}
\end{defi}

\smallskip
\noindent
Condition~\ref{it:el_safeness_cond1} captures a reason for
non-convexity that is similar to the \dllite case. For example,
we can recast Example~\ref{ex1} using 
%
% An
% example illustrating this condition is given by the TBox with closed
% predicates $(\Tmc,\Sigma_\Csf)$ where
$\Tmc = \{ A \sqsubseteq \exists r . B \}$ and
$\Sigma_\Csf = \{ B \}$. Then the inclusion
$A \sqsubseteq \exists r.B$ shows that
$(\Tmc,\Sigma_\Csf)$ is not safe as $r\not\in\Sigma_\Csf$ and
$B\in \Sigma_\Csf$. % One can also show in exactly the same way as in
% Example~\ref{ex1} that $(\Tmc,\Sigma_\Csf)$ is not convex for tCQs.
However, in \EL there is an additional reason for non-convexity that
is captured by Condition~\ref{it:el_safeness_cond2}. % . For  example,
% we can use the same TBox \Tmc, but replace $\Sigma_\Csf$ by $\{ r\}$.
%
\begin{exa}\label{elex2}
Let $\Tmc = \{A \sqsubseteq \exists r.B\}$ and $\Sigma_\Csf = \{r\}$. Clearly, by 
Condition~\ref{it:el_safeness_cond2}, $(\Tmc,\Sigma_\Csf)$ is not safe.
We show that $(\Tmc,\Sigma_\Csf)$ is not convex for dtCQs. Let
 $$
  \begin{array}{rcl}
    \Amc &=&
  \{A(a),r(a,b_{1}), A_{1}(b_{1}), r(a,b_{2}),A_{2}(b_{2})\} \\[1mm]
  q_i &=& \exists y \, (r(x,y) \wedge A_i(y) \wedge B(y)) 
\end{array}
$$
Then $(\Tmc,\Sigma_\Csf)$ is not convex because $\Amc\models (\Tmc,\Sigma_\Csf, q_1 \vee q_2)(a)$, whereas
$\Amc \not\models (\Tmc,\Sigma_\Csf,q_{i})(a)$ for any $i \in  \{1,2\}$. 
Observe that one cannot reproduce this example in DL-Lite: for example, for the
TBox $\Tmc'= \{A \sqsubseteq \exists r, \exists r^{-} \sqsubseteq B\}$ with $\Sigma'_\Csf=\{r\}$, 
we have $\Amc \models (\Tmc',\Sigma'_\Csf,B(x))(b_{i})$ for $i=1,2$ and thus convexity for dtCQs is not violated.
\end{exa}
Note that Condition~\ref{it:el_safeness_cond2} additionally requires
the non-existence of a certain concept $E$ which can be viewed as an
interpolant between $C$ and $\exists r.D$ that uses only closed
predicates. The following example illustrates why this condition is
needed.
\begin{exa}
  \label{ex2}
  Let $\Tmc= \{A \sqsubseteq \exists r.E,E\sqsubseteq B\}$ and first
  assume that $\Sigma_\Csf= \{r\}$. Then the CI $A\sqsubseteq
  \exists r.B$ satisfies Condition~\ref{it:el_safeness_cond2} and thus
  $(\Tmc,\Sigma_\Csf)$ is not safe.
  % (the inclusion $A\sqsubseteq \exists r.E$ shows this as well).
  Now let $\Sigma'_\Csf=\{r,E\}$. In this case, the CI
  $A\sqsubseteq \exists r.B$ does not violate safeness because $E$ can
  be used as a `closed interpolant'. Indeed, it is not difficult to
  show that $(\Tmc,\Sigma'_\Csf)$ is both safe and convex for dtCQs.
\end{exa}
We now formulate our dichotomy result for $\mathcal{EL}$.
\begin{thm}\label{eltheorem}
  Let $(\Tmc,\Sigma_{\Csf})$ be an $\EL$ TBox with closed predicates. Then
%  the following holds:
  \begin{enumerate}
  \item If $(\tbox,\Sigma_{\Csf})$ is not safe, then $(\Tmc,\Sigma_{\Csf})$ is not
    convex for dtCQs and evaluating dtCQs w.r.t.~$(\Tmc,\Sigma_{\Csf})$ is
    \conp-hard.
  \item If $(\Tmc,\Sigma_{\Csf})$ is safe, then $(\Tmc,\Sigma_{\Csf})$ is convex for tCQs and
    \begin{enumerate}
    \item $(\Tmc,\Sigma_{\Csf})$ and $(\Tmc,\emptyset)$ are UCQ inseparable on consistent ABoxes.
    \item UCQ evaluation w.r.t.\ $(\Tmc,\Sigma_{\Csf})$ is in \ptime.
    \end{enumerate}
  \end{enumerate}
\end{thm}
Towards a proof of Theorem~\ref{eltheorem}, we start with introducing
canonical models, prove some fundamental lemmas regarding such models,
and establish a variant of the Craig interpolation property for \EL
that we use to address Condition~\ref{it:el_safeness_cond2} of
Definition~\ref{def:elsafe}. In fact, we introduce several versions of
canonical models. For the proof of Point~(1) of
Theorem~\ref{eltheorem}, we use \emph{finite} canonical models for \EL
TBoxes and \EL concepts. For Point~(2) and for establishing Craig
interpolation, we use (essentially) \emph{tree-shaped} canonical
models of \EL TBoxes and possibly infinite ABoxes and,
as a special case, the same kind of canonical models of \EL TBoxes and
\EL concepts. The constructions of all these canonical models does not
involve closed predicates. However, to deal with closed predicates in
the proofs, it turns out that we need a more careful definition of
(tree-shaped) canonical models than usual.

We start with the definition of finite canonical models for \EL TBoxes
$\Tmc$ and $\EL$ concepts $C$.  Take for every
$D\in {\sf sub}(\Tmc,C)$ an individual name $a_{D}$ and define the
\emph{canonical model}
$\Imc_{\Tmc,C} =(\Delta^{\inter_{\Tmc,C}},\cdot^{\inter_{\Tmc,C}})$ of
$\tbox$ and $C$ as follows:
\begin{itemize}
\item $\Delta^{\inter_{\Tmc,C}}= \{a_{C}\} \cup \{ a_{C'} \mid \exists
  r.C'\in \mn{sub}(\Tmc,C)\}$;
\item $a_{D} \in A^{\Imc_{\Tmc,C}}$ if $\Tmc \models D \sqsubseteq A$, for all $A \in \NC$ and $a_{D} \in
  \Delta^{\Imc_{\Tmc,C}}$;
\item $(a_{D_{0}},a_{D_{1}})\in r^{{\Imc}_{\Tmc,C}}$ if $\Tmc \models
  D_{0} \sqsubseteq \exists r.D_{1}$ and $\exists r.D_{1} \in
  \mn{sub}(\Tmc)$ or $\exists r . D_{1}$ is a tlc of $D_{0}$, for all
  $a_{D_{0}},a_{D_{1}}\in \Delta^{\Imc_{\Tmc,C}}$ and $r\in \NR$.
\end{itemize}
Deciding whether $\Tmc\models C \sqsubseteq D$ is in \ptime
\cite{BaBrLu-IJCAI-05}, and thus $\inter_{\tbox,C}$
can be constructed in time polynomial in the size of $\tbox$ and $C$. 
The following lemma, shown in \cite{jlc} as Lemma~12, is the reason
for why $\Imc_{\Tmc,C}$ is called a canonical
model.
\begin{lem}\label{lem:el_can_model}
  Let $C$ be an $\EL$ concept and $\tbox$ an $\EL$ TBox. Then
  \begin{itemize}
  \item $\Imc_{\Tmc,C}$ is a model of $\Tmc$;
  \item for all $D_{0}\in {\sf sub}(\Tmc,C)$ and
    all $\EL$ concepts $D_{1}$: $\Tmc \models D_{0} \sqsubseteq D_{1}$
    iff $a_{D_{0}} \in D_{1}^{\mathcal{I}_{\Tmc,C}}$.
  \end{itemize}
\end{lem}
The next lemma, shown in \cite{jlc} as Lemma~16, is concerned
with the implication of existential restrictions in \EL. We will use
it in proofs below. It is proved using Lemma~\ref{lem:el_can_model}
and the construction of canonical models.
\begin{lem}\label{lem:existential_conseq}
  Suppose $\tbox\models C\sqsubseteq\existsr{r}{D}$, where $C$, $D$
  are $\el$ concepts and $\tbox$ is an $\el$ TBox. Then one of the
  following holds:
  \begin{itemize}
  \item there is a tlc $\existsr{r}{C'}$ of $C$ such that
    $\tbox\models C'\sqsubseteq D$;
  \item there is a concept $\existsr{r}{C'}\in\subc{}{\tbox}$ such that
    $\tbox\models C\sqsubseteq \existsr{r}{C'}$ and $\tbox\models
    C'\sqsubseteq D$.
  \end{itemize}
\end{lem}
We next construct tree-shaped canonical models. We start with
canonical models $\Jmc_{\Tmc,\Amc}$ of an $\mathcal{EL}$ TBox $\Tmc$
and a (possibly infinite) ABox $\Amc$.
In the construction of $\Jmc_{\Tmc,\Amc}$, we use \emph{extended
  ABoxes} that additionally admit assertions of the form $C(a)$ with
$C$ an arbitrary \EL concept. We construct a sequence of extended ABoxes
$\Amc_0,\Amc_1,\dots$, starting with $\Amc_0 = \Amc$.  In what
follows, we use additional individual names of the form
$a \cdot r_1\cdot C_1 \cdots r_k\cdot C_k$ with
$a \in \mn{Ind}(\Amc_0)$, $r_1,\dots,r_k$ role names that occur in
\Tmc, and $C_1,\dots,C_k \in \mn{sub}(\Tmc)$. We set
${\sf tail}(a \cdot r_1\cdot C_1 \cdots r_k\cdot C_k)=C_{k}$. Each
extended ABox $\Amc_{i+1}$ is obtained from $\Amc_i$ by applying the
following rules (the interpretation $\Imc_{\Amc_{i}}$ corresponding to
the extended ABox $\Amc_{i}$ ignores assertions $C(a)$ with $C$ not a concept name):
\begin{itemize}

  % \item[{\sf R1}] if $a \in \mn{Ind}(\Amc_i)$, then add $C_\Tmc(a)$.
 
\item[(R1)] if $C \sqcap D(a) \in \Amc_i$, then add $C(a)$ and
  $D(a)$ to $\Amc_{i}$;

\item[(R2)] if %$\inter_{\Amc_{i}}\models C(a)$ 
$a \in C^{\inter_{\Amc_{i}}}$ 
and $C \sqsubseteq D \in
  \Tmc$, then add $D(a)$ to $\Amc_{i}$;

\item[(R3)] if $\exists r.C(a)\in \Amc_{i}$ and there exist $b \in
  \adom{\Amc_{i}}$ with $r(a,b)\in \Amc_{i}$ and $\Amc_{i}\models 
  (\Tmc,\emptyset,q_{C})(b)$, then add $C(b)$ to $\Amc_{i}$; otherwise add $r(a,a \cdot r\cdot C)$ and $C(a\cdot 
  r\cdot C)$ to $\Amc_{i}$. (Recall that $q_{C}$ denotes the directed tree CQ corresponding to the concept $C$.)
\end{itemize}
Let $\Amc_c = \bigcup_{i \geq 0} \Amc_i$. Note that $\Amc_c$ may be
infinite even if $\Amc$ is finite. Also note that rule (R3) carefully 
avoids to introduce fresh successors as witnesses for existential
restrictions when this is not strictly necessary. This will be useful
when closing predicates which might preclude the introduction of
fresh successors. 
% , and that none of the above rules
% adds anything to $\Amc_{c}$.
% In the following, we write $\Amc_{c} \vdash \bot$ instead of
% $\Amc_{c} \cdash \bot(a)$.  e $\Amc_{c} \vdash \In what follows we
% denote by $\Imc_{\Tmc,\Amc}$ the interpretation $\Imc_{\Amc_{c}}$
% corresponding to $\Imc_{\Amc}$ extended by setting
% $a^{\Imc_{\Tmc,\Amc}}= a$, for $a\in {\sf Ind}(\Amc)$.
%
Let $\Jmc_{\Tmc,\Amc}$
be the interpretation that corresponds to $\Amc_{c}$.
% We can view $\Jmc_{\Tmc,\Amc}$
% as a canonical model for the ABox \Amc w.r.t.\ the TBox \Tmc. 
Points~1 and~2 of the following lemma show that $\Jmc_{\Tmc,\Amc}$
is canonical, essentially in the sense of
Lemma~\ref{lem:el_can_model}, and Points~3 and~4 show that, in
addition, it is \emph{universal} for UCQs: answers given by
$\Jmc_{\Tmc,\Amc}$ coincide with the certain answers.
\begin{lem}\label{aboxcan}
  Let $\Tmc$ be an $\EL$ TBox and $\Amc$ a possibly infinite ABox. Then
  \begin{enumerate}
  \item $\Jmc_{\Tmc,\Amc}$ is a model of $\Tmc$ and $\Amc$;
  \item for all $p\in \Delta^{\Jmc_{\Tmc,\Amc}}\setminus{\sf
      Ind}(\Amc)$ and all $\EL$ concepts $D$: $p\in
    D^{\Jmc_{\Tmc,\Amc}}$ iff \mbox{$\Tmc \models {\sf
        tail}(p)\sqsubseteq D$};
  \item for every model $\inter$ of $\tbox$ and $\abox$, there is a
    homomorphism $h$ from $\interj_{\tbox,\abox}$ to $\inter$ that preserves
         $\adom{\abox}$;
  \item for all UCQs $q(\vec{x})$ and tuples $\vec{a}$ in $\adom{\Amc}$: 
    $\Amc\models (\Tmc,\emptyset,q)(\vec{a})$ iff $\Jmc_{\Tmc,\Amc}\models q(\vec{a})$.
  \end{enumerate}
\end{lem}
We now construct tree-shaped canonical models $\Jmc_{\Tmc,C}$
of an \EL TBox~\Tmc and an \EL \emph{concept} $C$.
A \emph{path} in $C$
is a finite sequence $C_{0}\cdot
r_{1}\cdot C_{1}\cdots r_{n}\cdot C_{n}$, where $C_{0}=C$, $n\geq
0$, and $\exists r_{i+1}.C_{i+1}$ is a tlc of $C_{i}$, for $0\leq i
<n$. We use ${\sf paths}(C)$ to denote the set of paths in
$C$.
If $p\in
{\sf paths}(C)$, then ${\sf tail}(p)$ denotes the last element of
$p$.
The \emph{ABox} $\Amc_{C}$ associated with $C$ is defined by setting
\begin{eqnarray*}
  \Amc_{C} & = & \{ r(p,q) \mid p,q\in {\sf paths}(C); q = p \cdot r\cdot C'\}\\
  &  &      \{ A(p)\mid A \mbox{ a tlc of } {\sf tail}(p), p \in {\sf paths}(C)\}.
  % & & \{ \top(a_{p}) \mid \top \mbox{ a conjunct of } {\sf
  % Tail}(p),p\in {\sf paths}(C)\}
\end{eqnarray*}
Then $\Jmc_{\Tmc,C}:= \Jmc_{\Tmc,\Amc_{C}}$ is the \emph{tree-shaped canonical model} of $\Tmc$ and $C$. 
The following is an easy consequence of Lemma~\ref{aboxcan}.
\begin{lem}\label{conceptcan}
  Let $\Tmc$ be an $\EL$ TBox and $C$ an $\el$ concept. Then
  \begin{itemize}
  \item $\Jmc_{\Tmc,C}$ is a model of $\Tmc$;
  \item for all $p\in \Delta^{\Jmc_{\Tmc,C}}$ and all $\EL$ concepts
    $D$: $p\in D^{\Jmc_{\Tmc,C}}$ iff $\Tmc \models {\sf
      tail}(p)\sqsubseteq D$.
    % \item for all CQs $q(\vec{x})$ and $\vec{a}\subseteq {\sf
    %   Ind}(\Amc)$: $\Tmc,\Amc\models q(\vec{a})$ iff
    %   $\Jmc_{\Tmc,\Amc}\models q[\vec{a}]$.
  \end{itemize}
\end{lem}
We next give a lemma that connects an answers $a$
to a dtCQs $q_C$
on an ABox \Amc under a TBox~\Tmc with the entailment by \Tmc of
concept inclusions of the form $C_a^m
\sqsubseteq C$ where $C_a^m$ is obtained by unfolding \Amc at
$a$ up to depth $m$. More precisely, for every $m \geq 0$ define
%
% entailed  and certain answers over ABoxes. 
% Let $\Amc$ be an ABox. For $a\in {\sf Ind}(\Amc)$, we define an $\EL$ concept $C_{a}^{m}$ by ``unfolding'' $\Amc$ at 
% $a$ up to depth $m$:
$$
C_{a}^{0} = (\bigsqcap_{A(a)\in \Amc} A), \quad C_{a}^{m+1}=
(\bigsqcap_{A(a)\in \Amc} A)\sqcap (\bigsqcap_{r(a,b)\in \Amc} \exists
r.C_{b}^{m}).
$$
The following is shown in \cite{jlc} as~Lemma~22.
\begin{lem}\label{ac}
For all $\EL$ TBoxes $\Tmc$, $\mathcal{EL}$ concepts $C$, ABoxes $\Amc$, and $a\in {\sf Ind}(\Amc)$:
$$
\Amc\models (\Tmc,\emptyset,q_{C})(a) \quad \text{ iff } \quad \exists m\geq 0: \quad
\Tmc\models C_{a}^{m}\sqsubseteq C
$$
\end{lem}
We now establish a variant of the Craig interpolation property that is
suitable for addressing Condition~\ref{it:el_safeness_cond2} of
Definition~\ref{def:elsafe}. It has been studied before for $\alc$
and several of its extensions in the context of query rewriting for
DBoxes and of Beth definability~\cite{seylan09effective,tencate13beth}.
Note that it is different from the interpolation property investigated
in~\cite{jlc} for \el, which requires the interpolant to be a TBox
instead of a concept. For brevity, we set ${\sf
  sig}(\Tmc,C)={\sf sig}(\Tmc)\cup {\sf sig}(C)$ for any TBox~$\Tmc$ and concept $C$.

\begin{lem}[\EL Interpolation]\label{lem:interpolation}
  Let $\Tmc_{1},\Tmc_{2}$ be $\mathcal{EL}$ TBoxes and let 
  $D_{1},D_{2}$ be $\EL$ concepts with $\Tmc_{1}\cup \Tmc_{2} \models 
  D_{1} \sqsubseteq D_{2}$ and ${\sf sig}(\Tmc_{1},D_{1}) \cap {\sf 
  sig}(\Tmc_{2},D_{2})=
  \Sigma$. Then there exists an $\EL$ concept $F$ such that $\sig{F}{}\subseteq\Sigma$,
  $\Tmc_{1}\cup \Tmc_{2}\models D_{1} \sqsubseteq F$, and
  \mbox{$\Tmc_{1}\cup \Tmc_{2}\models F \sqsubseteq D_{2}$}.
\end{lem}
\begin{proof}
  Let $\Tmc_{1}\cup \Tmc_{2} \models D_{1} \sqsubseteq D_{2}$ with 
  ${\sf sig}(\Tmc_{1},D_{1}) \cap {\sf sig}(\Tmc_{2},D_{2})= \Sigma$. 
  Assume that the required $\el$ concept $F$ does not exist. Consider 
  the tree-shaped canonical model $\Jmc_{\Tmc_{1}\cup \Tmc_{2},D_{1}}$. 
  Denote by $\Amc_{\Sigma}$ the ABox corresponding to the 
  $\Sigma$-reduct of $\Jmc_{\Tmc_{1}\cup \Tmc_{2},D_{1}}$, thus
  \[\abox_{\Sigma}=\bigcup_{A\in\Sigma}\{A(a) \mid a\in 
  A^{\interj_{\tbox_{1}\cup\tbox_{2},D_{1}}}\}\cup \bigcup_{r\in \Sigma}\{r(a,b) \mid r(a,b)\in 
  r^{\interj_{\tbox_{1}\cup\tbox_{2},D_{1}}}\}
 \]
%\cup\bigcup_{a\in\Delta^{\interj_{\tbox_{1}\cup\tbox_{2},D_{0}}}}\top(a)\] 
  We may assume w.l.o.g.~that ${\sf Ind}(\Amc_{\Sigma})= \Delta^{\Jmc_{\Tmc_{1}\cup \Tmc_{2},D_{1}}}$. Recall that
  the individual names in $\Amc_{\Sigma}$ are paths. For the sake of readability, we denote them by $a_{p}$
  rather than $p$. Also recall that $q_{D}$ denotes the dtCQ corresponding to the $\mathcal{EL}$ concept $D$.

  \begin{claim*}
    $\Amc_{\Sigma}\not\models (\Tmc_{1}\cup \Tmc_{2},\emptyset,q_{D_{2}})(a_{D_{1}})$.
  \end{claim*}
  \begin{clmproof}
    Assume for a proof by contradiction that $\Amc_{\Sigma}\models (\Tmc_{1}\cup 
    \Tmc_{2},\emptyset,q_{D_{2}})(a_{D_{1}})$.  
    By Lemma~\ref{ac}, there is an $\el$ concept $F$ such that 
    $\sig{F}{}\subseteq\Sigma$, $\Tmc_{1}\cup \Tmc_{2}\models F 
    \sqsubseteq D_{2}$, and $a_{D_1} \in F^{\Jmc_{\Tmc_{1}\cup \Tmc_{2},D_{1}}}$. Then, using Lemma~\ref{conceptcan}, we 
    obtain $\Tmc_{1}\cup \Tmc_{2}\models D_1 \sqsubseteq F$. This contradicts our assumption that no such concept $F$ exists.
  \end{clmproof}
    
  Obviously, $\interj_{\tbox_{1}\cup\tbox_{2},D_{1}}$ is a model of 
  $\tbox_{1}\cup\tbox_{2}$ and $\abox_{\Sigma}$. Then, by 
  Lemma~\ref{aboxcan}, there is a homomorphism $h$ from 
  $\Jmc_{\Tmc_{1}\cup \Tmc_{2},\Amc_{\Sigma}}$ to 
  $\interj_{\tbox_{1}\cup\tbox_{2},D_{1}}$ with $h(a)=a$ for all 
  $a\in\adom{\abox_{\Sigma}}$. Conversely, $h'=\{a\mapsto a\mid 
  a\in\Delta^{\interj_{\tbox_{1}\cup\tbox_{2},D_{1}}}\}$ is a 
  homomorphism from the $\Sigma$-reduct of $\interj_{\tbox_{1}\cup\tbox_{2},D_{1}}$ 
  to $\Jmc_{\Tmc_{1}\cup \Tmc_{2},\Amc_{\Sigma}}$. Define the interpretation $\inter$ as follows:
  \begin{eqnarray*} 
    \domain & = & \Delta^{\interj_{\tbox_{1}\cup\tbox_{2},\abox_{\Sigma}}}\\
    \ext{P} & = & P^{\interj_{\tbox_{1}\cup\tbox_{2},\abox_{\Sigma}}}\cup 
      P^{\interj_{\tbox_{1}\cup\tbox_{2},D_{1}}}, \text{ for all $P\in\sig{\tbox_{1}}{D_{1}}\setminus\Sigma$}\\
    \ext{P}  & = & P^{\interj_{\tbox_{1}\cup\tbox_{2},\abox_{\Sigma}}}, \text{ for all $P\not\in \sig{\tbox_{1}}{D_{1}}\setminus\Sigma$}
 \end{eqnarray*}
  Observe that the mapping $h$ defined above is a homomorphism from $\inter$ to 
  $\interj_{\tbox_{1}\cup\tbox_{2},D_{1}}$ with $h(a)=a$ for all 
  $a\in\adom{\abox_{\Sigma}}$. Conversely, the mapping $h'$ defined above is a homomorphism from the $\sig{\tbox_{1}}{D_{1}}$-reduct of 
  $\interj_{\tbox_{1}\cup\tbox_{2},D_{1}}$ to $\inter$. Now it is readily checked that $\EL$ concepts $C$ are 
  preserved under homomorphisms (if $d\in C^{\Imc_{1}}$,
  then $h(d)\in C^{\Imc_{2}}$ if $h$ is a homomorphism from $\Imc_{1}$ to $\Imc_{2}$).
  Thus, $\Imc$ is a model of $\Tmc_{1}$ since $\interj_{\tbox_{1}\cup\tbox_{2},D_{1}}$ is a model of $\Tmc_{1}$ and 
  $a_{D_{1}}\in D_{1}^{\Imc}$ since $a_{D_{1}}\in D_{1}^{\interj_{\tbox_{1}\cup\tbox_{2},D_{1}}}$. Moreover, by construction, $\Imc$ is a model of $\Tmc_{2}$ and 
  $a_{D_{1}}\not\in D_{2}^{\Imc}$, by the claim proved above.
  We have shown that $\Tmc_{1}\cup \Tmc_{2}\not\models D_{1}\sqsubseteq D_{2}$ and thus derived a contradiction.
\end{proof}

We are now in the position to prove Theorem~\ref{eltheorem}. We first prove Point~(1).
The proof requires two separate constructions that both show non-convexity for dtCQs and address 
Cases~\ref{it:el_safeness_cond1} and~\ref{it:el_safeness_cond2} from 
Definition~\ref{def:elsafe}. It then follows from Lemma~\ref{lem:nonconvex_imp_conp} that
dtCQ evaluation w.r.t.~($\Tmc,\Sigma_{\Csf})$ is \conp-hard.

We begin by considering Case~\ref{it:el_safeness_cond1}.

\begin{lem}\label{lem:point1case1}
  Let $(\Tmc,\Sigma_\Csf)$ be an $\mathcal{EL}$ TBox with closed 
  predicates such that safeness is violated by an inclusion $C 
  \sqsubseteq \exists r.D$ because Condition~\ref{it:el_safeness_cond1} from Definition~\ref{def:elsafe}
  holds. Then $(\tbox,\Sigma_{\Csf})$ is not 
  convex for dtCQs.
\end{lem}
\begin{proof}
  Assume $C \sqsubseteq \exists r.D$ satisfies $\Tmc\models C \sqsubseteq \exists r.D$,
  there is no tlc $\exists r.C'$ of $C$ with $\Tmc\models C' \sqsubseteq D$, $r\not\in\Sigma_\Csf$,
  and ${\sf sig}(D) \cap \Sigma_\Csf\not=\emptyset$.
  Consider the finite canonical model
  $\Imc_{\Tmc,C}$ of $\Tmc$ and $C$. Assume w.l.o.g.\ that $C$ does not
  occur in $\Tmc$ (if it does, replace $C$ by $A\sqcap C$ for a fresh
  concept name $A$). Note that it follows that there is no
  $a\in \Delta^{\Imc_{\Tmc,C}}$ with $(a,a_{C})\in
  s^{\Imc_{\Tmc,C}}$ for any role name $s$. 
  
Let $\Imc_{r}$ be the interpretation obtained from $\Imc_{\Tmc,C}$ by removing all pairs $(a_{C},a_{E})$
from $r^{\Imc_{\Tmc,C}}$ such that $\exists r.E$ is not a tlc of $C$. Let $\Amc_{r}$ be the ABox corresponding 
to $\Imc_{r}$ and let $\Amc$ be the disjoint union of
two copies of $\Amc_{r}$. We denote the individual names of the first copy by $(a,1)$, $a\in \Delta^{\Imc_{\Tmc,C}}$, and the
individual names of the second copy by $(a,2)$, $a\in \Delta^{\Imc_{\Tmc,C}}$. Let $A_{1}$ and $A_{2}$ be fresh concept
names and set 
$$
\Amc'= \Amc \cup \{ A_{1}(a,1)\mid a\in \Delta^{\Imc_{\Tmc,C}}\} \cup
\{ A_{2}(a,2) \mid a\in \Delta^{\Imc_{\Tmc,C}}\}
$$   
Some predicate $P\in \Sigma_{\Csf}$ occurs in $D$. If a concept name $E\in \Sigma_\Csf$ occurs in $D$, then fix one such
$E$ and denote, for $i\in \{1,2\}$, by $D_{i}$ the resulting concept after one occurrence
of $E$ is replaced by $A_{i}\sqcap E$. For example, if $D=A \sqcap \exists s_{1}.E \sqcap \exists s_{2}.E$, $E\in \Sigma_{\Csf}$
and $A\not\in \Sigma_{\Csf}$, then either $D_{i}= A \sqcap \exists s_{1}.(A_{i} \sqcap E) \sqcap \exists s_{2}.E$ or
$D_{i} = A \sqcap \exists s_{1}.E \sqcap \exists s_{2}.(A_{i} \sqcap E)$. Similarly, if no concept name
from $\Sigma_\Csf$ occurs in $D$, then let $s\in \Sigma_\Csf$ be a role name such that a
concept of the form $\exists s.G$ occurs in $D$. Denote by $D_{i}$
the resulting concept after one occurrence of $\exists s.G$ is
replaced by $A_{i}\sqcap \exists s.G$.

We now use $\Amc'$ and the dtCQs $q_{\exists r.D_{i}}$ to prove that $(\Tmc,\Sigma_{\Csf})$ is not convex for dtCQs.
Using the condition $\Tmc\models C \sqsubseteq \exists r.D$ and the construction of $D_{1}$ and $D_{2}$, 
it is straightforward to show that $\Amc'\models (\Tmc,\Sigma_\Csf,q_{\exists r.D_{1}}\vee q_{\exists r.D_{2}})(a_{C},1)$.
We show that $\Amc'\not\models (\Tmc,\Sigma_\Csf,q_{\exists r.D_{i}})(a_{C},1)$ for $i=1,2$.
Let $i=1$ (the case $i=2$ is similar and omitted). We construct a model $\Jmc$ of
$\Tmc$ and $\Amc'$ that respects $\Sigma_\Csf$ with $(a_{C},1)\not\in (\exists r.D_{1})^{\Jmc}$. 
$\Jmc$ is defined as the interpretation corresponding to the ABox $\Amc'$ extended by
  $$
  \{r((a_{C},1),(e_{E},2)),r((a_{C},2),(a_{E},1)) \mid (a_{C},a_{E})\in r^{\Imc_{\Tmc,C}}\setminus r^{\Imc_{r}}\} 
  $$
  Using the fact that $\Imc_{\Tmc,C}$ is a model of $\Tmc$ it is readily checked that
  $\Jmc$ is a model of $\Tmc$ and $\Amc'$ that respects $\Sigma_\Csf$.  Moreover,
  $(a_{C},1)\not\in (\exists r.D_{1})^{\Jmc}$. To prove this assume
  $(a_{C},1)\in (\exists r.D_{1})^{\Jmc}$.  Then one of the following
  two conditions holds:
  \begin{itemize}
  \item there exists a tlc $\exists r.E$ of $C$ such that
    $(a_{E},1)\in D_{1}^{\Jmc}$;
  \item there exists $a_{E}$ with $(a_{C},a_{E})\in r^{\Imc_{\Tmc,C}}$
    such that $(a_{E},2)\in D_{1}^{\Jmc}$.
  \end{itemize}
  If the first condition holds, then $a_{E}\in D^{\Imc_{\Tmc,C}}$. Then, by 
  Lemma~\ref{lem:el_can_model}, $\Tmc\models E \sqsubseteq D$ for
  a tlc $\exists r.E$ of $C$ which contradicts Point~(2) of the definition of safeness.
  The second condition does not
  hold since $(a_{E},2)\in G^{\Jmc}$ iff $(a_{E},2)\in G^{\Jmc|_{\{(a,2) \mid a\in \Delta^{\Imc}\}}}$,
  for every $\EL$ concept $G$, and $A_{1}^{\Jmc} \cap \{(a,2) \mid a\in \Delta^{\Imc}\}=\emptyset$,
  but $D_{1}$ contains $A_{1}$.
  \end{proof}
We now consider Case~\ref{it:el_safeness_cond2} from Definition~\ref{def:elsafe}.
\begin{lem}\label{lem:point1case2}
  Let $(\Tmc,\Sigma_\Csf)$ be an $\mathcal{EL}$ TBox with closed 
  predicates such that safeness is violated by an inclusion $C 
  \sqsubseteq \exists r.D$ because Condition~\ref{it:el_safeness_cond2} from Definition~\ref{def:elsafe}
  holds. Then $(\tbox,\Sigma_{\Csf})$ is not 
  convex for dtCQs.
\end{lem}
\begin{proof}
  Assume $C \sqsubseteq \exists r.D$ satisfies $\Tmc\models C \sqsubseteq \exists r.D$,
  there is no tlc $\exists r.C'$ of $C$ with $\Tmc\models C' \sqsubseteq D$, and 
  Condition~\ref{it:el_safeness_cond2} holds. 
  Let  
$$
K= \{ G \mid \exists r.G \in {\sf sub}(\Tmc), \Tmc \models C
\sqsubseteq \exists r.G\}
$$ 
Observe that since there is no tlc $\exists r . C'$ of $C$ with $\Tmc\models C'
\sqsubseteq D$, by Lemma~\ref{lem:existential_conseq}, there exists
$G\in K$ with $\Tmc\models G \sqsubseteq D$. We now apply the interpolation lemma.
Obtain $\Tmc^{i}$ from $\Tmc$ by replacing every predicate $P\not\in\Sigma_{\Csf}$ by a 
fresh predicate $P_{i}$ of the same arity, $i\in \{1,2\}$. Similarly, for any \EL concept
$F$ we denote by $F^{i}$ the resulting concept when every predicate $P\not\in\Sigma_{\Csf}$ is
replaced by a fresh predicate $P_{i}$ of the same arity, $i\in \{1,2\}$. We show the following
using the interpolation lemma.
\begin{claim}
  For all $G\in K$: $\Tmc^{1}\cup \Tmc^{2}\not\models
  G^{1}\sqsubseteq D^{2}$.
\end{claim}
\begin{clmproof}
The proof is indirect. Assume there exists $G\in K$ such that $\Tmc^{1}\cup
\Tmc^{2}\models G^{1}\sqsubseteq D^{2}$.  By Lemma~\ref{lem:interpolation}, there exists an $\EL$ concept
$F$ with ${\sf sig}(F)\subseteq \Sigma_{\Csf}$
such that $\Tmc^{1}\cup \Tmc^{2}\models G^{1}\sqsubseteq F$ and
$\Tmc^{1}\cup \Tmc^{2}\models F\sqsubseteq D^{2}$. Then $\Tmc\models G
\sqsubseteq F$ and $\Tmc\models F \sqsubseteq D$. But then we obtain 
from $\Tmc\models C \sqsubseteq \exists r.G$ that $\Tmc\models C \sqsubseteq \exists
r.F$ which contradicts Condition~\ref{it:el_safeness_cond2}.
\end{clmproof}

By Claim~1 we can take the finite canonical models
$\Jmc_{G}:=\Imc_{\Tmc^{1}\cup \Tmc^{2},G^{1}}$, $G\in K$, and
obtain for $a_{G}:=a_{G^{0}}$ that $a_{G}\not\in (D^{2})^{\Jmc_{G}}$.
Let $\Amc_{G,\Sigma_\Csf}$ be the ABox corresponding to the $\Sigma_\Csf$-reduct 
of $\Jmc_{G}$. We may assume that the sets of individual names ${\sf Ind}(\Amc_{G,\Sigma_\Csf})$ are
mutually disjoint, for $G\in K$, and that $a_{G}\in {\sf Ind}(\Amc_{G,\Sigma_\Csf})$, for all $G\in K$.

\begin{claim} For every $G\in K$, there exist
\begin{itemize}
\item a model $\Imc_{G}^{1}$ of $\Tmc$ and $\Amc_{G,\Sigma_\Csf}$ that respects the closed predicates $\Sigma_{\Csf}$
such that $\Delta^{\Imc_{G}^{1}}={\sf Ind}(\Amc_{G,\Sigma_\Csf})$, $a_{G}\in G^{\Imc_{G}^{1}}$,
and $a_{G}\in H^{\Imc_{G}^{1}}$ only if $\Tmc\models G\sqsubseteq H$, for all 
$\EL$ concepts $H$;

\item a model $\Imc_{G}^{2}$ of $\Tmc$ and $\Amc_{G,\Sigma_\Csf}$ that respects the closed predicates $\Sigma_{\Csf}$
such that $\Delta^{\Imc_{G}^{2}}={\sf Ind}(\Amc_{G,\Sigma_\Csf})$, $a_{G}\not\in D^{\Imc_{G}^{2}}$, and $a_{G}\in H^{\Imc_{G}^{2}}$
only if $\Tmc\models G\sqsubseteq H$, for all $\EL$ concepts $H$.
\end{itemize}
\end{claim}
\begin{clmproof}
The interpretation $\Imc_{G}^{1}$ is obtained from $\Jmc_{G}$ by setting $P^{\Imc_{G}^{1}}:= (P^{1})^{\Jmc_{G}}$ for all predicates $P\in {\sf sig}(\Tmc,C,D)\setminus \Sigma_\Csf$ 
and $P^{\Imc_{G}^{1}}:= \emptyset$ for all predicates $P$ not in ${\sf sig}(\Tmc,C,D) \cup \Sigma_{\Csf}$. The properties stated follow from the properties of the finite canonical
model $\Imc_{\Tmc^{1}\cup \Tmc^{2},G^{1}}$. In particular, $a_{G}\in H^{\Imc_{G}^{1}}$ only if $\Tmc\models G\sqsubseteq H$ follows from Lemma~\ref{lem:existential_conseq}, Point~2.
The interpretation $\Imc_{G}^{2}$ is obtained from $\Jmc_{G}$ by setting $P^{\Imc_{G}^{2}}:= (P^{2})^{\Jmc_{G}}$ for all predicates $P\in {\sf sig}(\Tmc,C,D)\setminus \Sigma_\Csf$ 
and $P^{\Imc_{G}^{2}}:= \emptyset$ for all predicates $P$ not in ${\sf sig}(\Tmc,C,D) \cup \Sigma_{\Csf}$. The properties stated follow again 
from the properties of the finite canonical model $\Imc_{\Tmc^{1}\cup \Tmc^{2},G^{1}}$.
\end{clmproof}

Introduce two copies $\Amc_{G,\Sigma_\Csf}^{1}$ and $\Amc_{G,\Sigma_\Csf}^{2}$
of $\Amc_{G,\Sigma_\Csf}$, for $G\in K$.  We denote the individual names of the
first copy by $(a,1)$, for $a\in {\sf Ind}(\Amc_{G,\Sigma_\Csf})$, and the
individual names of the second copy by $(a,2)$, for $a\in {\sf Ind}(\Amc_{G,\Sigma_\Csf})$. 
Let $\Amc_{r}$ be the ABox defined in the beginning of the proof of Lemma~\ref{lem:point1case1}. 
Define the ABox $\Amc$ by taking two fresh concept names $A_{1}$ and $A_{2}$ and adding
to
$$
\Amc_{r} \cup \bigcup_{G\in K}\Amc_{G,\Sigma_\Csf}^{1} \cup
\Amc_{G,\Sigma_\Csf}^{2}
$$
the assertions
\begin{itemize}
\item $r(a_{C},(a_{G},1)),r(a_{C},(a_{G},2))$, for every $G\in K$;
\item $A_{1}(a_{G},1)$, for every $G\in K$;
\item $A_{1}(a_{E})$, for every tlc $\exists r.E$ of $C$;
\item $A_{2}(a_{G},2)$, for every $G\in K$.
\end{itemize}
We use $\Amc$ and the dtCQs $q_{\exists r.(A_{i}\sqcap D)}$ to show that $(\Tmc,\Sigma_{\Csf})$ is not convex for dtCQs.
The proof that $\Amc\models (\Tmc,\Sigma_\Csf,q_{\exists r.(A_{1}\sqcap D)}\vee q_{\exists r.(A_{2}\sqcap D)})(a_{C})$ is straightforward
using the condition that $\Tmc\models C \sqsubseteq \exists r.D$ and the construction of $\Amc$ ($r\in \Sigma_{\Csf}$ and all $r$-successors
of $a_{C}$ in $\Amc$ are either in $A_{1}$ or in $A_{2}$).
It remains to show $\Amc\not\models (\Tmc,\Sigma_\Csf,q_{\exists r.(A_{i}\sqcap D)})(a_{C})$, for $i=1,2$.
For $i=2$, construct a witness interpretation $\Jmc$ showing this 
by expanding all $\Amc_{G,\Sigma_\Csf}^{2}$, $G\in K$, to
(isomorphic copies of) $\Imc_{G}^{2}$, all $\Amc_{G,\Sigma_\Csf}^{1}$, $G\in K$, to (isomorphic copies of) $\Imc_{G}^{1}$, and
$\Amc_{r}$ to $\Imc_{r}$. Using the properties of $\Imc_{G}^{1}$ and
$\Imc_{G}^{2}$ established in the claim above, it is readily checked that $\Jmc$ is a
model of $\Tmc$ and $\Amc$ that respects $\Sigma_\Csf$. Moreover, $a_{C}\not\in (\exists r.(A_{2}\sqcap D))^{\Jmc}$ 
since $(a_{G},2)\not\in D^{\Jmc}$ for any $G\in K$ (by Claim~2).
  
For $i=1$, construct a witness interpretation $\Jmc$ showing this by
expanding all $\Amc_{G,\Sigma_\Csf}^{2}$, $G\in K$, to (isomorphic copies of) $\Imc_{G}^{1}$, 
all $\Amc_{G,\Sigma_\Csf}^{1}$, $G\in K$, to (isomorphic copies of) $\Imc_{G}^{2}$, and $\Amc_{r}$ to $\Imc_{r}$.
Using again the properties of $\Imc_{G}^{1}$ and $\Imc_{G}^{2}$ established above, it can
be checked that $\Jmc$ is a model of $\Tmc$ and $\Amc$ that respects $\Sigma_\Csf$. 
$a_{C}\not\in (\exists r.(A_{1}\sqcap D))^{\Jmc}$ since $a_{E}\not\in D^{\Jmc}$ for any tlc $\exists r.E$ of 
$C$ and since $(a_{G},1)\not\in D^{\Jmc}$ for any $G\in K$ (by Claim~2).
This finishes the proof.
\end{proof}

This finishes the proof of Point~(1) of Theorem~\ref{eltheorem}. We now prove Part~(a) of Point~(2). 
The proof strategy is exactly the same as in the proof for \dlliter.

\begin{lem}\label{lem:point2parta} Let $(\tbox,\Sigma_\Csf)$ be a safe $\EL$ TBox with closed 
  predicates. Then for every UCQ $q$, we have 
  \[\abox\models 
  (\tbox,\Sigma_\Csf,q)(\vec{a})\quad \text{ iff 
  }  \quad
  \abox\models(\tbox,\emptyset,q)(\vec{a})\]
  for all ABoxes $\abox$ that are consistent w.r.t.\ $(\tbox,\Sigma_\Csf)$ and all $\vec{a}\in\adom{\abox}$.
\end{lem}
\begin{proof}
  Let $(\Tmc,\Sigma_\Csf)$ be safe and assume that $\Amc$ is consistent
  w.r.t.~$(\Tmc,\Sigma_\Csf)$.
  We consider the tree-shaped canonical model $\Jmc_{\Tmc,\Amc}$ introduced above. It suffices to show that
  $\Jmc_{\Tmc,\Amc}$ respects $\Sigma_{\Csf}$. To this end it suffices to prove for all $A,r\in \Sigma_{\Csf}$:
\begin{enumerate}
\item for all $a\in {\sf Ind}(\Amc)$, if $a\in A^{\Jmc_{\Tmc,\Amc}}$, then $A(a)\in\abox$; 
\item for all $a,b\in {\sf Ind}(\Amc)$, if $r(a,b)\in r^{\Jmc_{\Tmc,\Amc}}$, then $r(a,b)\in\abox$;
\item for all $a\in {\sf Ind}(\Amc)$ and $C\in {\sf sub}(\Tmc)$, $a\cdot r\cdot C\not\in\Delta^{\interj_{\tbox,\abox}}$;
\item for all $d\in\Delta^{\interj_{\tbox,\abox}}\setminus\adom{\abox}$, 
    there is no $\el$ concept $D$ with $d\in D^{\interj_{\tbox,\abox}}$ 
    and $\sig{D}{}\cap\Sigma_\Csf\neq\emptyset$.
\end{enumerate}
For Item~(1), assume $a\in A^{\Jmc_{\Tmc,\Amc}}$. By Lemma~\ref{aboxcan}, $\abox\models (\tbox,\emptyset,A(x))(a)$,
and so we obtain $\abox\models (\tbox,\Sigma_\Csf,A(x))(a)$. By consistency of $\abox$ w.r.t.\ $(\tbox,\Sigma_\Csf)$, 
we then have $A(a)\in\abox$. Item~(2) follows directly from the construction of $\Jmc_{\Tmc,\Amc}$.
For Item~(3), assume for a proof by contradiction that there are $a\in {\sf Ind}(\Amc)$, $r\in \Sigma_{\Csf}$, and a concept $C$
such that $a\cdot r\cdot C\in\Delta^{\interj_{\tbox,\abox}}$. By Lemma~\ref{aboxcan}, we have $\abox\models
    (\tbox,\emptyset,q_{\existsr{r}{C}})(a)$. By Lemma~\ref{ac}, this implies
    that there is some $m\geq 0$ with $\tbox\models
    C_a^m\sqsubseteq\existsr{r}{C}$, where $C_{a}^{m}$ is the unfolding of $\Amc$ at $a$ of depth $m$.
    We show that this contradicts the assumption that $(\tbox,\Sigma_\Csf)$ is safe.
    There does not exist a tlc $\existsr{r}{C'}$ of $C_a^m$ with
    $\tbox\models C'\sqsubseteq C$ because otherwise there is some
    $b\in\adom{\abox}$ with $r(a,b)\in\abox$ and $\abox\models
    (\tbox,\emptyset,q_{C})(b)$ and thus, $a\cdot r\cdot C$ would have
    never been introduced by Rule (R3) in the construction of
    $\interj_{\tbox,\abox}$. Moreover, there is no
    $\EL$ concepts $E$ with ${\sf sig}(E)\subseteq \Sigma_\Csf$ and $\tbox\models C_a^m\sqsubseteq
    \existsr{r}{E}$ and $\tbox\models E\sqsubseteq C$ because
    otherwise there is a $b\in\adom{\abox}$ with $r(a,b)\in\abox$ and
    $\Imc_{\abox}\models E(b)$ since $\abox$ is
    consistent w.r.t.\ $(\tbox,\Sigma_\Csf)$. But then $\abox\models
    (\tbox,\emptyset,q_{C})(b)$ by the fact that $\interj_{\tbox,\abox}$
    is a model of $\tbox$ and $\abox$. Again, in this case, $a\cdot
    r\cdot C$ would have never been introduced by Rule (R3) in the
    construction of $\interj_{\tbox,\abox}$.
    Hence $C_a^m\sqsubseteq\existsr{r}{C}$ witnesses that $(\tbox,\Sigma_\Csf)$ is not safe.

    For Item~(4), assume for a proof by contradiction that there is a
    $d\in\Delta^{\interj_{\tbox,\abox}}\setminus\adom{\abox}$ and an
    $\el$ concept $D$ such that $d\in D^{\interj_{\tbox,\abox}}$ and
    $\sig{D}{}\cap\Sigma_\Csf\neq\emptyset$. By definition, $d=a\cdot
    r_0\cdot C_0\cdots r_n\cdot C_n$ for some $a\in\adom{\abox}$. Let
    $G=C_0\sqcap \exists r_1.\exists r_2\ldots\exists
    r_n.D$. Obviously, $\sig{G}{}\cap\Sigma_\Csf\neq\emptyset$ and
    $a\cdot r_0\cdot C_0\in G^{\interj_{\tbox,\abox}}$. By the
    latter and Lemma~\ref{aboxcan}, we have $\abox\models
    (\tbox,\emptyset,q_{\existsr{r_0}{G}})(a)$. By Lemma~\ref{ac}, this implies
    that there is some $m\geq 0$ with $\tbox\models
    C_a^m\sqsubseteq\existsr{r_0}{G}$. We show that it follows that $(\tbox,\Sigma_\Csf)$ is not safe.
    We have $\sig{G}{}\cap\Sigma_\Csf\neq\emptyset$ and, by Item~(3),
    $r_{0}\not\in \Sigma_\Csf$. To show that $(\tbox,\Sigma_\Csf)$ is not safe it remains to show that there
    is no tlc $\existsr{r_0}{C'}$ of $C_a^m$ with $\tbox\models
    C'\sqsubseteq G$. This is indeed the case because otherwise there
    is some $b\in\adom{\abox}$ with $r_0(a,b)\in\abox$ and
    $\abox\models (\tbox,\emptyset,q_{C_0})(b)$ and thus, $a\cdot r_0\cdot
    C_0$ would have never been introduced by Rule (R3) in the construction of $\interj_{\tbox,\abox}$. 
  \end{proof}
This finishes the proof of Part~(a) of Point~(2) of Theorem~\ref{eltheorem}.
Before we prove Part~(b) of Point~(2) we show the following 
observation of independent interest.

\begin{lem}\label{lem:independent}
  Let $(\Tmc,\Sigma_\Csf)$ be a safe $\el$ TBox with closed predicates. 
  Then there exists an $\EL$ TBox $\Tmc'$
  equivalent to $\Tmc$ such that for any $C \sqsubseteq D \in
  \Tmc'$, ${\sf sig}(D) \subseteq \Sigma_\Csf$ or ${\sf sig}(D) \cap
  \Sigma_\Csf=\emptyset$.
\end{lem}
\begin{proof}
We apply the following three rules exhaustively (and recursively) to $\Tmc$:
  \begin{itemize}
  \item replace any $C\sqsubseteq D_{1}\sqcap D_{2}$ by $C\sqsubseteq D_{1}$ and $C\sqsubseteq D_{2}$;
  \item replace any $C \sqsubseteq \exists r.D$ such that there exists a tlc $\exists r.C'$ of
    $C$ with $\Tmc\models C' \sqsubseteq D$ by $C'\sqsubseteq D$;
  \item replace any $C \sqsubseteq \exists r.D$ with $r\in \Sigma_\Csf$ and ${\sf
      sig}(D)\not\subseteq \Sigma_\Csf$ by $C \sqsubseteq \exists r.F$ and
    $F\sqsubseteq D$, where $F$ is an
    $\EL$ concepts with ${\sf sig}(F)\subseteq \Sigma_{\Csf}$ such that $\Tmc\models C \sqsubseteq \exists r.F$ and
    $\Tmc\models F \sqsubseteq D$. (Note that such a concept $F$ always exists by Condition~\ref{it:el_safeness_cond2}.)
  \end{itemize}
  It is straightforward to show that the resulting TBox $\Tmc'$ is as required.
\end{proof}

Recall that UCQ evaluation for $\EL$ TBoxes without closed predicates is in \ptime. 
Thus, the following lemma and Part~(a) directly imply Part~(b) of Theorem~\ref{eltheorem}.

\begin{lem}~\label{lem:point2partb}
  Let $(\Tmc,\Sigma_\Csf)$ be a safe $\el$ TBox with closed predicates. 
  Then consistency of ABoxes w.r.t.\ $(\Tmc,\Sigma_\Csf)$ is in \ptime.
\end{lem}
\begin{proof}
We may assume that $\Tmc$ is in the form of the claim above: for all $C \sqsubseteq D \in
  \Tmc$, ${\sf sig}(D) \subseteq \Sigma_\Csf$ or ${\sf sig}(D) \cap
  \Sigma_\Csf=\emptyset$. We show that the following conditions are
equivalent, for every ABox $\Amc$:
\begin{enumerate}
\item $\Amc$ is consistent w.r.t.~$(\Tmc,\Sigma_{\Csf})$;
\item for all $C\sqsubseteq F\in \Tmc$ with ${\sf sig}(F)\subseteq \Sigma_\Csf$ and all $a\in {\sf Ind}(\Amc)$, if $\Amc\models (\Tmc,\emptyset,q_{F})(a)$, then
$\Imc_{\Amc}\models F(a)$.
\end{enumerate}
The implication from Condition~(1) to Condition~(2) is obvious. Conversely, assume that Condition~(2) holds.
It suffices to show that the tree-shaped canonical model $\Jmc_{\Tmc,\Amc}$ respects closed predicates $\Sigma_{\Csf}$.
One can readily check that the proofs of Points~(2) and (4) of Lemma~\ref{lem:point2parta} do not use 
the condition that $\Amc$ is consistent
w.r.t.~$(\Tmc,\Sigma_{\Csf})$. Thus, it suffices to prove that Points~(1) and (3) of Lemma~\ref{lem:point2parta}
hold for $\Jmc_{\Tmc,\Amc}$. 
But they follow directly from Condition~(2) and the construction of $\Jmc_{\Tmc,\Amc}$.
The result now follows from the fact that Condition~(2) can be checked in polynomial time in the size of $\Amc$.
\end{proof}
This finishes the proof of Theorem~\ref{eltheorem}.

\section{Quantified Query Case: Deciding Tractability of \ptime Query Evaluation}
\label{sec:tboxdec}
We consider the meta problem to decide whether query evaluation w.r.t.~a TBox with closed predicates is tractable.
We show that the following problems are in \ptime:
\begin{enumerate}
\item decide whether UCQ evaluation w.r.t.~\dlliter TBoxes with closed predicates is FO-rewritable (equivalently, in \ptime); and
\item decide whether UCQ evaluation w.r.t.~$\EL$ TBoxes with closed predicates is in \ptime. 
\end{enumerate}
In both cases, we use the characterization via safeness given in the previous section and show that safeness can be decided in \ptime.
For \dlliter, the proof is actually straightforward: to check safeness of a \dlliter TBox with closed predicates $(\Tmc,\Sigma_\Csf)$ 
it suffices to consider all basic concepts $B_1,B_2$ and roles $r$ from ${\sf sig}(\Tmc)$ (of which there are only polynomially many) 
and make satisfiability checks for basic concepts w.r.t.~\dlliter TBoxes and entailment checks of \dllitecore CIs and RIs by \dlliter TBoxes
according to the definition of safeness. Both can be done in polynomial time~\cite{CDLLR07}.
\begin{thm}
It is in \ptime to decide whether a \dlliter TBox with closed predicates is safe.
\end{thm}
Such a straightforward argument does not work for \EL TBoxes
since Definition~\ref{def:elsafe} quantifies over all $\EL$ concepts $C$,
$D$, and $E$, of which there are infinitely many. In the following, we
show that, nevertheless, safeness of an \EL TBox with closed
predicates $(\Tmc,\Sigma_\Csf)$ can be decided in \PTime. The first step of the proof is
to convert $\Tmc$ into a \emph{reduced} $\EL$ TBox $\Tmc^*$ with the following properties:
\begin{enumerate}
      \renewcommand{\theenumi}{(red\arabic{enumi})}
      \renewcommand{\labelenumi}{\theenumi}

\item\label{tstarcond1} $\Tmc^*$ contains no CI of the form $C \sqsubseteq D_1 \sqcap D_2$;

\item\label{tstarcond2} if $C \sqsubseteq \exists r . D \in \Tmc^*$,
% and  $\mn{sig}(\exists r . D) \not\subseteq \Sigma$, 
then there is no tlc $\exists r.C'$ of $C$ with $\Tmc^* \models C' \sqsubseteq D$.

\end{enumerate}
\begin{lem}\label{lem:reducee}
For every $\el$ TBox, one can compute in polynomial time an equivalent reduced $\EL$ TBox.
\end{lem}
\begin{proof}
Assume that $\Tmc$ is an $\el$ TBox. Compute $\Tmc^*$ by applying the following two rules exhaustively to $\Tmc$:
\begin{itemize}
\item replace any CI $C \sqsubseteq D_1 \sqcap D_2$ with the CIs $C \sqsubseteq D_1$
and $C \sqsubseteq D_2$;
\item replace any CI $C \sqsubseteq \exists r .D$ for which there exists a tlc $\exists r.C'$ of $C$ with $\Tmc\models C'\sqsubseteq D$
by $C'\sqsubseteq D$.
\end{itemize}
It is straightforward to prove that $\Tmc^{\ast}$ is reduced, equivalent to $\Tmc$, and is constructed in polynomial time
(using the fact that checking $\Tmc\models C \sqsubseteq D$ is in \ptime \cite{BaBrLu-IJCAI-05}).
\end{proof}
We now formulate a stronger version of safeness. While
Definition~\ref{def:elsafe} quantifies over all CIs
$C \sqsubseteq \exists r . D$ that are \emph{entailed} by the TBox
\Tmc, the stronger version only considers CIs of this form that are
\emph{contained} in \Tmc. For deciding tractability based on safeness,
this is clearly a drastic improvement since only the concept $E$
from Definition~\ref{def:elsafe} remains universally quantified.
\begin{defi}\label{def:elstrongsafe}
  An $\EL$ TBox with closed predicates $(\Tmc,\Sigma_\Csf)$ is
  \emph{strongly safe} if there exists no $\EL$ CI $C \sqsubseteq \exists r.D \in \Tmc$ such 
  that one of the following holds:
\begin{enumerate}
  \renewcommand{\theenumi}{(st\arabic{enumi})}
  \renewcommand{\labelenumi}{\theenumi}
\item\label{it:rw_cond1} $r\not\in\Sigma_\Csf$ and there is some $\EL$ concept
  $E$ such that $\tbox\models D\sqsubseteq E$ and ${\sf sig}(E)
  \cap \Sigma_\Csf\not=\emptyset$;
\item \label{it:rw_cond2} $r\in \Sigma_\Csf$, ${\sf sig}(D)
  \not\subseteq\Sigma_\Csf$, and there is no $\EL$ concept $E$ with ${\sf sig}(E) \subseteq \Sigma_\Csf$ such that
  $\Tmc\models C \sqsubseteq \exists r.E$ and $\Tmc\models
  E \sqsubseteq D$. \hfill$\triangle$
\end{enumerate}

\end{defi}

\smallskip
\noindent
The crucial observation is that, for \EL TBoxes in reduced form, the
original notion of safeness can be replaced by strong safeness.
\begin{lem}\label{lem:safe_eq_tbox_star_local}
  If $\Tmc$ is a reduced $\el$ TBox and $\Sigma_{\Csf}$ a signature, then $(\Tmc,\Sigma_\Csf)$ is safe iff it is strongly safe.
\end{lem}
\begin{proof}
  Suppose that $\tbox$ satisfies Conditions~\ref{tstarcond1}
  and~\ref{tstarcond2} for reduced $\EL$ TBoxes.

  \noindent $(\Rightarrow)$ Suppose that $(\Tmc,\Sigma_\Csf)$ is not
  strongly safe, that is, there is some
  $C\sqsubseteq\existsr{r}{D}\in\tbox$ satisfying \ref{it:rw_cond1} or
  \ref{it:rw_cond2}.  If $C\sqsubseteq\existsr{r}{D}$ satisfies
  Condition~\ref{it:rw_cond1}, then $r\not\in\Sigma_\Csf$ and there is some
  concept $E$ such that $\tbox\models D\sqsubseteq E$ and ${\sf
  sig}(E) \cap \Sigma_\Csf\not=\emptyset$. We show that $(\tbox,\Sigma_{\Csf})$ is not safe because the CI $C\sqsubseteq\existsr{r}{(D\sqcap E)}$ 
  violates safeness:
  \begin{enumerate}

  \item $\Tmc \models C \sqsubseteq \exists r . (D \sqcap E)$ since 
    $C \sqsubseteq \exists r . D \in \Tmc$ and $\Tmc \models D \sqsubseteq E$.

  \item there is no tlc $\existsr{r}{C'}$ of $C$ with $\tbox\models
    C'\sqsubseteq D \sqcap E$; this follows from Condition~\ref{tstarcond2}.

  \item Condition~\ref{it:el_safeness_cond1} is satisfied because $r \notin \Sigma_\Csf$ and 
    $\mn{sig}(D \sqcap E) \cap \Sigma_\Csf \neq \emptyset$ since $\mn{sig}(E) \cap \Sigma_\Csf \neq \emptyset$.

  \end{enumerate}
  If $C\sqsubseteq\existsr{r}{D}$ satisfies
  Condition~\ref{it:rw_cond2}, then it follows directly that $(\tbox,\Sigma_{\Csf})$ is not safe because $C\sqsubseteq\existsr{r}{D}$ satisfies
  Condition~\ref{it:el_safeness_cond2}.

  \proofdirectionskip ($\Leftarrow$) Suppose that $(\tbox,\Sigma_{\Csf})$ is not safe. Take any $\el$ CI 
  $C\sqsubseteq\existsr{r}{D}$ violating safeness. In the following, we
  use the tree-shaped canonical model $\interj_{\tbox,C}$ defined above. For the sake of
  readability denote the individual name $p$ of $\interj_{\tbox,C}$ by $a_{p}$ (in particular, $a_{C}$
  denotes $C$). Note that Lemma~\ref{conceptcan} yields
  $a_{C}\in\extn{(\existsr{r}{D})}{\interj_{\tbox,C}}$ since $\Tmc\models C \sqsubseteq \exists r.D$. 
  Thus there is some $d\in\domainn{\interj_{\tbox,C}}$ such that
  $(a_{C},d)\in\extn{r}{\interj_{\tbox,C}}$ and
  $d\in\extn{D}{\interj_{\tbox,C}}$. By definition of
  $\interj_{\tbox,C}$, $d=a_{C\cdot r\cdot E}$ for some $\EL$ concept
  $E$. By Lemma~\ref{conceptcan} and
  $d\in\extn{D}{\interj_{\tbox,C}}$, we have $\tbox\models E\sqsubseteq D$.

  Let $\Amc_C = \Amc_0, \Amc_1,\dots$ be the ABoxes used in the
  construction of $\interj_{\tbox,C}$. By definition of $\Amc_C$ and
  Condition~\ref{tstarcond2}, we have $C\cdot r\cdot E\not\in\paths{C}$, that is,
  $d=a_{C\cdot r\cdot E}$ must have been generated by (R3).
  Consequently, there is an $i\in\natno$ such that
  $\existsr{r}{E}(a_{C})\in\abox_i$, and $\abox_{i+1}=\abox_i\cup\{
  r(a_{C},a_{C\cdot r\cdot E}),E(a_{C\cdot r\cdot E})\}$.
  Using Condition~\ref{tstarcond1} one can now easily prove that
  $\existsr{r}{E}(a_{C})$ can only have been added due to an application
  of (R2).
Thus, there is some $C'\sqsubseteq \existsr{r}{E}\in\tbox$ with $\Imc_{\abox_j}\models C'(a_{C})$. 
We obtain $a_{C}\in\extn{(C')}{\interj_{\tbox,C}}$ and this implies, by Lemma~\ref{conceptcan},
that $\tbox\models C\sqsubseteq C'$. 
% Moreover,
% Condition~\ref{tstarcond2} yields
% $\sig{}{\existsr{r}{E}}\not\subseteq\Sigma$. 
As $(\tbox,\Sigma_{\Csf})$ is not safe due to the CI
$C\sqsubseteq\existsr{r}{D}$, we obtain one of the following cases:
\begin{itemize}
\item $C \sqsubseteq \exists r . D$ satisfies
  Condition~\ref{it:el_safeness_cond1}.
  Then $r\not\in\Sigma_\Csf$ and $\sig{}{D}\cap\Sigma_\Csf\neq\emptyset$.  Since
  $\tbox\models E\sqsubseteq D$, we thus have that $C'\sqsubseteq
  \existsr{r}{E}\in\tbox$ satisfies Condition~\ref{it:rw_cond1}. We have shown that \Tmc is not strongly safe.

\item  $C \sqsubseteq \exists r . D$ satisfies
  Condition~\ref{it:el_safeness_cond2}.
  Then $r\in\Sigma_\Csf$, $\sig{}{D}\not\subseteq\Sigma_\Csf$, and there is no  $\EL$ concept $F$ with ${\sf sig}(F)\subseteq \Sigma_\Csf$ and
  $\tbox\models C\sqsubseteq\existsr{r}{F}$ and $\tbox\models F\sqsubseteq D$.
  We aim at showing that $C'\sqsubseteq \existsr{r}{E}\in\tbox$
  satisfies Condition~\ref{it:rw_cond2}.  We already know that $r \in
  \Sigma_\Csf$. From $\tbox\models C\sqsubseteq C' \sqsubseteq \exists
  r.E$ and $\tbox\models E\sqsubseteq D$, we obtain $\mn{sig}(E) \not\subseteq \Sigma_\Csf$ (otherwise set $F:=E$ above to derive a contradiction). 
  We also obtain that there is no $\EL$ concept $F$ with ${\sf sig}(F)\subseteq \Sigma_\Csf$ and with $\tbox\models
  C'\sqsubseteq\existsr{r}{F}$ and $\tbox\models F\sqsubseteq E$. Again it follows that $\Tmc$ is not strongly safe.
  \qedhere
%This, however, is a consequence of ($\dagger$) and the facts
%  that $\tbox\models C\sqsubseteq C'$ and $\tbox\models E\sqsubseteq
%  D$.
  \end{itemize}
\end{proof}
\begin{lem}\label{lem:99}
It is in \ptime to decide whether a reduced \EL TBox with closed predicates is strongly safe.
\end{lem}
\begin{proof}
Assume $(\Tmc,\Sigma_{\Csf})$ is a reduced $\EL$ TBox with closed predicates.
Assume $C\sqsubseteq \exists r.D\in \Tmc$ is given. It suffices to show that
Conditions~\ref{it:rw_cond1} and~\ref{it:rw_cond2} can be checked in polynomial time.
For Condition~\ref{it:rw_cond1}, it suffices to show that one can check in polynomial time whether there
exists an $\EL$ concept $E$ with $\Tmc\models D \sqsubseteq E$ and ${\sf sig}(E)\cap \Sigma_{\Csf}\not=\emptyset$.
We reduce this to a reachability problem in the directed graph induced by the finite canonical model $\Imc_{\Tmc,D}$. 
In detail, let $G=(V,R)$ be the directed graph
with $V=\Delta^{\Imc_{\Tmc,D}}$ and $R = \bigcup_{r\in \NR}r^{\Imc_{\Tmc,D}}$. Let $T= \bigcup_{A\in \Sigma_{\Csf}}A^{\Imc_{\Tmc,D}}\cup
\bigcup_{r\in \Sigma_{\Csf}}(\exists r)^{\Imc_{\Tmc,D}}$.
Using Lemma~\ref{lem:el_can_model}, it is readily checked that there exists an \el concept $E$
with $\Tmc\models D \sqsubseteq E$ and ${\sf sig}(E)\cap \Sigma_{\Csf}\not=\emptyset$ iff
there exists a path from $a_{D}$ to a node in $T$ in $G$. The latter reachability problem can be checked in
polynomial time.

For Condition~\ref{it:rw_cond2}, assume $r\in \Sigma_{\Csf}$ and let $\Amc$ denote the ABox corresponding to 
the $\Sigma$-reduct of $\Imc_{\Tmc,C}$. We show that the following conditions are equivalent:
\begin{enumerate}
\item there exists an $\EL$ concept $E$ such that ${\sf sig}(E)\subseteq \Sigma_{\Csf}$, $\Tmc \models C\sqsubseteq \exists r.E$
and $\Tmc\models E \sqsubseteq D$;
\item there exists $a\in {\sf Ind}(\Amc)$ such that $(a_{C},a)\in r^{\Imc_{\Tmc,C}}$ and $\Amc\models (\Tmc,\emptyset,q_{D})(a)$.
\end{enumerate}
For the proof of the implication from (1) to (2) take an \el concept $E$ satisfying (1). Then $a_{C} \in (\exists r.E)^{\Imc_{\Tmc,C}}$, 
by Lemma~\ref{lem:el_can_model}. Then there exists $a\in {\sf Ind}(\Amc)$ such that $(a_{C},a)\in r^{\Imc_{\Tmc,C}}$ and $a\in E^{\Imc_{\Tmc,C}}$. 
Hence $a\in E^{\Imc_{\Amc}}$ as ${\sf sig}(E)\subseteq \Sigma_{\Csf}$. There is an unfolding $C_{a}^{m}$ of $\Amc$ at $a$ of depth $m$ such that 
$\emptyset \models C_{a}^{m}\sqsubseteq E$.
From $\Tmc\models E \sqsubseteq D$ we obtain $\Tmc\models C_{a}^{m}\sqsubseteq D$. 
But then, by Lemma~\ref{ac}, $\Amc\models (\Tmc,\emptyset,q_{D})(a)$, as required.

Conversely, let $a\in {\sf Ind}(\Amc)$ such that $(a_{C},a)\in r^{\Imc_{\Tmc,C}}$ and $\Amc\models (\Tmc,\emptyset,q_{D})(a)$.
By Lemma~\ref{ac}, there exists an unfolding $C_{a}^{m}$ of $\Amc$ at $a$ such that $\Tmc \models C_{a}^{m}\sqsubseteq D$.
It is readily checked that $E=C_{d}^{m}$ is as required for Condition~(1).

Condition~(2) can be checked in \ptime since query evaluaton for OMQCs in $(\mathcal{EL},\emptyset,\text{dtCQ})$ is in \ptime
(in combined complexity).
\end{proof}

\begin{thm}
It is in \ptime to decide whether an $\mathcal{EL}$ TBox with closed predicates is safe.
\end{thm}
\begin{proof}
Assume $(\Tmc,\Sigma_{\Csf})$ is given. By Lemma~\ref{lem:reducee}, we can construct, in polynomial time,
a reduced $\mathcal{EL}$ TBox $\Tmc'$ equivalent to $\Tmc$. By Lemma~\ref{lem:99}, we can check in
\ptime whether $(\Tmc',\Sigma_{\Csf})$ is strongly safe. By Lemma~\ref{lem:safe_eq_tbox_star_local}, strong safeness is
equivalent to safeness for $(\Tmc',\Sigma_{\Csf})$.
\end{proof}

\section{Closing Concept Names in the Fixed Query Case and Surjective CSPs}
\label{sect:closingconcepts}

% \section{From Surjective CSP to
% (DL-Lite$_{\mn{core}}$,tree-UCQ)}\label{sec:scsp_dllite_core_tucq}

We now switch from the quantified query case to the fixed query case.
In this section, we consider OMQC languages that only admit closing
concept names while the case of closing role names is deferred to the
subsequent section. Regarding the former, our main aim is to establish
a close connection between UCQ evaluation for such OMQC languages and
generalized surjective constraint satisfaction problems (CSPs).  Let
BUtCQ denote the class of Boolean queries that can be obtained from a
union of tCQs by existentially quantifying the answer variable and let
BAQ denote the class of \emph{Boolean atomic queries} which take the
form $\exists x A(x)$, $A$ a concept name. We consider OMQC languages
between $(\dllitecore,\NC,\text{BUtCQ})$ and
$(\alchi,\NC,\text{BUtCQ})$ as well as between
$(\EL,\NC,\text{BUtCQ})$ and $(\alchi,\NC,\text{BUtCQ})$ and show that
for all these, a \ptime/\conp dichotomy is equivalent to a \ptime/{\sc
  NP} dichotomy for generalized surjective CSPs, a problem that is
wide open.  In fact, understanding the complexity of surjective CSPs,
generalized or not, is a very difficult, ongoing research effort. As
pointed out in the introduction, there are even concrete surjective
CSPs with very few elements whose complexity is unknown and, via the
connection established in this section, these problems can be used to
derive concrete OMQCs from the mentioned languages whose computational
properties are currently not understood.
%It also follows other consequence is that full complexity classifications of the
% relevant OMQC classes as well as dichotomies are very difficult to
% attain.

We next introduce CSPs and then give a more detailed overview of the
results obtained in this section. An interpretation \Imc is a
\emph{$\Sigma$-interpretation} if it only interprets predicates in
$\Sigma$, that is, all other predicates are interpreted as empty. For
every finite $\Sigma$-interpretation \Imc we denote by
$\mn{CSP}(\Imc)$ the following \emph{constraint satisfaction problem
  (in signature $\Sigma$)}: given a finite
$\Sigma$-interpretation~$\Jmc$, decide whether there is a homomorphism
$h$ from $\Jmc$ to $\inter$. The \emph{surjective constraint
  satisfaction problem}, $\mn{CSP}(\Imc)^{\mn{sur}}$, is the variant
of $\mn{CSP}(\Imc)$ where we require $h$ to be surjective. $\Imc$ is
then called the \emph{template} of $\mn{CSP}(\Imc)^{\mn{sur}}$. In
this article we only consider CSPs with predicates of arity at most two.
%
% A \emph{BUtCQ} is a union of queries $\exists x\, q(x)$ with $q$ a tCQ
% and a \emph{BUdtCQ} is defined likewise based on dtCQs. 
% In the remainder of this section, we first show that for every
% $\mn{CSP}(\Imc)^{\mn{sur}}$, there is an OMQC $Q$ from
% $(\dllitecore,\NC,\text{BUtCQ})$ such that the evaluation problem for
% $Q$ has the same complexity as the complement of
% $\mn{CSP}(\Imc)^{\mn{sur}}$, up to polynomial time reductions. We then
% observe that the same holds for $(\mathcal{EL},\NC,\text{BAQ})$.
% For the converse direction, we construct, for every OMQC in $(\mathcal{ALCHI},\NC,\text{BUtCQ})$,
% a \emph{generalized} surjective constraint satisfaction problem of the
% same complexity, up to polynomial reductions.
 A \emph{generalized
  surjective CSP in signature $\Sigma$} is characterized by a
\emph{finite set} $\Gamma$ of finite $\Sigma$-interpretations instead
of a single such interpretation, denoted
$\mn{CSP}(\Gamma)^{\mn{sur}}$. The problem is to decide, given a
$\Sigma$-interpretation $\Jmc$, whether there is a surjective
homomorphism from $\Jmc$ to some interpretation in~$\Gamma$. The
interpretations $\Imc$ in $\Gamma$ are called the \emph{templates} of
$\mn{CSP}(\Gamma)^{\mn{sur}}$. % Note that, in the non-surjective case,
% every generalized CSP can be translated into an equivalent
% non-generalized CSP \cite{DBLP:journals/ejc/FoniokNT08}. In the
% surjective case, such a translation is not known.

We first show that for every constraint satisfaction problem $\mn{CSP}(\Gamma)^{\mn{sur}}$, there is an
OMQC $Q$ from $(\dllitecore,\NC,\text{BUtCQ})$ such that the
evaluation problem for $Q$ has the same complexity as the complement
of $\mn{CSP}(\Gamma)^{\mn{sur}}$, up to polynomial time reductions; we
then observe that the same holds for $(\mathcal{EL},\NC,\text{BAQ})$.
To achieve a cleaner presentation, we first present the construction
for non-generalized surjective CSPs and then sketch the modifications
required to lift it it to generalized surjective CSPs.
Consider $\mn{CSP}(\Imc)^{\mn{sur}}$ in signature $\Sigma$. Let $A$,
$V$, and $V_{d}$, $d\in \Delta^{\Imc}$, be concept names not in
$\Sigma$, and ${\sf val}$, ${\sf aux}_{d}$, and ${\sf force}_{d}$,
$d\in \Delta^{\Imc}$, be role names not in $\Sigma$.  Define the OMQC
$Q_{\Imc}=(\Tmc,\Sigma_\Asf,\Sigma_\Csf,q)$ from
$(\dllitecore,\NC,\text{BUtCQ})$ as follows:
%
% todo[inline]{For uniformity with the other direction, we might want
% to admit unary predicates as well; also: Is it clear that this can
% be assumed w.l.o.g.\ also for surjective CSPs?}  \todo[color=green,
% inline]{We need to assume something about $\inter$, e.g., it
% contains at least two elements, so that we have an interpretation
% that we can not homomorphically embed in $\inter$. Otherwise, the
% second lemma in this section is not a polynomial-time many-one
% reduction but a Turing one} Define the DL-Lite$_{\mn{core}}$ TBox
% $\Tmc$ with closed predicates $\Sigma_{\mn{closed}}$ and Boolean
% tree-shaped UCQ $q$ as follows.
%
% Same setup as above, but now we assume $\mn{CSP}(\Imc)$ to be a
% surjective CSP, that is, the involved homomorphisms $h$ have to hit
% all elements of the template \Imc.
%
% This can easily be enforced by two modifications. First, add to the
% TBox \Tmc the following CIs, for each $d \in \Delta^\Imc$ and $r \in
% \Sigma_{\mn{csp}}$:
% %
% $$
% \begin{array}{rclcrcl}
%     \exists r &\sqsubseteq& \exists \mn{force}_d & \qquad & 
%     \exists^- r &\sqsubseteq& \exists \mn{force}_d \\[1mm]
%     \exists \mn{force}^-_d & \sqsubseteq& 
% \end{array}
% $$
% %
% They select an element that must map to $d$. This is then guaranteed
% by disjunctively adding the following to $q$:
% %
% $$
% \bigvee_{d,d' \in \Delta^\Imc \mid d \neq d'} \exists x \exists y
% \exists z \, \mn{force}_d(z,x) \wedge \mn{val}(x,y) \wedge
% V_{d'}(y).
% $$
% %
% In the proof of the claim, the homomorphism defined in the `if'

$$
\begin{array}{rcl}
  \Tmc &= &  \{  A \sqsubseteq \exists \mn{val}, \ \exists \mn{val}^- \sqsubseteq V \} \, \cup\\[1mm]
  && \{  A \sqsubseteq \exists \mn{aux}_d, \ \exists \mn{aux}_d^- \sqsubseteq V \sqcap V_d \mid  d \in \Delta^\Imc \} \, \cup \\[1mm]
  && \{ A \sqsubseteq \exists \mn{force}_d, \
  \exists \mn{force}^-_d  \sqsubseteq A\mid d\in\domain \} \\[2mm]
  % \Tmc &= & \{ C \sqsubseteq \exists \mn{val} \mid C \in \Gamma \}
  % \cup \{ \exists \mn{val}^- \sqsubseteq V \} \, \cup\\[1mm]
  % && \{ C \sqsubseteq \exists \mn{aux}_d, \ \exists \mn{aux}_d^-
  % \sqsubseteq V \sqcap V_d \mid C \in \Gamma, d \in \Delta^\Imc \}
  % \\[3mm]
  % && \{ \exists r \sqsubseteq \exists \mn{val}, \exists r^-
  % \sqsubseteq \exists \mn{val} \mid r \in \Sigma_{\mn{csp}} \text{
  % binary} \} \, \cup\\[1mm]
  % &&\exists e \sqsubseteq \exists \mn{aux1}, \\[1mm] \exists
  % \mn{aux}_1^- &\sqsubseteq& C \sqcap R, \\[1mm]
  % &&\exists e \sqsubseteq \exists \mn{aux2}, \\[1mm] \exists
  % \mn{aux}_2^- &\sqsubseteq& C \sqcap G, \\[1mm]
  % &&\exists e \sqsubseteq \exists \mn{aux3}, \\[1mm] \exists
  % \mn{aux}_3^- &\sqsubseteq& C \sqcap B \quad \} \\[3mm]
%     
  \Sigma_{\Csf} &=&
  \{ A,V \} \cup \{ V_d \mid d \in \Delta^\Imc \} \\[2mm]
  \Sigma_{\Asf} &=& \Sigma \cup \Sigma_\Csf \\[2mm]
  q &=& q_1\lor q_2\lor q_3 \lor q_{4}
\end{array}
$$
where %$q_1$, $q_2$, and $q_3$ are defined in Figure~\ref{fig:scsp_dllite_core_tucq}.
$$
\begin{array}{r@{\,}c@{\,}l}
  q_1  &=& \,\,\,\,\displaystyle\bigvee_{d,d' \in \Delta^\Imc \mid d \neq d'}
  \exists x \exists y_1 \exists y_2 \, A(x) \wedge \mn{val}(x,y_1) \;
  \wedge  \\[-3mm]
  && \hspace*{3.2cm}  \mn{val}(x,y_2) \wedge V_d(y_1) \wedge V_{d'}(y_2) \\[1mm]
   q_2  &=&  
   \displaystyle\bigvee_{d \in \Delta^\Imc, E \in
   	\Sigma \mid d\not\in E^\Imc} \!\!\!\!\!\!\!\!\!\!\!
   \exists x\exists
   y \, A(x) \wedge E(x) \; \wedge
   \\[-3mm]
   && \hspace*{3.2cm} \mn{val}(x,y) \wedge V_d(y) \\[1mm]
  q_3  &=&  \!\!\!\!\!\displaystyle\bigvee_{d,d' \in \Delta^\Imc, r \in
    \Sigma \mid (d,d') \notin r^\Imc} \!\!\!\!\!\!\!\!\!\!\!
  \exists x\exists
  y\exists x_1  \exists y_1 \,  A(x) \wedge A(y) \wedge r(x,y) \;
  \wedge \\[-3mm]
  && \hspace*{3.2cm} \mn{val}(x,x_1) \wedge \mn{val}(y,y_1) \; \wedge
  \\[1.5mm]
  && \hspace*{3.2cm}V_d(x_1) \wedge V_{d'}(y_1) \\[1mm]
  q_4 &=&  \,\,\,\,\displaystyle \bigvee_{d,d' \in \Delta^\Imc \mid d \neq d'}
  \exists x \exists y \exists z \, A(x) \wedge \mn{force}_d(z,x) \;
  \wedge \\[-3mm]
  && \hspace*{3.2cm} \mn{val}(x,y) \wedge V_{d'}(y).
\end{array}
$$
The following lemma links $\mn{CSP}(\Imc)^{\mn{sur}}$ to the
constructed OMQC $Q_{\Imc}$.
\begin{lem}\label{lem:comp-csp-reduc-to-omqc}
The complement of $\mn{CSP}(\Imc)^{\mn{sur}}$ and the evaluation problem for $Q_{\Imc}$ 
are polynomially reducible to each other. 
\end{lem}
\begin{proof}
  Assume that $\mn{CSP}(\Imc)^{\mn{sur}}$ is given.
  For the polynomial reduction of $\mn{CSP}(\Imc)^{\mn{sur}}$ to the evaluation problem for $Q_{\Imc}$, 
  let $\interj$ be a $\Sigma$-interpretation that is an input of
  $\mn{CSP}(\Imc)^{\mn{sur}}$. Let $\Amc_\Jmc$ be the ABox corresponding to \Jmc. Introduce, for every $d\in \Delta^{\Imc}$, a fresh individual name $a_{d}$ and let
  the ABox \Amc be defined as
  \[\abox_\interj\cup\{A(a_d)\mid
  d\in\Delta^\Jmc\}\cup\{V(a_d),V_d(a_d)\mid d\in\domain\}.\]
  Obviously, $\abox$ can be constructed in polynomial time. We claim
  that $\interj\in\mn{CSP}(\Imc)^{\mn{sur}}$ iff $\abox\not\models
  Q_{\Imc}$.

  \smallskip
  
  $(\Rightarrow)$ Suppose that there is a surjective homomorphism $h$
  from \Jmc to \Imc. Define the interpretation $\inter'$ as follows:
  \begin{eqnarray*}
    \Delta^{\inter'} & = & \mn{Ind}(\abox)\\
    A^{\inter'} & = & \mn{Ind}(\abox_\interj)\\
    V^{\inter'} & = & \domain\\
    V_d^{\inter'} & = & \{a_d\},\text{ for all }d\in\domain\\
    \mn{val}^{\inter'}  & = & \{(a,a_{h(a)})\mid a\in\mn{Ind}(\abox_\interj)\}\\ 
    \mn{aux}_d^{\inter'} & = & \{(a,a_d)\mid a\in\mn{Ind}(\abox_\interj)\}\text{, for all }d\in\domain \\
    \mn{force}_d^{\inter'} & = & \{(a,a')\in\mn{Ind}(\abox_\interj)\times\mn{Ind}(\abox_\interj)\mid h(a')=d\}\text{, for all }d\in\domain\\
    P^{\inter'} & = & \extj{P}\text{, for all predicates }P \not\in(\{A,V,\mn{val}\}\cup\{V_d,\mn{aux}_d,\mn{force}_d\mid d\in\domain\})
  \end{eqnarray*}
  One can now verify that $\inter'$ is a model of \Tmc and \Amc that
  respects closed predicates $\Sigma_{\Csf}$, and that
  $\inter'\not\models q$. Thus, $\Amc\not\models Q_{\Imc}$, as required.

  \smallskip
  
  $(\Leftarrow)$ Suppose $\abox\not\models Q_{\Imc}$. Then there is a model 
  $\inter'$ of $\Tmc$ and \Amc that respects closed predicates 
  $\Sigma_\Csf$ and such that $\inter'\not\models q$. Define 
  $h=\{(d,a_e)\in\mn{val}^{\inter'}\mid d\in \Delta^\Jmc \}$. We show 
  that $h$ is a surjective homomorphism from $\Jmc$ to $\inter$.

  We first show that the relation $h$ is a function. Assume that this
  is not the case, that is, there are $d\in\Delta^\Jmc$ and
  $e_1,e_2\in\domain$ such that $e_1\neq e_2$ and $(d,a_{e_i}) \in
  \mn{val}^{\Imc'}$ for $i \in \{1,2\}$. Note that $a_{e_i} \in
  V_{e_i}^{\Imc'}$. Thus we get $\inter'\models q_1$, which is a
  contradiction against our choice of $\Imc'$. 

  To show that $h$ is
  total, take some $d \in \Delta^\Jmc$. Then $d \in A^{\Imc'}$ and
  thus the first line of \Tmc yields an $f \in V^\Jmc$ with $(d,f) \in
  \mn{val}^{\Imc'}$. Since $V$ is closed, we must have $f=a_e$ for
  some $e$, and thus $h(a_e)=f$.

  We show that $h$ is a homomorphism. We show, using $q_{3}$, that $h$ preserves role names. Using $q_{2}$, one can show in 
  the same way that $h$ preserves concept names. Assume for a
  contradiction that there is $(d,e) \in r^\Jmc$ with
  $(h(d),h(e))\not\in\ext{r}$. The latter implies that 
  the following is a disjunct of $q_3$:
  $$
  \begin{array}{l}
    \exists
    x\exists y\exists x_1 \exists y_1\, A(x) \wedge A(y) \wedge r(x,y)
    \wedge \mn{val}(x,x_1) \; \wedge \\[1mm]
    \hspace*{2cm} \mn{val}(y,y_1) \wedge V_{h(d)}(x_1) \wedge V_{h(e)}(y_1).
  \end{array}
  $$
  Note that $d,e\in A^{\inter'}$, $(d,a_{h(a)}),
  (e,a_{h(e)})\in\mn{val}^{\inter'}$, $a_{h(d)}\in
  V^{\inter'}_{h(d)}$, and $a_{h(e)}\in V^{\inter'}_{h(e)}$. Thus $\inter'\models q_3$, which contradicts our choice of~$\Imc'$.

  It remains to show that $h$ is surjective. Fix a $d\in\domain$. We
  have to show that there is an $e \in \Delta^\Jmc$ with
  $h(e)=d$. Take some $f\in\Delta^\Jmc$. Then by the third line of
  \Tmc and since $A$ is closed, there is some $e\in\Delta^\Jmc$ such
  that $(f,e)\in\mn{force}_d^{\inter'}$. We show that $e$ is as
  required.  Assume to the contrary that $h(e) \neq d$. Then the
  following is a disjunct of $q_4$:
  $$
  A(x) \wedge \mn{force}_d(z,x) \wedge \mn{val}(x,y) \wedge
  V_{h(e)}(y).
  $$
  Note that $f \in A^{\Imc'}$, $(e,a_{h(e)}) \in \mn{val}^{\Imc'}$,
  and $a_{h(e)} \in V_{h(e)}^{\Imc'}$. Thus, $\Imc'\models q_{4}$ which contradicts our choice of
  $\Imc'$. This finishes the proof of the reduction from $\mn{CSP}(\Imc)^{\mn{sur}}$ to evaluating $Q_{\Imc}$.
  
\bigskip

  We now give the polynomial reduction of the evaluation problem for $Q_{\Imc}$ to $\mn{CSP}(\Imc)^{\mn{sur}}$. 
  Assume a $\Sigma_\Asf$-ABox $\Amc$ is given. To decide whether $\Amc\models Q_{\Imc}$, we start with the following:
  \begin{enumerate}

  \item If $\Amc$ does not contain any assertion of the form $A(a)$, then $\Amc \not\models Q_{\Imc}$. In fact, let
    $\Imc_\Amc$ be \Amc viewed as an interpretation.
    Then $\Imc_\Amc$ is a model of \Amc that respects closed predicates $\Sigma_\Csf$. Since $\Amc$ does not contain any assertion of the form $A(a)$, $\Imc_\Amc$ is also a model of \Tmc
    and satisfies $\inter_\abox\not\models q$ (note that each disjunct
    of $q$ demands the existence of an instance of $A$). Thus answer `$\Amc\not\models Q_{\Imc}$'.

  \item Otherwise, if $\Amc$ does not contain for each $d \in 
	\Delta^\Imc$ an individual name $a$ with $V(a),V_d(a)\in \Amc$, 
	then $\Amc$ is not consistent w.r.t.\ $(\Tmc,\Sigma_\Csf)$. Thus 
	answer `$\Amc\models Q_{\Imc}$'.

      \item Otherwise, if \Amc contains an individual name $a$ with
        $V(a)\in \Amc$ and $V_d(a)\not\in \Amc$ for each $d\in
        \Delta^{\Imc}$, then $\Amc \not\models Q_{\Imc}$. In fact,
        we can build a model of \Amc and \Tmc that makes $q$
        false in the following way:
        Line~1 of \Tmc can be satisfied by linking every element to
        $a$ via $\mn{val}$; Line~2 can be satisfied since Case~(2)
        above does not apply; Line~3 can trivially be satisfied. All
        remaining choices can be taken in an arbitrary way.

  \end{enumerate}
  If none of the above applies, let $\Amc_{|A}$ be the restriction
  of \Amc to $\{ a\in {\sf Ind}(\Amc) \mid A(a) \in \Amc\}$. Since Case~(1)
  above does not apply, $\Amc_{|A}$ is non-empty. Let $\Jmc_{A}$ be the $\Sigma$-reduct of the 
  interpretation corresponding to $\Amc_{|A}$.
  We show that $\Jmc_{A}\in \mn{CSP}(\Imc)^{\mn{sur}}$ iff $\abox\not\models Q_{\Imc}$.

  $(\Leftarrow)$. Assume that
  $\abox\not\models Q_{\Imc}$. Then there is a model
  $\Jmc$ of \Tmc and \Amc that respects closed predicates
  $\Sigma_{\Csf}$ and such that $\Jmc \not\models q$. By the first
  line of $\Tmc$, since $V$ is closed, Case~(3) does not apply, and by $q_{1}$, for each 
  $a \in \mn{Ind}(\Amc_{|A})$ there is
  exactly one $d \in \Delta^\Imc$ such that $a \in (\exists \mn{val}. V_d)^{\Jmc}$. 
  Define a homomorphism $h:\Jmc_{A} \rightarrow \Imc$ by
  mapping each $a$ in $\Amc_{|A}$ to the value $d \in \Delta^\Imc$ thus
  determined. By $q_{2}$ and $q_{3}$, $h$ is indeed a
  homomorphism. By the third line of \Tmc and $q_{4}$
  and since $A$ is closed, $h$ must be surjective.

  $(\Rightarrow)$. Assume that $\Jmc_{A} \in
  \mn{CSP}(\Imc)^{\mn{sur}}$, and let $h$ be a surjective homomorphism
  from $\Jmc_{A}$ to \Imc. Build an interpretation $\Jmc$ as
  follows. Start by setting $\Jmc =\Imc_\Amc$. Since Case~(2) above does
  not apply, for each $d \in \Delta^\Imc$ we can select an individual name
  $a_d$ of $\Amc$ such that $V(a_{d})$ and $V_d(a_{d})$ are in \Amc. For each
  individual name $a$ in $\Amc_{|A}$, extend $\Jmc$ by adding $(a,a_{h(a)})$ to
  $\mn{val}^{\Jmc}$ and $(a,a_d)$ to $\mn{aux}_d^{\Jmc}$ for each $d \in
  \Delta^\Imc$. Since $h$ is surjective, for each $d \in \Delta^\Imc$
  there must be an individual name $a_d'$ of $\Amc_{|A}$ with $h(a_d')=d$. Further
  extend $\Jmc$ by adding $(a,a_d')$ to $\mn{force}_d^{\Jmc}$ for all $a
  \in \mn{Ind}(\Amc_{|A})$ and all $d \in \Delta^\Imc$.  It is readily
  checked that $\Jmc$ is a model of $\Tmc$ and $\Amc$ that respects
  closed predicates $\Sigma_{\Csf}$, and that $\Jmc \not\models q$. Thus, $\Amc\not\models Q_{\Imc}$, as required.
\end{proof}
%
%
%We say that two decision problems $P_1$ and $P_2$ are
%\emph{polynomially equivalent} if $P_1$ polynomially reduces to $P_2$
%and vice versa.
%% , and identify Boolean OMQCs with the decision problem of answering
%% them. (DEFINED ABOVE)
%%
%\begin{lem}\label{lem:scsp_to_dllitecore_ucq}
%  The complement of $\mn{CSP}(\Imc)^{\mn{sur}}$ is polynomially
%  equivalent to evaluating $Q$.
%\end{lem}
%%
Note that the same reduction works when \dllitecore is
replaced with \EL. One simply has to replace the TBox $\Tmc$ by the $\EL$ TBox 
$$
\begin{array}{rcl}
	\Tmc' &= &  \{  A \sqsubseteq \exists \mn{val}.V\} \, \cup\\[1mm]
	&& \{  A \sqsubseteq \exists \mn{aux}_d . (V \sqcap V_d) \mid  d \in \Delta^\Imc \} \, \cup \\[1mm]
	&& \{ A \sqsubseteq \exists \mn{force}_d . A \mid d\in\domain \} 
\end{array}	
$$
and observe that all CQs in $q$ have the form $\exists x q'(x)$
with $q'(x)$ a dtCQ which enables the following modification:
introduce a fresh concept name $B$, then for each CQ $\exists x q'(x)$
in $q$, take the \EL concepts $C_{q'}$ that corresponds to $q'(x)$ and
extend $\Tmc'$ with $C_{q'} \sqsubseteq B$, and finally replace $q$
with the BAQ $\exists x \, B(x)$. 

We now describe how to extend the reduction from surjective CSPs to
generalized surjective CSPs. Let $\mn{CSP}(\Gamma)^\mn{sur}$ be such a
CSP. Let $\Gamma = \{ \Imc_1,\dots,\Imc_n \}$. The main idea is to use
$n$ copies of each non-$\Sigma$ symbol in the above reduction, one for
each template in $\Gamma$.  Let the $i$-th copy of $A$ be $A_i$, of
$\mn{val}$ be $\mn{val}_i$, and so on. This gives us $n$ copies of the
TBox \Tmc and the UCQ $q$ in the above reduction, which we call
$\Tmc_1,\dots,\Tmc_n$ and $q_1,\dots,q_n$. Note that the $\Tmc_i$ do
not share any symbols and that the $q_i$ share only the symbols from
$\Sigma$. We define $Q_\Gamma=(\Tmc,\Sigma_\Asf,\Sigma_\Csf,q)$ where
$\Tmc = \Tmc_1 \cup \cdots \cup \Tmc_n$, $q$ is the BUtCQ obtained
from $q_1 \wedge \cdots \wedge q_n$ by pulling disjunction outside,
and $\Sigma_\Asf$ and $\Sigma_\Csf$ are defined as expected. It is
then possible to prove an analogue of
Lemma~\ref{lem:comp-csp-reduc-to-omqc}, we only sketch the required
modifications. In the reduction of $\mn{CSP}(\Gamma)^{\mn{sur}}$ to
the evaluation problem for $Q_{\Gamma}$, one builds on ABox \Amc for
each $\Imc \in \Gamma$, each as in the corresponding of the proof of
Lemma~\ref{lem:comp-csp-reduc-to-omqc}, and then takes their union.
In the reduction of the evaluation problem for $Q_{\Gamma}$ to
$\mn{CSP}(\Gamma)^{\mn{sur}}$, one first checks whether for some $i$,
the given $\Sigma_\Asf$-ABox \Amc contains an assertion $A_i(a)$, but
no assertion $V_i(a)$ and answers `$\Amc \models Q_\Gamma$' if this is
the case (this corresponds to Point~(2) in the original proof). One
then checks whether for some $i$ there is no assertion of the form
$A_i(a)$ and answers `$\Amc \not\models Q_\Gamma$' if this is
the case (corresponding to Point~(1) in the original proof). Point~(3)
and the remainder of the reduction need no major adaptations.

In summary, we have obtained the following result.
\begin{thm}
  \label{thn:cspfirstdir}
  For every $\mn{CSP}(\Gamma)^\mn{sur}$, % in binary signature,
  there is an OMQC $Q_{\Gamma}$ in  
  $(\dllitecore,\NC,\text{BUtCQ})$ such that the 
  complement of $\mn{CSP}(\Gamma)^\mn{sur}$ has the same complexity as 
  the evaluation problem for $Q_{\Gamma}$, up to polynomial time reductions. The same 
  holds for $(\EL,\NC,\text{BAQ})$.
\end{thm}
We note that, as can easily be verified by checking the
constructions in the proof of Lemma~\ref{lem:comp-csp-reduc-to-omqc},
the complement of $\mn{CSP}(\Gamma)^\mn{sur}$ and the evaluation problem 
for $Q_{\Gamma}$ actually have the same complexity up to \emph{FO reductions} 
\cite{Immerman}.  This links the complexity of the two problems even 
closer. For example, if one is complete for {\sc LogSpace} or in {\sc 
AC}$^0$, then so is the other.

\smallskip

We now establish a rather general converse of Theorem~\ref{thn:cspfirstdir} by
showing that for every OMQC $Q$ from
$(\alchi,\NC,\text{BUtCQ})$, there is a generalized
surjective CSP % (over binary signature)
that has the same complexity as the complement of the evaluation problem for $Q$, up to polynomial time
reductions. 

% \footnote{{\color{blue}Alternatively, we could go via an
% intermediate step: reduce to homomorphism problems where some unary
% predicates must be surjective, and then reduce that to surjective
% homomorphism problems. The first translation then preserves
% expressive power. We could do the same with fixed role names.  In
% this way, we shift much of the analysis to the CSP-world.}}
%
Let $Q=(\Tmc,\Sigma_\Asf,\Sigma_\Csf,q)$ be an OMQC from
$(\alchi,\NC,\text{BUtCQ})$.  We can assume w.l.o.g.\ that $q$ is a
BAQ, essentially because every tCQ can be rewritten into an \ALCI
concept; see the remark on \EL and BAQs made after the proof of
Lemma~\ref{lem:comp-csp-reduc-to-omqc}. Thus, let $q=\exists x \, A_{0}(x)$
with $A_{0}$ a concept name in $\Tmc$.
% has the form $\exists x \,
% A_0(x)$ with $A_0$ occurring in \Tmc. The argument as to why this assumption is w.l.o.g. is given above already (for dtCQs and $\mathcal{EL}$):
% every disjunct of any BUtCQ takes the form $\exists x\; q'(x)$ with $q'(x)$ a tCQ and thus one can add all $C_{q'} \sqsubseteq A_0$ 
% to the TBox with $A_0$ a fresh concept name and replace $q$ with $\exists x \, A_0(x)$.
%
We use the notation for types introduced in
Section~\ref{sec:dichotomytbox}. A subset $T$ of the set
${\sf TP}(\Tmc)$ of $\Tmc$-types is \emph{realizable in a
  countermodel of $Q$} if there is a $\Sigma_\Asf$-ABox \Amc and model \Imc
of $\Tmc$ and $\Amc$ that respects closed predicates $\Sigma_\Csf$
such that $\Imc \not\models q$ and
$T = \{ \mn{tp}_\Imc(a) \mid a \in \mn{Ind}(\Amc) \}$.
The desired surjective generalized CSP is defined by taking one template for each 
$T \subseteq \mn{TP}(\Tmc)$ that is realizable in a countermodel of $Q$. 
The signature $\Sigma$ of the CSP comprises the predicates
in $\Sigma_\Asf$ and one concept name $\overline{A}$ for each concept name in $\Sigma_\Csf$. We assume
w.l.o.g.\ that there is at least one concept name in $\Sigma_{\Csf}$
and at least one concept name $A_{\mn{open}} \in \Sigma_{\Asf}
\setminus \Sigma_{\Csf}$.
% \footnote{Adding fresh concept names to $\Sigma_\Csf$ or to
% $\Sigma_\Asf \setminus \Sigma_\Csf$ does clearly not change the
% complexity of OMQCs.}

\smallskip

Pick for every $A\in \Sigma_{\Csf}$ an element $d_{A}$. 
Then for each $T \subseteq \mn{TP}(\Tmc)$ realizable in a countermodel of $Q$ 
we define the template $\Imc_T$ as follows:
$$
\begin{array}{rcl}
  \Delta^{\Imc_T} &=& T \uplus \{ d_A \mid A \in \Sigma_{\Csf} \}\\[1mm]
  A^{\Imc_T} &=& \{ t \in T \mid A \in t \} \cup \{ d_B \mid B \in \Sigma_{\Csf}\setminus \{A\} \} \\[1mm]
  \overline{A}^{\Imc_T} &=& \{ t \in T \mid A \notin t \} \cup \{ d_B
  \mid B \in \Sigma_{\Csf}\setminus\{A\} \} \\[1mm]
  r^{\Imc_T} &=& \{ (t,t') \in T \times T\mid t \rightsquigarrow_r  t'
  \} \, \cup \\[1mm]
  && \{ (d,d') \in \Delta^{\Imc_T} \times \Delta^{\Imc_T} \mid
  \{ d,d'\} \setminus T \neq \emptyset \} .
\end{array}
$$
Note that, in $\Imc_T$ restricted to domain $T$, $\overline{A}$ is
interpreted as the complement of $A$. At each element $d_A$, all
concept names except $A$ and $\overline{A}$ are true, and these
elements are connected to all elements with all roles. Intuitively, we
need the concept names $\overline{A}$ to ensure that when an assertion
$A(a)$ is missing in an ABox \Amc with $A$ closed, then $a$ can only
be mapped to a template element that does not make $A$ true; this is
done by extending \Amc with $\overline{A}(a)$ and exploiting that
$\overline{A}$ is essentially the complement of $A$ in each
$\Imc_T$. The elements $d_A$ are then needed to deal with inputs to
the CSP where some point satisfies neither $A$ nor $\overline{A}$.
%
% , and that
% %
% \begin{equation}
% \tag{$*$}
% t \in C^{\Imc_T} \text{ iff } C \in t \text{ for all } C \in
% \mn{sub}(\Tmc,q)
% \text{ and } t \in T.
% \end{equation}
%
Let $\Gamma_{Q}$ be the set of all interpretations $\Imc_T$ obtained in
the described way.
%
% closed/surjective: wenn man in der TBox $A \sqsubseteq \exists r
% . B$ hat mit $B$ closed, dann muss jede ABox, die nen $A$-Punkt
% enth\"alt, auch nen $B$-Punkt enthalten. Also geht BELIEBIGE
% Teilmenge $\Gamma$ nicht.
%
% \bigskip SATISFIABILITY / CONSISTENCY
%
\begin{lem}
  \label{lem:first}
  Let $Q=(\Tmc,\Sigma_\Asf,\Sigma_\Csf,q)$ be an OMQC from
  $(\alchi,\NC,\text{BUtCQ})$. Then the evaluation problem for $Q$ reduces in polynomial time to 
  the complement of $\mn{CSP}(\Gamma_{Q})^{\mn{sur}}$.
\end{lem}
\begin{proof}
  Let \Amc be a $\Sigma_{\Asf}$-ABox that is an input for $Q$ and let 
  $\Amc'$ be its extension with
  \begin{enumerate}

  \item all assertions $\overline{A}(a)$ such that $A \in
    \Sigma_{\Csf}$, $a\in {\sf Ind}(\Amc)$, and $A(a) \notin \Amc$;

  \item assertions $A_{\mn{open}}(a_B)$, where $a_{B}$ is a fresh individual name
   for each $B \in \Sigma_{\Csf}$.

  \end{enumerate}
  We claim that $\Amc \not\models Q$ iff there is an interpretation
  $\Imc_T \in \Gamma_{Q}$ such that there exists a surjective
  homomorphism from $\Imc_{\Amc'}$ to $\Imc_T$. The assertions of
  type~(2) are needed to obtain a homomorphism that is surjective in
  the `$\Rightarrow$' direction, despite the presence of the elements
  $d_B$ in the templates in~$\Gamma_Q$.

  \smallskip
  \noindent
  ($\Leftarrow$). Let $\Imc_T \in \Gamma_{Q}$ and let $h$ be a surjective
  homomorphism from $\Imc_{\Amc'}$ to $\Imc_T$. Note that each element $a$ of
  ${\sf Ind}(\Amc)$ is mapped by $h$ to some element $t \in T$ of $\Imc_T$ because
  $A(a) \in \Amc'$ or $\overline{A}(a) \in \Amc'$ for every $A \in
  \Sigma_{\Csf}$ (which is non-empty). Since $\Imc_T \in \Gamma_{Q}$,
  there are a $\Sigma_{\Asf}$-ABox~\Bmc and a model \Imc of $\Tmc$ and
  $\Bmc$ that respects closed predicates $\Sigma_{\Csf}$ such that
  $\Imc \not\models q$ and $ T = \{ \mn{tp}_\Imc(a) \mid a \in
  \mn{Ind}(\Bmc) \}$.
%
  % Let $\mu$ be a function that assigns to every $a
  % \in \mn{Ind}(\Amc)$ an individual $\mu(a) \in \mn{Ind}(\Bmc)$ such
  % that $\mn{tp}_\Imc(b) = h(a)$.
%
  % Assume w.l.o.g.\ that $\mn{Ind}(\Amc) \cap
  % \mn{Ind}(\Bmc)=\emptyset$.
  For each $a \in \mn{Ind}(\Amc)$, set $t_a=h(a) \in T$ and for each
  $d \in \Delta^\Imc$, set $t_d=\mn{tp}_\Imc(d)$.  Construct an
  interpretation \Jmc as follows:
  $$
  \begin{array}{rcl}
    \Delta^\Jmc &=& \mn{Ind}(\Amc) \cup (\Delta^\Imc \setminus \mn{Ind}(\Bmc)) \\[1mm]
    A^\Jmc &=& \{ d \in \Delta^\Jmc \mid A \in t_d \} \\[1mm]
    % A^\Jmc &=& \{ a \in \mn{Ind}(\Amc) \mid A \in h(a) \} \, \cup
    % \\[1mm]
    % && A^\Imc \setminus \mn{Ind}(\Bmc) \\[1mm]
    % r^\Jmc &=& \{ (a,b) \in \mn{Ind}(\Amc) \times \mn{Ind}(\Amc)
    % \mid
    % r(a,b) \in \Amc \} \\[1mm]
    % r^\Jmc &=& \{ (a,b) \in \mn{Ind}(\Amc) \times \mn{Ind}(\Amc)
    % \mid
    % (\mu(a),\mu(b)) \in r^\Imc \} \, \cup \\[1mm]
    % &&\{ (a,d) \in \mn{Ind}(\Amc) \times (\Delta^\Imc \setminus
    % \mn{Ind}(\Bmc)) \mid (\mu(a),d) \in r^ \Imc\} \, \cup \\[1mm]
    % &&\{ (d,a) \in (\Delta^\Imc \setminus
    % \mn{Ind}(\Bmc)) \times \mn{Ind}(\Amc)\mid (d,\mu(a)) \in r^
    % \Imc\}
    % \, \cup \\[1mm]
    % &&r^\Imc \cap ((\Delta^\Imc \setminus
    % \mn{Ind}(\Bmc)) \times (\Delta^\Imc \setminus
    % \mn{Ind}(\Bmc))) \\[1mm]
    r^\Jmc &=& \{ (d,e) \in \Delta^\Jmc \times \Delta^\Jmc \mid
    t_d\rightsquigarrow_r t_e\}.
  \end{array}
  $$
  %
  % for all closed roles $r$ and open roles $s$.
  First note that \Jmc is clearly a model of $\Amc$ that respects
  closed predicates $\Sigma_{\Csf}$. Specifically, if $A(a) \in \Amc$,
  then $h(a) \in A^{\Imc_T}$, thus $A \in h(a) =t_a$ by construction
  of $\Imc_T$ which yields $a \in A^\Jmc$ by construction of \Jmc; if
  $r(a,b) \in \Amc$, then $(h(a),h(b)) \in r^{\Imc_T}$, thus $t_a
  \rightsquigarrow_r t_b$ implying $(a,b) \in r^\Jmc$; finally if $A
  \in \Sigma_{\Csf}$ and $d \in A^\Jmc$, then we must have $d=a$ for
  some $a \in \mn{Ind}(\Amc)$ by definition of \Jmc and since $d
  \notin A^\Imc$ for all $d \in \Delta^\Imc \setminus \mn{Ind}(\Bmc)$.
  Thus, $A \in t_a=h(a)$ by construction of \Jmc. This implies $A(a)
  \in \Amc$ since otherwise $\overline{A}(a) \in \Amc'$, which would
  imply $\overline{A} \in h(a)$, in contradiction to $A \in h(a)$.

  \smallskip
  
  It thus remains to show that \Jmc is a model of \Tmc and $\Jmc \not
  \models q$. By definition, \Jmc satisfies all RIs in
  \Tmc. Satisfaction of the CIs in $\Tmc$ and $\Jmc \not
  \models q$ follow from the subsequent claim together with the condition that no type in $T$ contains $A_{0}$ and each type in $\Imc_T$ is satisfied in a model of $\Tmc$.
  \begin{claim*}
    For all $d \in \Delta^\Jmc$ and $C \in \mn{cl}(\Tmc)$, we have $d 
    \in C^\Jmc$ iff $C \in t_d$.
  \end{claim*}
  \begin{clmproof}
  The proof is by induction on the structure of $C$, with the
  induction start and the cases $C= \neg D$ and $C = D_1 \sqcap D_2$
  being trivial.  Thus let $C = \exists r . D$ and first assume $d \in
  C^\Jmc$.  Then there is an $e \in D^\Jmc$ with $(d,e) \in
  r^\Jmc$. Thus $t_d \rightsquigarrow_r t_e$ by definition of \Jmc,
  and IH yields $D \in t_e$. By definition of `$\rightsquigarrow_r$',
  we must thus have $C \in t_d$ as required.  Now let $C \in t_d$. We
  distinguish two cases:
  \begin{itemize}

  \item $d=a \in \mn{Ind}(\Amc)$.

    Let $a' \in \mn{Ind}(\Bmc)$ be such that
    $h(a)=\mn{tp}_\Imc(a')$. Since $t_a=h(a)$, we must have $a' \in
    C^\Imc$ and thus there is some $e \in D^\Imc$ with $(a',e) \in
    r^\Imc$, which yields $\mn{tp}_\Imc(a') \rightsquigarrow_r
    \mn{tp}_\Imc(e)$ and $D \in \mn{tp}_\Imc(e)$. If $e=b' \in
    \mn{Ind}(\Bmc)$, then since $h$ is surjective there is some $b \in
    \mn{Ind}(\Amc)$ with $h(b)=\mn{tp}_\Imc(b')$. We have
    $t_a=\mn{tp}_\Imc(a')$ and $t_b=\mn{tp}_\Imc(b')$, thus $t_a
    \rightsquigarrow_r t_b$ which yields $(a,b) \in r^\Jmc$ by
    definition of \Jmc. We also have $D \in t_b$, which by IH yields
    $b \in D^\Jmc$.

  \item $d \notin \mn{Ind}(\Amc)$.

    Then $d \in \Delta^\Imc \setminus \mn{Ind}(\Bmc)$. Since $C \in
    t_d$, we thus have $C \in \mn{tp}_\Imc(d)$. Thus, there is an $e
    \in D^\Imc$ with $(d,e) \in r^\Imc$, which implies
    $\mn{tp}_\Imc(d) \rightsquigarrow_r \mn{tp}_\Imc(e)$ and $D \in
    \mn{tp}_\Imc(e)$.  If $e \notin \mn{Ind}(\Bmc)$, then the
    definition of \Jmc and IH yields $d \in C^\Jmc$. Thus assume $e=b'
    \in \mn{Ind}(\Bmc)$. Since $h$ is surjective, there is some $b \in
    \mn{Ind}(\Amc)$ with $h(b)=\mn{tp}_\Imc(b')$. Since
    $t_d=\mn{tp}_\Imc(d)$ and $t_b=h(b)$, we have $t_d
    \rightsquigarrow_r t_b$, thus $(d,b) \in r^\Jmc$. By IH, $D \in
    \mn{tp}_\Imc(b')=h(b)$ yields $b \in D^\Jmc$.
    
  \end{itemize}
  \end{clmproof}
  
  \noindent
  $(\Rightarrow)$.  Assume that $\Amc \not\models Q$. Then there is a 
  model \Imc of \Tmc and $\Amc$ that respects closed predicates 
  $\Sigma_{\Csf}$ and such that $\Imc\not\models q$.  Let $\Imc_T \in 
  \Gamma_{Q}$ be the corresponding template, that is, $ T = \{ 
  \mn{tp}_\Imc(a) \mid a \in \mn{Ind}(\Amc) \}$. For each $a \in 
  \mn{Ind}(\Amc)$, set $h(a)=\mn{tp}_\Imc(a)$; for each $a_B \in 
  \mn{Ind}(\Amc') \setminus \mn{Ind}(\Amc)$, set $h(a_B)=d_B$ (recall that such $a_{B}$ have
  been added to ${\sf Ind}(\Amc)$ for every $B\in \Sigma_{\Csf}$). It is 
  readily checked that $h$ is a surjective homomorphism from $\Imc_{\Amc'}$ 
  to~$\Imc_T$.  In particular, $\overline{A}(a) \in \Amc'$ implies 
  $A(a) \notin \Amc'$, thus $A \notin \mn{tp}_\Imc(a)$ (since $A$ is 
  closed), which yields $h(a)=\mn{tp}_\Imc(a) \in 
  \overline{A}^{\Imc_\Tmc}$ by definition of~$\Imc_\Tmc$.
\end{proof}

\begin{lem}
   Let $Q=(\Tmc,\Sigma_\Asf,\Sigma_\Csf,q)$ be an OMQC from
  $(\alchi,\NC,\text{BUtCQ})$. Then $\mn{CSP}(\Gamma_{Q})^{\mn{sur}}$ 
  reduces in polynomial time to the complement of the evaluation problem for $Q$.
\end{lem}
\begin{proof}
  Let $\Amc'$ be the ABox corresponding to an input $\Jmc$ for $\mn{CSP}(\Gamma_{Q})^{\mn{sur}}$. An
  element $a$ of ${\sf Ind}(\Amc')$ is \emph{special for} $A \in \Sigma_{\Csf}$ if
  $A(a) \notin \Amc'$ and $\overline{A}(a) \notin \Amc'$; it is
  \emph{special} if it is special for some $A \in \Sigma_{\Csf}$.
  First perform the following checks:
  \begin{enumerate}

  \item if there is a non-special element $a$ of ${\sf Ind}(\Amc')$ such that
    $A(a) \in \Amc'$ and $\overline{A}(a) \in \Amc'$ for some $A \in
    \Sigma_{\Csf}$, then return `no' (there is no template in $\Gamma_{Q}$
    that has any element to which $a$ can be mapped by a
    homomorphism);

  \item if $\Amc'$ does not contain a family of distinct elements
    $(a_A)_{A \in \Sigma_{\Csf}}$, such that each $a_A$ is special for
    $A$, then return `no' (we cannot map surjectively to the elements
    $d_A$ of the templates in $\Gamma_{Q}$).

  \end{enumerate}
  Note that, to check Condition~2, we can go through all candidate
  families in polytime since the size of $\Sigma_{\Csf}$ is constant.
  If none of the above checks succeeds, then let $\Amc$ be the ABox
  obtained from $\Amc'$ by
  \begin{itemize}

  \item deleting all assertions of the form $\overline{A}(a)$ and

  \item deleting all special elements.

  \end{itemize}
  We have to show that $\Amc \not\models Q$ iff there exists an $\Imc_T 
  \in \Gamma_{Q}$ such that there is a surjective homomorphism from $\Jmc$  
  to~$\Imc_T$.

  \smallskip
  \noindent
  ($\Leftarrow$). Let $\Imc_T \in \Gamma_{Q}$ and let $h$ be a surjective
  homomorphism from $\Jmc$ to $\Imc_T$. Note that each element $a$ of
  ${\sf Ind}(\Amc)$ is mapped by $h$ to some element $t \in T$ of $\Imc_T$ because
  $A(a) \in \Amc'$ or $\overline{A}(a) \in \Amc'$ for every $A \in
  \Sigma_{\Csf}$ (which is non-empty). Since $\Imc_T \in \Gamma_{Q}$,
  there is a $\Sigma_{\Asf}$-ABox~\Bmc and model \Imc of $\Tmc$ and
  $\Bmc$ that respects closed predicates $\Sigma_{\Csf}$ and such that
  $\Imc \not\models q$ and $ T = \{ \mn{tp}_\Imc(a) \mid a \in
  \mn{Ind}(\Bmc) \}$. We can now proceed as in the proof of
  Lemma~\ref{lem:first} to build a model $\Jmc'$ of \Tmc and $\Amc$ that
  respects closed predicates $\Sigma_{\Csf}$ and such that $\Jmc'\not\models q$.

  \smallskip
  \noindent
  ($\Rightarrow$).  Assume that $\Amc
  \not\models Q$. Then there is a model \Imc of
  \Tmc and $\Amc$ that respects closed predicates $\Sigma_{\Csf}$ and
  such that $\Imc\not\models q$.  Let $\Imc_T \in \Gamma_{Q}$ be the
  corresponding template, that is, $ T = \{ \mn{tp}_\Imc(a) \mid a \in
  \mn{Ind}(\Amc) \}$. For each $a \in \mn{Ind}(\Amc)$, set
  $h(a)=\mn{tp}_\Imc(a)$; for each element $a \in \mn{Ind}(\Amc')
  \setminus \mn{Ind}(\Amc)$, we can choose some $A \in \Sigma_{\Csf}$
  such that $A(a) \notin \Amc'$ and $\overline{A}(a) \notin \Amc'$, and
  set $h(a)=d_A$; by Check~2 above, these choices can be made such
  that the resulting map $h$ is surjective. Moreover, it is readily
  checked that $h$ is a homomorphism from $\Jmc$ to $\Imc_T$.  In
  particular, $\overline{A}(a) \in \Amc'$ implies $A(a) \notin \Amc'$
  by Check~1, thus $A \notin \mn{tp}_\Imc(a)$ (since $A$ is closed),
  which yields $h(a)=\mn{tp}_\Imc(a) \in \overline{A}^{\Imc_\Tmc}$ by
  definition of~$\Imc_\Tmc$.
\end{proof}

We have thus established the following result.
\begin{thm}
  \label{thn:cspsecdir}
  For every OMQC $Q$ from $(\alchi,\NC,\text{BUtCQ})$, there
  is a generalized $\mn{CSP}(\Gamma_{Q})^\mn{sur}$ 
  such that the evaluation problem for $Q$ has the same complexity as the complement of
  $\mn{CSP}(\Gamma_{Q})^\mn{sur}$, up to polynomial time reductions.
\end{thm}
Again, the theorem can easily be strengthened to state the same
complexity up to FO reductions. Note that the DL \alchi
used in Theorem~\ref{thn:cspsecdir} is a significant extension of the
DLs referred to in Theorem~\ref{thn:cspfirstdir} and thus our results
apply to a remarkable range of DLs: all DLs between
\dllitecore and \alchi as well as all DLs between
\EL and \alchi.

\section{Closing Role Names in the Fixed Query Case: Turing Machine Equivalence}
\label{TMequi}

We generalize the setup from the previous section by allowing also
role names to be closed. Our main results are that for every
non-determinstic polynomial time Turing machine $M$, there is an OMQC
$Q$ in $(\dlliter,\NC \cup \NR,\text{BUtCQ})$ such that evaluating $Q$
and the complement of $M$'s word problem are polynomial time reducible
to each other, and that it is undecidable whether evaluating OMQCs in
$(\dlliter,\NC \cup \NR,\text{BUtCQ})$ is in \ptime (unless \ptime =
\np).  By Ladner's theorem, it follows that there are
\conp-intermediate OMQCs (unless \ptime = \np) and that a full
complexity classification of the OMQCs in this language is beyond
reach of the techniques available today. As in the previous section,
the same results hold for $(\mathcal{EL},\NC\cup \NR,\text{BAQ})$.

To establish these results, we utilize two related results
from~\cite{DBLP:journals/lmcs/LutzW17,DBLP:journals/tods/BienvenuCLW14}:
(1)~for every \np Turing machine~$M$, there is an ontology-mediated
query $Q$ from $(\mathcal{ALCF},\emptyset,\text{BAQ})$ such that
evaluating $Q$ is reducible in polynomial time to the complement of
$M$'s word problem and vice versa, where $\mathcal{ALCF}$ is the
extension of $\mathcal{ALC}$ with functional roles; and (2)~ it is
undecidable whether an OMQC from
$(\mathcal{ALCF},\emptyset,\text{BAQ})$ is in \ptime. For using these
results in our context, however, it is more convenient to phrase them
in terms of (a certain kind of) monadic disjunctive datalog programs
with inequality rather than in terms of OMQCs from
$(\mathcal{ALCF},\emptyset,\text{BAQ})$. This is what we do in the
following, starting with the introduction of a suitable version of
monadic disjunctive datalog.  For a more thorough introduction, see~\cite{EiterGottlob}.

A \emph{monadic disjunctive datalog rule (MDD
  rule)} $\rho$ takes the form
$$
P_1(x) \vee \cdots \vee P_m(x) \leftarrow R_1(\vec{x}_{1})\land \cdots\land R_n(\vec{x}_n) \quad \mbox{ or } \quad {\sf goal}\leftarrow R_1(\vec{x}_{1})\land \cdots\land R_n(\vec{x}_n)
$$
with $m,n>0$ and where all $P_{i}$ are unary predicates, ${\sf goal}$ is the \emph{goal predicate} of arity $0$,
and all $R_{i}$ are predicates of arity one or two, including possibly the non-equality predicate $\not=$.
We refer to $P_1(x) \vee \cdots \vee P_m(x)$ and, respectively, ${\sf goal}$ as the \emph{head} of $\rho$, and to $R_1(\vec{x}_{1}) \wedge \cdots \wedge
R_n(\vec{x}_{n})$ as the \emph{body}. A \emph{monadic disjunctive
  datalog (MDD) program} $\Pi$ is a finite set of MDD rules containing at least one rule with the goal predicate in its head and
no rule with the goal predicate in its body. Predicates that occur in the head of at
least one rule of $\Pi$ are \emph{intensional (IDB) predicates}, denoted ${\sf IDB}(\Pi)$, and
all remaining predicates in $\Pi$ are \emph{extensional (EDB) predicates}, denoted ${\sf EDB}(\Pi)$. 
An interpretation $\Imc$ is a \emph{model} of $\Pi$ if it satisfies all rules in $\Pi$ (viewed as universally quantified first-order sentences).
$\Pi$ is \emph{entailed on} a ${\sf EDB}(\Pi)$-ABox $\Amc$, in symbols $\Amc\models \Pi$, iff ${\sf goal}$ is true in every model of $\Pi$ and $\Amc$.
Note that it suffices to consider models that respect closed predicates ${\sf EDB}(\Pi)$. 
The \emph{evaluation problem for $\Pi$} is the problem to decide whether $\Pi$ is entailed by an ${\sf EDB}(\Pi)$-ABox~$\Amc$.

For our reduction, we use the following kind of MDD programs that we
call basic. A binary predicate $r$ is \emph{functional in an ABox}
$\Amc$ if $r(a,b_{1}),r(a,b_{2})\in \Amc$ implies $b_{1}=b_{2}$ and
$r$ is \emph{empty in} $\Amc$ if $r$ does not occur in $\Amc$. Then an
MDD program $\Pi$ is \emph{basic} if
\begin{itemize}
\item $\Pi$ uses exactly two binary predicates, $r_{1},r_{2}$, and
  contains exactly the following functionality rules, for $i=1,2$:
$$
{\sf goal} \leftarrow r_{i}(x,y) \wedge r_{i}(x,z) \wedge (y\not=z) 
$$
\item all remaining rules of $\Pi$ are of the form
$$
P_{1}(x)\vee \cdots \vee P_{n}(x) \leftarrow q \quad \mbox{ or } \quad {\sf goal}\leftarrow q
$$
where $n\geq 1$ and $q$ is a dtCQ with root $x$ (with the quantifier prefix removed).
\item if $r_{1},r_{2}$ are functional and at least one $r_{i}$ is empty in an ${\sf EDB}(\Pi)$-ABox $\Amc$, then $\Amc\not\models\Pi$.
\end{itemize}
The following result can be obtained by starting from the results for
$(\mathcal{ALCF},\emptyset,\text{BAQ})$ from
\cite{DBLP:journals/lmcs/LutzW17,DBLP:journals/tods/BienvenuCLW14}
mentioned above and translating the involved OMQs into a basic MDD
program. Such a translation is given in
\cite{DBLP:journals/tods/BienvenuCLW14} for the case of \ALC TBoxes
and MDD programs without inequality, but the extension to functional
roles and inequality is trivial.
\begin{thm}\label{thm:mdd}~\\[-4mm]
\begin{enumerate}
\item For every non-deterministic polynomial time Turing machine $M$, there exists a basic MDD program $\Pi$ such that the
evaluation problem for $\Pi$ and the complement of $M$'s word problem are polynomial time reducible to each other.
\item It is undecidable whether the evaluation problem for a basic MDD program is in \ptime (unless \ptime=\np).
\end{enumerate}
\end{thm}
We next prove the following central theorem.
\begin{thm}
  \label{thm:tmequi}
  For every basic MDD program $\Pi$, one can construct an
  OMQC $Q$ in the language $(\mathcal{EL},\NC\cup\NR,\text{BAQ})$ such that the evaluation problem for $Q$
  and $\Pi$ are polynomial time reducible to each other. The same is true for $(\dlliter,\NC\cup \NR,\text{BUtCQ})$.
\end{thm}
\begin{proof}
  Assume a basic MDD program $\Pi$ of the form defined above is given. We first construct an OMQC $Q_{\Pi}=(\mathcal{T}_{\Pi},\Sigma_{\Pi},\Sigma_{\Pi},q_{\Pi})$ in
  $(\mathcal{EL},\NC\cup\NR,\text{BUtCQ})$ and then obtain the
  required OMQCs in $(\mathcal{EL},\NC\cup\NR,\text{BAQ})$ and $(\dlliter,\NC\cup\NR,\text{BUtCQ})$ by rather straightforward
  modifications of $Q_{\Pi}$. Note that we construct a $Q_{\Pi}$ in which the ABox signature and set of closed predicates coincide.
  We set $\Sigma_{\Pi} = {\sf EDB}(\Pi)\cup \{T,F,V\}$, where $T,F,V$ are fresh concept names. We also use auxiliary predicates which are
  not in the ABox signature of $Q_{\Pi}$: role names ${\sf val}_{P}$ for every unary $P\in {\sf IDB}(\Pi)$ and role names $s_{i}$
  and concept names $A_{i},B_{i}$, for $i=1,2$. $\Tmc_{\Pi}$ contains the following CIs:
$$
\begin{array}{rcl}
  T &\sqsubseteq& V \\[1mm]
  F &\sqsubseteq& V  \\[1mm]
  \top &\sqsubseteq& \exists {\sf val}_{P}.V, \text{ for all unary $P\in {\sf IDB}(\Pi)$ } \\[1mm]
  \top &\sqsubseteq& \exists s_{i}.(\exists r_{i}.A_{i} \sqcap \exists r_{i}.B_{i}), \text{ for $i=1,2$ }.
\end{array}
$$
Using $\Tmc_{\Pi}$, we encode the truth value of IDB predicates $P$ using the CQs 
$$
P^{T}(x,y):= ({\sf val}_{P}(x,y)\wedge T(y)), \quad P^{F}(x,y):= ({\sf val}_{P}(x,y)\wedge F(y)).
$$
For any tCQ $q$, we denote by $q^{T}$ the result of replacing every occurrence of an 
IDB $P(x)$ in $q$ by $P^{T}(x,y_{0})$, where the variable $y_{0}$ is fresh for every 
occurrence of $P(x)$, and existentially quantified.
Thus, $q^{T}$ is again a tCQ (and a dtCQ if $q$ is already a dtCQ).
The final CI is used to encode functionality of the roles $r_{1},r_{2}$. We define
CQs $q_{\mathcal{F}}^{1}$ and $q_{\mathcal{F}}^{2}$ by setting
$$
q_{\mathcal{F}}^{i}= (s_{i}(x,y)\wedge r_{i}(y,z)\wedge A_{i}(z) \wedge B_{i}(z)),
$$
for $i=1,2$. Then, for 
the OMQC $Q_{i}=(\{\top \sqsubseteq \exists s_{i}.(\exists r_{i}.A_{i} \sqcap \exists r_{i}.B_{i})\},\Sigma_{\Pi},\Sigma_{\Pi},
\exists y\,\exists z\,q_{\mathcal{F}}^{i})$
and any $\Sigma_{\Pi}$-ABox $\Amc$:
\begin{itemize}
\item if $r_{i}$ is empty in $\Amc$, then $\Amc$ is not consistent w.r.t.~$(\{\top \sqsubseteq \exists s_{i}.(\exists r_{i}.A_{i} \sqcap \exists r_{i}.B_{i})\},\Sigma_{\Pi})$, and 
\item if $r_{i}$ is not empty in $\Amc$, then $r_{i}$ is functional in $\Amc$ iff 
$\Amc\models Q_{i}(a)$, for some (equivalently, all) $a\in {\sf Ind}(\Amc)$.
\end{itemize}
Define $q_{\Pi}$ as the union of the following Boolean CQs, where for brevity we omit the existential quantifiers:
\begin{itemize}
\item $q_{\mathcal{F}}^{1}\wedge q_{\mathcal{F}}^{1} \wedge q^{T}$, for every rule 
${\sf goal}\leftarrow q\in \Pi$, where we assume that the only variable shared by
any two of the conjuncts $q_{\mathcal{F}}^{1}$, $q_{\mathcal{F}}^{2}$ and $q^{T}$ is $x$.
\item $q_{\mathcal{F}}^{1} \wedge q_{\mathcal{F}}^{2} \wedge q^{T} \wedge \bigwedge_{1\leq i \leq n}P_{i}^{F}$, 
for every $P_{1}(x)\vee \cdots \vee P_{n}(x) \leftarrow q \in \Pi$,
where we assume again that the only variable shared by any two of the conjuncts $q_{\mathcal{F}}^{1}$, 
$q_{\mathcal{F}}^{2}$, $q^{T}$, $P_{i}^{F}$, $1\leq i \leq n$, is $x$.
\end{itemize}
We prove the following

\begin{claim*}
The problem of evaluating $\Pi$ and the problem of evaluating $Q_{\Pi}$ are polynomial time reducible to each other.
\end{claim*}
\begin{clmproof}
($\Rightarrow$) Assume an $\mn{EDB}(\Pi)$-ABox $\Amc$ is given
as an input to $\Pi$. If $r_{1}$ or $r_{2}$ is not functional in $\Amc$, then output
`$\Amc\models \Pi$'. Otherwise, if $r_{1}$ or $r_{2}$ is empty, then output `$\Amc\not\models\Pi$'. 
Now assume that $r_{1}$ and $r_{2}$ are not empty and both are functional in $\Amc$.
Let
$$
\Amc'=\Amc \cup \{T(a),F(b),V(a),V(b)\},
$$
where we asume w.l.o.g.~that $a,b$ occur in ${\sf Ind}(\Amc)$.  We show that
$\Amc \models \Pi$ iff $\Amc'\models Q_{\Pi}$.

Assume first that $\Amc\not\models \Pi$. Let $\Imc$ be a model of
$\Amc$ and $\Pi$ that respects closed predicates $\mn{EDB}(\Pi)$ and
satisfies no body of any rule ${\sf goal}\leftarrow q\in \Pi$.  Define $\Imc'$ in the same
way as $\Imc$ except that 
\begin{itemize}
\item $T^{\Imc'}=\{a\}$, $F^{\Imc'}=\{b\}$, and $V^{\Imc'}=\{a,b\}$;
  % \item Let $s_{r}^{\Imc'}=A_{r}^{\Imc'} = B_{r}^{\Imc'}:=\emptyset$
  %   if $r$ is a functional role with $r^{\Imc}=\emptyset$;
\item $s_{i}^{\Imc'} = \Delta^{\Imc}\times {\sf dom}(r_{i}^{\Imc})$ and $A_{i}^{\Imc'} = B_{i}^{\Imc'} = \Delta^{\Imc}$, for $i=1,2$,
where ${\sf dom}(r^{\Imc})$ denotes the domain of $r^{\Imc}$;
\item 
${\sf val}_{P}^{\Imc'} = (P^{\Imc} \times \{a\}) \cup ((\Delta^{\Imc}\setminus P^{\Imc}) \times \{b\})$, 
for all unary $P\in {\sf IDB}(\Pi)$.
\end{itemize}
It is straightforward to show that $\Imc'$ is a model of $\Tmc_{\Pi}$
and $\Amc'$ that respects closed predicates $\Sigma_{\Pi}$. It
remains to show that $\Imc'\not\models q_{\Pi}$.  To this end it is
sufficient to show that
\begin{enumerate}
\item No $q^{T}$ with ${\sf goal}\leftarrow q\in \Pi$ is satisfied in $\Imc'$;
\item No $q^{T}\wedge \bigwedge_{1\leq i \leq n} P_{i}^{F}$ with $P_{1}(x)\vee \cdots \vee P_{n}(x) \leftarrow q\in \Pi$
is satisfied in $\Imc'$.
\end{enumerate}
Point~(1) holds since $P^{\Imc}=\{ d \mid \Imc'\models \exists y \, P^{T}(d,y))\}$ for all unary $P\in {\sf IDB}(\Pi)$,
by definition of $\Imc'$ and since $q$ is not satisfied in $\Imc$ for any rule ${\sf goal}()\leftarrow q\in \Pi$. 
Point~(2) holds since all rules $P_{1}(x)\vee \cdots \vee P_{n}(x) \leftarrow q\in \Pi$ are satisfied in $\Imc$
and $P^{\Imc}= \Delta^{\Imc}\setminus\{ d \mid \Imc'\models \exists y\, P^{F}(d,y)\}$ for all unary $P\in {\sf IDB}(\Pi)$.

\medskip

Assume now that $\Amc'\not\models Q_{\Pi}$. Take a model
$\Imc$ of $\Amc'$ that respects closed predicates $\Sigma_{\Pi}$ and
such that $\Imc\not\models q_{\Pi}$. Define a model $\Imc'$ by modifying $\Imc$ by setting
$P^{\Imc'} = \{ d \mid \Imc\models \exists y\,P^{T}(d,y)\}$, for all unary $P\in {\sf IDB}(\Pi)$. It follows
from the condition that $r_{1},r_{2}$ are non-empty and
functional in $\Amc'$ that $\Imc\models \forall x (\exists y\exists z q_{\mathcal{F}}^{1} \wedge \exists y \exists z
q_{\mathcal{F}}^{2})$. 
From $\Imc\not\models q_{\Pi}$ we obtain that no $q$ with ${\sf goal}()\leftarrow q\in \Pi$ is satisfied
in $\Imc'$ and that all rules $P_{1}(x)\vee \cdots \vee P_{n}(x) \leftarrow q\in \Pi$ are satisfied in $\Imc'$.
Thus, $\Imc'$ is a model of $\Amc$ and $\Pi$ witnessing that $\Amc\not\models\Pi$.

\medskip
\noindent
($\Leftarrow$) Assume a $\Sigma_{\Pi}$-ABox $\Amc$ is given as an
input to $Q_\Pi$. There exists a model of $\Tmc_{\Pi}$ and $\Amc$ that
respects closed predicates $\Sigma_{\Pi}$ iff (i) $V$ is
non-empty in $\Amc$, (ii) $T,F$ are both contained
in $V$ in $\Amc$, and (iii) $r_{1},r_{2}$ are
non-empty in $\Amc$. Thus, output `$\Amc\models Q_{\Pi}$' whenever (i),
(ii), or (iii) is violated.  Now assume (i), (ii), and (iii) hold.  If
$r_{1}$ or $r_{2}$ are not functional in $\Amc$, then we can construct a model of $\Tmc_{\Pi}$ and $\Amc$ that respects
closed predicates $\Sigma_{\Pi}$ and such that $\exists x (\exists y\exists z q_{\mathcal{F}}^{1} \wedge \exists y \exists z
q_{\mathcal{F}}^{2})$ is not satisfied in $\Imc$. Hence, we output `$\Amc\not\models Q_{\Pi}$'.
Thus, assume in addition to (i), (ii) and (iii) that $r_{1}$ and $r_{2}$ are functional in $\Amc$. 
We distinguish five cases.
We only consider the first case in detail, the remaining cases are proved similarly.
\begin{enumerate}
\item If $F^{\Imc_{\Amc}}\cup T^{\Imc_{\Amc}}\not=V^{\Imc_{\Amc}}$, then output `$\Amc\models Q_{\Pi}$' if there exists a rule ${\sf goal}\leftarrow q\in \Pi$ such that
$q$ contains not IDBs and $q$ (which then equals $q^{T}$) is satisfied in $\Imc_{\Amc}$.
This is clearly correct since $\Amc\models (\Tmc_{\Pi},\Sigma_{\Pi},\Sigma_{\Pi},q_{\mathcal{F}}^{1}\wedge q_{\mathcal{F}}^{2}\wedge q^{T})$ 
follows. 
Otherwise output `$\Amc\not\models Q_{\Pi}$'. To prove correctness, let $a\in
V^{\Imc_{\Amc}}\setminus (F^{\Imc_{\Amc}}\cup T^{\Imc_{\Amc}})$. Construct a model $\Imc$ of $\Tmc_{\Pi}$ and $\Amc$ that respects $\Sigma_{\Pi}$ 
by extending $\Imc_{\Amc}$ by setting ${\sf val}_{P}^{\Imc}= \Delta^{\Imc}\times \{a\}$ for all unary IDB predicates $P$
and defining $s_{i}^{\Imc},A_{i}^{\Imc},B_{i}^{\Imc}$, $i=1,2$, arbitrarily so that $\Tmc_{\Pi}$ is satisfied. 
Then no $q^{T}$ with ${\sf goal}\leftarrow q\in \Pi$ and no 
$\bigwedge_{1\leq i \leq n} P_{i}^{F}$ with $P_{1}(x)\vee \cdots \vee P_{n}(x) \leftarrow q\in \Pi$ is satisfied in $\Imc$.
Thus $\Imc\not\models q_{\Pi}$.
\item If $T^{\Imc_{\Amc}}=F^{\Imc_{\Amc}}=V^{\Imc_{\Amc}}$, then output `$\Amc\models Q_{\Pi}$' if there exists a rule ${\sf goal}\leftarrow q\in \Pi$ 
or $P_{1}(x)\vee \cdots \vee P_{n}(x) \leftarrow q\in \Pi$ such that $q'$ is satisfied in $\Imc_{\Amc}$ for the query $q'$ obtained from 
$q$ by removing every atom $P(y)$ from $q$ with $P$ a unary IDB. Otherwise output `$\Amc\not\models Q_{\Pi}$'. 
\item If $T^{\Imc_{\Amc}}=V^{\Imc_{\Amc}}$ and $F^{\Imc_{\Amc}}\not=V^{\Imc_{\Amc}}$, then output `$\Amc\models Q_{\Pi}$' if there exists a rule ${\sf goal}\leftarrow q\in \Pi$ 
such that $q'$ is satisfied in $\Imc_{\Amc}$ for the query $q'$ obtained from $q$ by removing every atom $P(y)$ from $q$ with $P$ a unary IDB. 
Otherwise output `$\Amc\not\models Q_{\Pi}$'. 
\item If $F^{\Imc_{\Amc}}=V^{\Imc_{\Amc}}$ and $T^{\Imc_{\Amc}}\not=V^{\Imc_{\Amc}}$, then output `$\Amc\models Q_{\Pi}$' if there exists a rule 
${\sf goal}\leftarrow q\in \Pi$ or $P_{1}(x)\vee \cdots \vee P_{n}(x) \leftarrow q\in \Pi$ such that $q$ does not contain
any IDB and $q$ is satisfied in $\Imc_{\Amc}$. Otherwise output `$\Amc\not\models Q_{\Pi}$'.
\item If none of the four cases above apply, obtain $\Amc'$ from $\Amc$ by removing all assertions using $T,F$, or $V$. 
Then $\Amc'\models \Pi$ iff $\Amc\models Q$, and we have established the polynomial time reduction. \qedhere
\end{enumerate}
\end{clmproof}

The modification of $Q_{\Pi}$ needed to obtain an OMQC from
$(\mathcal{EL},\NC\cup \NR,\text{BAQ})$ is the same as in the proof of Theorem~\ref{thn:cspfirstdir}:
the query $q_{\Pi}$ is a BUdtCQ and so we can replace it with a query of the form $\exists x\,A(x)$:
as the disjuncts of $q_{\Pi}$ are of the form $\exists x\, q'(x)$ with $q'(x)$ a dtCQ, we can take the 
\EL concepts $C_{q'}$ corresponding to $q'(x)$ and extend $\Tmc_{\Pi}$ with $C_{q'}
\sqsubseteq A$ for every such disjunct $\exists x\, q'(x)$ of $q$.

It remains to show how one can modify $Q_{\Pi}$ to obtain an equivalent OMQC $Q_{\Pi}'$ from the language
$(\dlliter,\NC\cup\NR,\text{BUtCQ})$. First, 
to eliminate $\top$ on the left-hand-side of CIs in $\Tmc_{\Pi}$, we replace 
each CI $\top\sqsubseteq C$ by the CIs $A\sqsubseteq C$, $\exists
r\sqsubseteq C$, and $\exists r^{-}\sqsubseteq C$ for any concept name
$A\in \Sigma_{\Pi}$ and role name $r\in \Sigma_{\Pi}$. Second, we
employ the standard encoding of qualified existential restrictions in
\dlliter by replacing exhaustively any $B \sqsubseteq
\exists r.D$ by $B \sqsubseteq \exists s$, $\exists s^{-} \sqsubseteq
A_{D}$, $A_{D} \sqsubseteq D$, and $s \sqsubseteq r$, where $A_{D}$ is
a fresh concept name and $s$ is a fresh role name. Let
$\Tmc_{\Pi}'$ be the resulting TBox.  Then $Q_{\Pi}'=(\Tmc_{\Pi}',\Sigma_{\Pi},\Sigma_{\Pi},q_{\Pi})$ is as required.
\end{proof}
From Theorems~\ref{thm:mdd} and~\ref{thm:tmequi}, we obtain the main result of this section.
\begin{thm}\label{thm:mmm}~\\[-4mm]
\begin{enumerate}
\item For every non-deterministic polynomial time Turing machine $M$ one can construct a OMQC $Q$ in the languages
$(\mathcal{EL},\NC\cup\NR,\text{BAQ})$ and $(\dlliter,\NC\cup \NR,\text{BUtCQ})$
such that the evaluation problem for $Q$ and $M$'s word problem are polynomial time reducible to each other. 
\item It is undecidable whether the evaluation problem for OMQCs in
$(\mathcal{EL},\NC\cup\NR,\text{BAQ})$ and $(\dlliter,\NC\cup \NR,\text{BUtCQ})$ is in \ptime (unless \ptime=\np).
\end{enumerate}
\end{thm}
Note that Theorem~\ref{thm:mmm} does not cover $\text{\dllite}_{\text{core}}$. In fact,
the computational status of the language $(\text{\dllite}_{\text{core}},\NC \cup
\NR,\text{BUtCQ})$ remains open, and in particular it remains open whether
Theorem~\ref{thm:mmm} can be strengthened to this case.%  We are,

\section{Quantifier-Free UCQs and FO-Rewritability}
\label{sect:fullq}

The results in the previous sections have shown that intractability
comes quickly when predicates are closed. The aim of this section is
to identify a useful OMQC language whose UCQs are guaranteed to be
FO-rewritable. It turns out that one can obtain such a language by
combining \dlliter with quantifier-free UCQs, that is, unions of
quantifier-free CQs; we denote this class of queries with UqfCQ. Our
main result is that all OMQCs from the language
$(\dlliter,\NC \cup \NR,\text{UqfCQ})$ are FO-rewritable under the
mild restriction that there is no RI which requires an open role to be
contained in a closed one. We believe that this class of OMQCs is
potentially relevant for practical applications. Note that the query
language SPARQL, which is used in many web applications, is closely
related to UqfCQs and, in fact, does not admit existential
quantification under its standard entailment regimes
\cite{DBLP:conf/semweb/GlimmK10}. We also prove that the restriction
on RIs is needed for tractability, by constructing a \conp-hard OMQC
in $(\dlliter,\NC \cup \NR,\text{UqfCQ})$.
\begin{thm}
  \label{thm:rewr}
  Every OMQC $(\Tmc,\Sigma_{\Asf},\Sigma_{\Csf},q)$ from
  $(\dlliter,\NC \cup \NR,\text{UqfCQ})$ such that
  $\Tmc$ contains no RI of the form $s \sqsubseteq r $
  with ${\sf sig}(s)\not\subseteq\Sigma_{\Csf}$ and ${\sf sig}(r)\subseteq \Sigma_{\Csf}$ is
  FO-rewritable.
\end{thm}
We first show that ABox consistency w.r.t.\ 
$(\tbox,\Sigma_{\Asf},\Sigma_{\Csf})$ is FO-rewritable, for every 
\dlliter TBox $\Tmc$ not containing any RI of the form $s \sqsubseteq r $
with ${\sf sig}(s)\not\subseteq\Sigma_{\Csf}$ and ${\sf sig}(r)\subseteq \Sigma_{\Csf}$.
We make use of Theorem~\ref{thm:basic1} and assume w.l.o.g.\ that 
$\Sigma_{\Csf}=\Sigma_{\Asf}$. 
Let $\mn{con}(\Tmc)$ be the set of all concept names in \Tmc, and all concepts 
$\exists r, \exists r^-$ such that $r$ is a 
role name that occurs in \Tmc.  A \emph{\Tmc-type} is a set $t 
\subseteq \mn{con}(\Tmc)$ such that for all $B_1,B_2 \in 
\mn{con}(\Tmc)$:
\begin{itemize}

\item if $B_1\in t$ and $\tbox\models B_1\sqsubseteq B_2$, then $B_2 \in t$;
\item if $B_1\in t$ and $\tbox\models B_1\sqsubseteq \neg B_2$, then $B_2\notin t$.

\end{itemize}
A \emph{$\tbox$-typing} is a set $T$ of \Tmc-types.  
A \emph{path in
  $T$} is a sequence $t,r_1,\ldots,r_n$ where $t \in T$,
$\existsr{r_1}{},\dots,\existsr{r_{n}}{} \in \mn{con}(\tbox)$ use no
predicates from $\Sigma_{\Csf}$, $\existsr{r_1}{}\in t$ and for
$i\in\{1,\ldots,n-1\}$, $\mathcal{T} \models \exists r_i^- \sqsubseteq
\exists r_{i+1}$ and \mbox{$r_i^- \ne r_{i+1}$}. The path is
\emph{$\Sigma_{\Csf}$-participating} if for all $i \in \{1,\dots,n-1\}$,
there is no $B\in \mn{con}(\tbox)$ with $\sig{B}{}\subseteq\Sigma_{\Csf}$ and
$\tbox\models\existsr{r_i^-}{}\sqsubseteq B$ while there is such a $B$
for $i=n$.  A $\Tmc$-typing $T$ is \emph{$\Sigma_{\Csf}$-realizable} if for
every $\Sigma_{\Csf}$-participating path $t,r_1,\ldots,r_n$ in $T$, there is
some $u\in T$ such that $\{B\in \mn{con}(\tbox)\mid
\tbox\models\existsr{r^-_n}{}\sqsubseteq B\}\subseteq u$.

A \Tmc-typing $T$ provides partial information about a model \Imc of
\Tmc and a $\Sigma_{\Csf}$-ABox \Amc by taking $T$ to contain the
types that are realized in \Imc by ABox
elements. $\Sigma_{\Csf}$-realizability then ensures that we can build
from $T$ a model that respects the closed predicates
in~$\Sigma_{\Csf}$. To make this more precise, define a
\emph{$\tbox$-decoration of} a $\Sigma_{\Csf}$-ABox $\abox$ to be a
mapping $f$ that assigns to each $a\in \mn{Ind}(\Amc)$ a \Tmc-type
$f(a)$ such that $f(a)|_{\Sigma_{\Csf}} = t^a_\abox|_{\Sigma_{\Csf}}$
where $t^a_\abox=\{B\in \mn{con}(\tbox)\mid a \in B^{\Imc_\Amc}\}$ and
$S|_{\Sigma_{\Csf}}$ denotes the restriction of the set $S$ of
concepts to those members that only use predicates
from~$\Sigma_{\Csf}$.
%  For brevity, let
%$R_{\Sigma_{\Csf}}=\{s\sqsubseteq r\mid \tbox\models s\sqsubseteq r
%\text{ and } \mn{sig}(s\sqsubseteq r)\subseteq\Sigma_{\Csf} \}$.
  %
% \{B\in t_\abox^a\mid \sig{B}{}\subseteq\Sigma\}=\{B\in f(a)\mid
% \sig{B}{}\subseteq\Sigma\} $
% where %for each $a\in\adom{\abox}$, we denote by
%
%
The following lemma is proved in the appendix.
\begin{lem}
  \label{lem:deco}
  \mbox{A $\Sigma_{\Csf}$-ABox $\abox$ is consistent w.r.t.~$(\Tmc,\Sigma_{\Csf})$ iff}
  \begin{enumerate}
  \item $\abox$ has a $\tbox$-decoration $f$ whose image is a
    $\Sigma_{\Csf}$-realizable $\tbox$-typing and
  \item if $s(a,b)\in\abox$, $\tbox\models s\sqsubseteq r$, and $\mn{sig}(s\sqsubseteq r)\subseteq\Sigma_{\Csf}$,
     then $r(a,b)\in\abox$.
  \end{enumerate}
\end{lem}
  We now construct the required FOQ. For all role names $r$ and
  variables $x,y$, define $\psi_r(x,y)=r(x,y)$ and
  $\psi_{r^-}(x,y)=r(y,x)$.  For all concept names $A$ and roles $r$,
  define $\psi_A(x)=A(x)$ and $\psi_{\exists r}(x)=\exists
  y\,\psi_r(x,y)$. For each \Tmc-type $t$, set
  $$\psi_t(x)=\hspace*{-5mm}
  \bigwedge_{B\in \mn{con}(\tbox)\setminus t \text{ with
    }\sig{B}{}\subseteq\Sigma_{\Csf} }\neg\psi_B(x)\land \hspace*{-4mm}
  \bigwedge_{B\in t\text{ with }\sig{B}{}\subseteq\Sigma_{\Csf}}\psi_B(x)$$
  and for each \Tmc-typing $T=\{t_1,\dots,t_n\}$, set
  $$\psi_T=\forall x\bigvee_{t\in T}\psi_{t}(x)\land \exists
  x_1\cdots\exists x_n(\bigwedge_{i\neq j} x_i \neq x_j\land
  \bigwedge_i\psi_{t_i}(x_i) ).
  $$
  Let $\Rmc$ be the set of all $\Sigma_{\Csf}$-realizable typings and set
  $$
  \Psi_{\Tmc,\Sigma_{\Csf}} = \bigvee_{T\in\Rmc}\psi_T \wedge
  \bigwedge_{\tbox\models s\sqsubseteq r, \mn{sig}(s\sqsubseteq r)\subseteq\Sigma_{\Csf}}\forall x\forall
  y(\psi_s(x,y)\rightarrow \psi_r(x,y)).
  $$
  Note that the two conjuncts of $\Psi_{\Tmc,\Sigma_{\Csf}}$ express exactly
  Points~(1) and~(2) of Lemma~\ref{lem:deco}.  We have thus shown
  the following.
  \begin{prop}
    % $\Psi_{\Tmc,q}$ is an FO-rewriting of $Q$ regarding ABox
    % consistency.
	A $\Sigma_{\Csf}$-ABox $\abox$ is consistent w.r.t.\ $(\tbox,\Sigma_{\Csf})$ iff 
	$\inter_\abox \models \Psi_{\Tmc,\Sigma_{\Csf}}$.
  \end{prop}
  The next step is to construct an FO-rewriting of 
  $Q=(\tbox,\Sigma_{\Csf},\Sigma_{\Csf},q)$ over $\Sigma_{\Csf}$-ABoxes that are 
  consistent w.r.t.~$(\Tmc,\Sigma_{\Csf})$. Whereas the FO-rewriting $\Psi_{\Tmc,\Sigma_{\Csf}}$ above is Boolean and
  identifies ABoxes that have a common model with \Tmc respecting closed predicates $\Sigma_{\Csf}$, we now aim to
  construct a FOQ $\Phi_{Q}(\vec{x})$ such that for all
  $\Sigma_{\Csf}$-ABoxes \Amc consistent w.r.t.~$(\Tmc,\Sigma_{\Csf})$ and $\vec{a} \in \mn{Ind}(\Amc)$, we have
  $\Imc_\Amc \models \Phi_{Q}(\vec{a})$ iff $\Amc \models Q(\vec{a})$. The desired FO-rewriting
  of $Q$ is then constructed as $\neg\Psi_{\Tmc,\Sigma_{\Csf}} \vee \Phi_Q(\vec{x})$.
  The construction of $\Phi_{Q}(\vec{x})$ is
  based on an extended notion of \Tmc-typing called
  \emph{$(\Tmc,q)$-typing} that provides partial information about a
  model \Imc of
  \Tmc and a $\Sigma_{\Csf}$-ABox \Amc respecting $\Sigma_{\Csf}$ which avoids an assignment from
  $\vec{x}$ to certain individual names~$\vec{a}$.
  
  Let $q=\bigvee_{i\in I}q_{i}$ with answer variables
  $\vec{x}=x_{1},\ldots,x_{n}$.
  %	Assume $\pi$ assigns individual names $\pi(x_{i})$ to $x_{i}$,
  % for $1\leq i \leq n$.
  A $(\Tmc,q)$-\emph{typing} $T$ is a quadruple
  $(\sim,f_{0},\Gamma,\Delta)$ where
  \begin{itemize}
  \item $\sim$ is an equivalence relation on $\{x_{1},\ldots,x_{n}\}$;
  \item $f_{0}$ is a function that assigns a \Tmc-type $f_{0}(x_{i})$
    to each $x_{i}$, $1\leq i \leq n$, such that
    $f_{0}(x_{i})=f_{0}(x_{j})$ when $x_{i}\sim x_{j}$;
  \item $\Gamma$ is a \Tmc-typing;
  \item $\Delta$ is a set of atoms $s(x_{i},x_{j})$, $s\in \Sigma_{\Csf}$,
    such that $s(x_{i},x_{j})\in \Delta$ iff $s(x_{i}',x_{j}')\in
    \Delta$ when $x_{i}\sim x_{i}'$ and $x_{j}\sim x_{j}'$.
  \end{itemize}
  Intuitively, $\sim$ describes the answer variables that are
  identified by an assignment $\pi$ for $q$ in an ABox $\Amc$, $f_{0}(x_{i})$
  describes the $\Tmc$-type of the ABox individual name $\pi(x_{i})$,
  $\Gamma$ describes the $\Tmc$-types of ABox individual names that are not
  in the range of $\pi$, and $\Delta$ fixes role relationships that do
  \emph{not} hold between the $\pi(x_{i})$. Let $X$ be a set of atoms.
%=\{\alpha_{i} \mid
%  i\in I\}$ 
%be a set of atoms with $\alpha_{i}$ in $q_{i}$ for all
%  $i\in I$.  
Then $T$ \emph{avoids $X$} if the following conditions hold:

  \smallskip
  \noindent
  1. for all $x_{i}$, $1\leq i \leq n$, if $A\in f_{0}(x_{i})$, then $A(x_{i})\not\in X$;

  \smallskip
  \noindent
  2. for all $x_{i}$, $1\leq i \leq n$, if $\exists s\in f_{0}(x_{i})$, then for $S=\{B\in\mn{con}(\tbox)\mid
  \tbox\models\existsr{s^-}{}\sqsubseteq B\}$ the following holds: (i) $S$ contains no predicate
  from $\Sigma_{\Csf}$ or (ii) there is a $u \in \Gamma$ such that
  $S\subseteq u$ or (iii) there is a $y$ such that $S\subseteq
  f_{0}(y)$ and there are no $x'\sim x_{i}$ and $y' \sim y$ such that
  $r(x',y') \in X$ and $\Tmc\models s\sqsubseteq r$, or 
  $r(y',x') \in X$ and $\Tmc\models s \sqsubseteq r^{-}$;
  % (y',x') \in X$ and $\Tmc\models s\sqsubseteq r^{-}$;

  \smallskip
  \noindent
  3. if $r(x,y)\in X$, then $\Delta$ contains all $s(x,y)$ with $s\in
  \Sigma_{\Csf}$ and $\Tmc\models s \sqsubseteq r$ and all $s(y,x)$ with
  $s\in \Sigma_{\Csf}$ and $\Tmc\models s^{-}\sqsubseteq r$.

  \smallskip
  \noindent
  $T$ \emph{avoids} $q$ if it avoids some set $X$ of
  atoms containing an atom $\alpha_i$ in $q_i$ for any $i\in I$.
  We use $\mn{tp}(T)$ to denote the \Tmc-typing $\Gamma$ extended with all
  $\Tmc$-types in the range of $f_0$.  Let $\Amc$ be a $\Sigma_{\Csf}$-ABox
  and let $\pi$ assign individual names $\pi(x_{i})$ to $x_{i}$,
  $1\leq i \leq n$, such that $\pi(x_{i})=\pi(x_{j})$ iff $x_{i}\sim
  x_{j}$.  A $\Tmc$-decoration $f$ of $\Amc$ \emph{realizes}
  $T=(\sim,f_{0},\Gamma,\Delta)$ \emph{using} $\pi$ iff $\mn{tp}(T)$
  is the range of $f$, $f_{0}(x_{i})=f(\pi(x_{i}))$ for $1\leq i \leq
  n$, and $r(\pi(x_{i}),\pi(x_{j}))\not\in \Amc$ if $r(x_{i},x_{j})\in
  \Delta$ for $1\leq i,j\leq n$ and all $r\in \Sigma_{\Csf}$.  $\Amc$
  \emph{realizes} $T$ using $\pi$ if there exists a $\Tmc$-decoration
  $f$ that realizes $T$ using~$\pi$.

  \begin{lem}
    \label{lem:deco2}
	Let $\Amc$ be a $\Sigma_{\Csf}$-ABox consistent 
	w.r.t.~$(\Tmc,\Sigma_{\Csf})$. Then $\Amc \not\models 
	Q(\pi(x_{1}),\ldots,\pi(x_{n}))$ iff $\Amc$ realizes some 
	$(\Tmc,q)$-typing $T$ using $\pi$ that avoids $q$ and such that 
	$\mn{tp}(T)$ is $\Sigma_{\Csf}$-realizable.
  \end{lem}
The proof is a modification of the proof of Lemma~\ref{lem:deco} and given in the
appendix.

We now construct the actual rewriting $\Phi_{Q}(\vec{x})$. For every
$(\Tmc,q)$-typing $T=(\sim,f_{0},\Gamma,\Delta)$ with
$\Gamma=\{t_{1},\ldots,t_{k}\}$
% that avoids a sequence $\vec{\alpha}$ of atoms $\alpha_{i}\in
% q_{i}$, $i\in I$, let $\psi_{T}=\psi_{=}\wedge$ $\Sigma$-Abox $\Amc$
% let and ABox any equivalence relation $\sim$ on
% $\{x_{1},\ldots,x_{n}\}$ we denote by $q^{\sim}$ the result of
% identifying all $x_{i}\sim x_{j}$ in $q$. We define a rewriting
% $q_{\Tmc}^{\sim}$ for each $q^{\sim}$. For each assignment $\pi$ for
% $q_{\Tmc}^{\sim}$ adding to $q$ the conjunct $x_{j}=x_{k}$ if
% $x_{j}\sim x_{k}$ and the conjunct $\neg (x_{j}=x_{k)$ if
% $x_{j}\not\sim x_{k}$, where $j,k\leq n$. For each $\Tmc,q$-typing
% and $\Tmc,q$-rewriting $T=(f_0,\Gamma)$ for $\pi$ with
% $\Gamma=\{t_1,\ldots,t_n\}$, define
  %
let $\Psi_{T}(\vec{x})$ be the conjunction of the following:
$$
\bigwedge_{1\leq i \leq n}\psi_{f_0(x_{i})}(x_{i}) \wedge
\bigwedge_{x_{i}\sim x_{j}}(x_{i}=x_{j}) \wedge
\bigwedge_{x_{i}\not\sim x_{j}}(x_{i}\not=x_{j})
$$
$$ 
\bigwedge_{r(x_{i},x_{j})\in \Delta}\neg r(x_{i},x_{j}) \wedge \forall
y (\bigwedge_{1\leq i \leq n}(y \not=x_{i})\rightarrow \bigvee_{t\in
  \Gamma}\psi_{t}(y))
$$
$$ 
\exists y_1\cdots\exists y_k(\bigwedge_{j\neq i} y_j \neq y_i \wedge
\bigwedge_{j\leq k, i\leq n} x_i \neq y_j\wedge \bigwedge_{j \leq k}
\psi_{t_j}(y_j))
$$
Then $\Phi_{Q}(\vec{x})$ is the conjunction over all $\neg
\Psi_{T}(\vec{x})$ such that $T$ avoids $q$ and $\mn{tp}(T)$ is
$\Sigma_{\Csf}$-realizable.

\begin{prop}\label{prop:fo}
  Let $\Amc$ be a $\Sigma_{\Csf}$-ABox that is consistent 
  w.r.t.~$(\Tmc,\Sigma_{\Csf})$. Then $\Amc \models Q(\vec{a})$ iff 
  $\inter_\abox \models \Phi_{Q}(\vec{a})$, for all $\vec{a}$ in $\mn{Ind}(\Amc)$.
\end{prop}
\begin{proof}
  Let $\vec{a}=(a_{1},\ldots,a_{n})$. Assume $\Amc \not\models Q(a_{1},\ldots,a_{n})$. Let 
  $\pi(x_{i})=a_{i}$ for $1\leq i \leq n$. By Lemma~\ref{lem:deco2}, 
  $\Amc$ realizes some $(\Tmc,q)$-typing $T$ using $\pi$ that avoids $q$ 
  such that ${\sf tp}(T)$ is $\Sigma_{\Csf}$-realizable. It is readily checked 
  that $\Imc_{\Amc}\models \Psi_{T}(\pi_{1}(x_{1}),\ldots,\pi(x_{n}))$.  
  Thus, $\inter_\abox \not\models \Phi_{Q}(a_{1},\ldots,a_{n})$

  Conversely, assume that $\inter_\abox \not\models
  \Phi_{Q}(a_{1},\ldots,a_{n})$.  Take a $(\Tmc,q)$-typing $T$ that avoids $q$ such that ${\sf
    tp}(T)$ is $\Sigma_{\Csf}$-realizable and $\inter_\abox \models
  \Psi_{T}(a_{1},\ldots,a_{n})$. Let $\pi(x_{i})=a_{i}$ for $1\leq i 
  \leq n$.  It is readily checked that $\Amc$ realizes $T$ using $\pi$. 
  Thus $\Amc \not\models Q(a_{1},\ldots,a_{n})$, by 
  Lemma~\ref{lem:deco2}.
\end{proof}
This finishes the proof of Theorem~\ref{thm:rewr}.
We now show that without the restriction on RIs adopted in
Theorem~\ref{thm:rewr}, OMQCs from $(\dlliter,\NC \cup
\NR,\text{UqfCQ})$ are no longer FO-rewritable. In fact, we prove the
following, slightly stronger result by reduction from propositional
satisfiability.
\begin{thm}
  \label{thm:nphard} There is a \dlliter TBox with closed predicates $(\Tmc,\Sigma_{\Csf})$ 
  such that ABox consistency w.r.t.\ $(\Tmc,\Sigma_{\Csf})$ is \np-complete.
\end{thm}
\begin{proof}
  The proof is by reduction of the satisfiability problem for
  propositional formulas in conjunctive normal form (CNF).  Consider a
  propositional formula in CNF $\varphi=c_{1} \wedge \cdots \wedge
  c_{n}$, where each $c_{i}$ is a disjunction of literals. We write
  $\ell\in c_{i}$ if $\ell$ is a disjunct in $c_{i}$.
  % $l_{i}^{1},\ldots,l_{i}^{k(i)}$.
  Let $x_{1},\ldots, x_{m}$ be the propositional variables in
  $\varphi$.  Define an ABox $\Amc_{\varphi}$ with individual names
  $c_{1},\ldots,c_{n}$ and $x_{i}^\top$, $x_{i}^{\bot}$, $x_{i}^{{\sf
      aux}}$ for $1\leq i \leq m$, a concept name $A$, and role names
  $r,r'$ as the following set of assertions:
  \begin{itemize}
  \item $r(c_{i},x_{j}^{\top})$, for all $x_{j}\in c_{i}$ and $1\leq i
    \leq n$;
  \item $r(c_{i},x_{j}^{\bot})$, for all $\neg x_{j}\in c_{i}$ and
    $1\leq i \leq n$;
  \item $r'(x_{j}^{\top},x_{j}^{\bot})$, $r'(x_{j}^{\bot},x_{j}^{{\sf
        aux}})$, for $1\leq j \leq m$;
  \item $A(c_{i})$, for $1\leq i \leq n$.
  \end{itemize}
  Let $s$ and $s'$ be additional role names and let 
$$
\Tmc = \{s \sqsubseteq r, A \sqsubseteq \exists s, \exists s^{-} \sqsubseteq \exists s', s' \sqsubseteq r',
  \exists s'^{-} \sqcap \exists s^{-} \sqsubseteq \bot\}.
$$
  Let $\Sigma_{\Csf} = \{A,r,r'\}$. We show that $\Amc_{\varphi}$ is
  consistent w.r.t.~$(\Tmc,\Sigma_{\Csf})$ iff $\varphi$ is
  satisfiable.  Assume first that $\Amc_{\varphi}$ is consistent
  w.r.t.~$(\Tmc,\Sigma_{\Csf})$. Let $\Imc$ be a model of
  \Tmc and $\Amc_{\varphi}$ that respects closed predicates
  $\Sigma_{\Csf}$.  Define a propositional valuation $v$ by setting
  $v(x_{j})= 1$ if there exists $i$ such that
  $(c_{i},x_{j}^{\top})\in s^{\Imc}$ and set $v(x_{j})=0$ if there exists
  $i$ such that $(c_{i},x_{j}^{\bot})\in s^{\Imc}$. Observe that $v$ is
  well-defined since if $(c_{i},x_{j}^{\top})\in s^{\Imc}, (c_{k},x_{j}^{\bot})
  \in s^{\Imc}$, then $(x_{j}^{\top},x_{j}^{\bot})\in s'^{\Imc}$ and so
  $x_{j}^{\bot}\in (\exists s'^{-} \sqcap \exists s^{-})^{\Imc}$ which
  contradicts the assumption that $\Imc$ satisfies $\exists s'^{-}
  \sqcap \exists s^{-} \sqsubseteq \bot$. Next observe that for every
  $c_{i}$ there exists a disjunct $\ell\in c_{i}$ such that
  $(c_{i},x_{j}^{\top})\in s^{\Imc}$ if $\ell= x_{j}$ and
  $(c_{i},x_{j}^{\bot})\in s^{\Imc}$ if $\ell=\neg x_{j}$. Thus,
  $v(\varphi)=1$ and $\varphi$ is satisfiable.

  Conversely, assume that $\varphi$ is satisfiable and let $v$ be an
  assignment with $v(\varphi)=1$. Define an interpretation $\Imc$ by expanding $\Imc_{\Amc_{\varphi}}$ as follows:
  \begin{eqnarray*}
    s^{\Imc}  & =  & \{ (c_{i},x_{j}^{\top})\} \mid x_{j}\in c_{i}, v(x_{j})=1,i\leq n\} \cup  \{ (c_{i},x_{j}^{\bot}) \mid \neg x_{j}\in c_{i}, v(x_{j}) =0,i\leq n\}\\
    s'^{\Imc}  & =  & \{ (x_{j}^{\top},x_{j}^{\bot}) \mid v(x_{j})=1\} \cup \{ (x_{j}^{\bot},x_{j}^{{\sf aux}}) \mid v(x_{j})=0\}
  \end{eqnarray*}
  It is readily checked that $\Imc$ is a model of \Tmc and
  $\Amc_{\varphi}$ that respects closed predicates $\Sigma_\Csf$.
\end{proof}
We close this section with noting that, for the case of \EL, quantifier-free queries are computationally no more
well-behaved than unrestricted queries. In fact, we have seen that OMQCs in $\mathcal{EL}$ using dtUCQs can be
equivalently expressed using atomic database queries $A(x)$ by adding CIs of the form $C_{q}\sqsubseteq A$ to
the TBox.

\section{Conclusion}
We have investigated the data complexity of ontology-mediated query
evaluation with closed predicates, focussing on a non-uniform
analysis. At the TBox level we have obtained \ptime/\conp dichotomy
results for the lightweight DLs $\mathcal{EL}$ and \dlliter. At the
query level, the situation is drastically different: there is provably
no \ptime/\conp dichotomy for neither \dlliter nor $\mathcal{EL}$
(unless \ptime=\conp) and even without closing role names,
understanding the complexity of queries is as hard as understanding
the complexity of the generalized surjective constraint satisfaction
problems.  We have also shown that by combining \dlliter with
quantifier-free database queries one obtains FO-rewritable queries and
that even for expressive DLs query evaluation is always in \conp. Many
challenging open questions remain.

Regarding the data complexity classification at TBox level, it is shown 
in~\cite{Lutz:2013:ODA:2540128.2540276} that the dichotomy proof given for \dlliter and $\mathcal{EL}$ does 
not go through for the extension $\mathcal{ELI}$ of $\mathcal{EL}$ with inverse roles. In fact, in 
contrast to \dlliter and $\mathcal{EL}$, there are $\mathcal{ELI}$ TBoxes
with closed predicates $(\Tmc,\Sigma_{\Csf})$ such that CQ evaluation w.r.t.~$(\Tmc,\Sigma_{\Csf})$ is in \ptime, 
but $(\Tmc,\Sigma_{\Csf})$ and $(\Tmc,\emptyset)$ are not CQ-inseparable on consistent ABoxes.
In particular, it remains open whether there is a \ptime/\conp dichotomy for TBoxes with closed predicates
in $\mathcal{ELI}$. The same question remains open for $\mathcal{ALCHI}$ TBoxes 
(recall that there is a \ptime/\conp dichotomy for for $\mathcal{ALCHI}$ TBoxes without closed 
predicates~\cite{DBLP:journals/lmcs/LutzW17,DBLP:conf/pods/HernichLPW17}) and for expressive Horn languages
such as Horn-$\mathcal{SHIQ}$. Also of interest are ontologies consisting of tuple-generating
dependencies (tgds) which generalizes both DL-Lite$_{\mathcal{R}}$ and $\mathcal{EL}$.
In this case, however, the \conp upper bound established here for $\mathcal{ALCHI}$ does not 
hold, even for the moderate extension consisting of linear tgds~\cite{DBLP:conf/lics/BenediktBCP16,Benedikt}.
  
Regarding the data complexity classification at the OMQC level, it would be of interest to
consider $\text{\dllite}_{\text{core}}$: it remains open whether there is a \ptime/\conp
dichotomy for the language $(\text{\dllite}_{\text{core}},\NC \cup \NR,\text{BUtCQ})$
and whether Theorem~\ref{thm:mmm} can be strengthened to this case.

\bigskip
\noindent
\textbf{Acknowledgments.} Frank Wolter was supported by EPSRC grant
EP/M012646/1.  Carsten Lutz was supported by the ERC Consolidator
Grant 647289 CODA. We thank the anonymous reviewers for their suggestions and    
comments.  

%\bibliographystyle{plain}
%\bibliography{ourbib}

\newpage
\appendix
\section{Missing Proofs for Section~\ref{sect:basicres}}
{\bf Lemma~\ref{lem:alchi_forest_model}}
{\em  Let $\abox$ be a $\Sigma_{\Asf}$-ABox, $\vec{a}$ a tuple in ${\sf Ind}(\Amc)$, and
  $Q=(\Tmc,\Sigma_\Asf,\Sigma_\Csf,q)$ a OMQC from $(\alchi,\NC \cup
  \NR,\text{UCQ})$. Then the following are equivalent:
  \begin{enumerate}
  \item $\abox\models Q(\vec{a})$;
  \item $\inter\models q(\vec{a})$ for all forest-shaped models $\inter$ of $\tbox$ and $\abox$ that respect $\Sigma_\Csf$ and such that
     \begin{itemize}
            \item the arity of $\Delta^{\Imc}$ is $|\Tmc|$, 
            \item ${\sf Ind}(\Amc)$ is the set of roots of $\Delta^{\Imc}$,
            \item for every $d\in \Delta^{\Imc}\setminus {\sf Ind}(\Amc)$ and $\exists r.C\in {\sf cl}(\Tmc)$ with $d\in (\exists r.C)^{\Imc}$, there exists $a\in {\sf Ind}(\Amc)$ with $(d,a)\in r^{\Imc}$ and $a\in C^{\Imc}$ or 
there exists a successor $d'$ of $d$ in $\Delta^{\Imc}$
                  such that $(d,d')\in r^{\Imc}$ and $d'\in C^{\Imc}$.
     \end{itemize} 
  \end{enumerate}
}
\begin{proof}
The implication from (1) to (2) is trivial. 
  % since every forest-shaped model of $\Tmc$ and $\Amc$ is a
  % modelThis is the easy direction and we make a contrapositive
  % argument. Suppose there is some forest-shaped model $\inter$ as
  % specified in the lemma such that $\inter\not\models q$. By the
  % virtue of $\inter$ being a model, we immediately conclude that
  % $\tbox,\abox\not\models_{c(\Sigma_ \Csf)}q$ and we are done.
For the converse direction, suppose $\abox\not\models Q(\vec{a})$.
  Then there is some model $\interj$ of $\tbox$ and $\abox$ that respects closed predicates 
  $\Sigma_\Csf$ such that $\interj\not\models q(\vec{a})$. We construct, by 
  induction, a sequence of interpretations 
  $\interp{_0},\interp{_1},\ldots$. The domain of each $\inter_{i}$ 
  consists of sequences of the form $d_{0}\cdot d_{1}\cdots d_{n}$, 
  where $d_{j}\in\domainj$ for all $j\in\{0,\ldots,n\}$. We call such 
  sequences \emph{paths} and denote the last element in a path $p$ by 
  $\tail{p}$, e.g., $\tail{d_{0}\cdots d_{n}}=d_{n}$.
  
  We define $\inter_0$ as the restriction of $\interj$ to
  $\adom{\abox}$.
  
  Assume now that $\inter_i$ is given. Let $p\in\domainp{_i}$ 
  such that for some $e\in\domainj$ and $\exists r.C\in {\sf cl}(\Tmc)$, we have 
  $(\tail{p},e)\in\extj{r}$ and $e\in C^{\Jmc}$ and there is no $p'\in\domainp{_{i}}$ with 
  $\tail{p'}=e'$ and $(p,p')\in\extp{r}{_{i}}$ and $e'\in C^{\Jmc}$. Assume first that 
  $e\not\in\adom{\abox}$. We extend $\inter_i$ to $\inter_{i+1}$ by setting
\begin{eqnarray*}
\Delta^{\Imc_{i+1}} & = & \Delta^{\Imc_{i}}\cup \{p\cdot e\}\\
s^{\Imc_{i+1}} & = & s^{\Imc_{i}} \cup \{ (p,p\cdot e) \mid (\tail{p},e)\in\extj{s}\}\cup
\{ (p\cdot e,p) \mid (e,\tail{p}) \in\extj{s}\}\\
A^{\Imc_{i+1}} & = & A^{\Imc_{i}} \cup \{ p\cdot e \mid e\in\extj{A}\}
\end{eqnarray*}
for all role names $s$ and concept names $A$. Suppose now that $e=a$ for some $a\in\adom{\abox}$. In this case, we 
  extend $\inter_i$ to $\inter_{i+1}$ by adding the tuple $(p,e)$ to 
  $\extp{s}{_i}$, for every role $s$ such that 
  $(\tail{p},e)\in\extj{s}$.
  
  We assume that the above construction is \emph{fair} in the sense that 
  if the conditions of the inductive step are satisfied for some 
  $p\in\domainp{_{i}}$, $e\in\domainj$, and $\exists r.C\in {\sf cl}(\Tmc)$, with $i\geq 0$, 
  then there is some $j > i$ such that the inductive step is applied to 
  $p$, $e$, and $\exists r.C$.
  
  Now we define the interpretation $\inter$ as the limit of the sequence $\Imc_{0},\Imc_{1},\ldots$:
  \begin{itemize}
  \item $\domain=\bigcup_{i\geq 0}\domainp{_i}$;
  \item $\ext{P}=\bigcup_{i\geq 0}\extp{P}{_i}$, for all
    $P\in\conceptnames\cup\rolenames$.
  \end{itemize}

  It is clear that $\inter$ is a forest-shaped interpretation with 
  $\domain$ a $\length{\tbox}$-ary forest having precisely 
  $\adom{\abox}$ as its roots. That $\inter$ is a model of $\abox$ is 
  an easy consequence of the facts that $\interj$ is a model of 
  $\abox$, $\inter_{0}$ is the restriction of $\interj$ to 
  $\adom{\abox}$, and $\inter$ is an extension of $\inter_{0}$. That 
  $\inter$ respects closed predicates $\Sigma_\Csf$ is by definition. 
  We now show that $\inter$ is a model of $\tbox$. The following is easily proved by structural induction.
  
  \begin{claim*}
    For all $p\in\domain$ and $C\in {\sf cl}(\Tmc)$, $p\in\ext{C}
    \text{ iff } \tail{p}\in\extj{C}$.
  \end{claim*}
% \begin{clmproof}
%     The proof is by structural induction. The base case follows 
%     immediately by item~1 in Claim~1 and the boolean cases are trivial. 
%     Therefore, we only consider the case where $C=\existsr{r}{D}$.

%     Suppose $p\in\ext{(\existsr{r}{D})}$. Then there is some 
%     $p'\in\domain$ such that $(p,p')\in\ext{r}$ and $p'\in\ext{D}$. By 
%     the former and item~2 in Claim~1, we have 
%     $(\tail{p},\tail{p'})\in\extj{r}$; and by the latter and the 
%     induction hypothesis, we have $\tail{p'}\in\extj{D}$. Hence 
%     $\tail{p}\in\extj{(\existsr{r}{D})}$, as required. 
    
%     For the other direction, suppose 
%     $\tail{p}\in\extj{(\existsr{r}{D})}$. Then there is some 
%     $e\in\domainj$ such that $(\tail{p},e)\in\extj{r}$ and 
%     $e\in\extj{D}$. By item~3 in Claim~1, we then have some 
%     $p'\in\domain$ with $\tail{p'}=e$ and $(p,p')\in\ext{r}$. By 
%     $\tail{p'}=e$, $e\in\extj{D}$, and the induction hypothesis, we 
%     obtain $p'\in\ext{D}$. Hence $p\in\ext{(\existsr{r}{D})}$.
%   \end{clmproof}  
The fact that $\Jmc$ is a model of $\Tmc$ now implies that 
$\inter$ is a model of every CI in $\tbox$. That 
$\inter$ is a model of every RI in $\tbox$ follows by construction.
Hence we conclude that $\inter$ is a model of $\tbox$.

Finally, to show that $\inter\not\models q(\vec{a})$, observe that 
$h=\{p\mapsto\tail{p}\mid p\in\domain\}$ is a homomorphism from $\inter$ to 
$\interj$ preserving $\NI$. Thus, $\inter\not\models q(\vec{a})$ follows from Lemma~\ref{homo}
and $\interj\not\models q(\vec{a})$.
\end{proof}

\begin{lem}
The interpretation $\Imc$ defined in the proof of Lemma~\ref{lem:alchi_mosaic} is a model of $\tbox$ and $\abox$ 
that respects closed predicates $\Sigma_\Csf$ such that $\inter\not\models q(\vec{a})$.
\end{lem}
\begin{proof}
The following conditions follow directly from the construction of $\inter$ and the conditions on mosaics:
    \begin{itemize}
    \item $\inter$ is a model of $\abox$;
    \item $\inter$ is a model of every RI in $\tbox$;
    \item $\ext{P}=\{\vec{a}\mid P(\vec{a})\in\abox\}$, for all predicates $P\in \Sigma_{\Csf}$.
    \end{itemize}
    It remains to show that $\inter$ is a model of every concept
    inclusion in $\tbox$. Define for every $d\in\domain$, a
    $\tbox$-type $t_d$ as follows.
    \begin{itemize}
    \item if $d\in\adom{\abox}$, then let $t_d=\tau(d)$ for some
      $(\interj,\tau)\in M$;
    \item if $d\in\domain\setminus\adom{\abox}$, then $t_d=\tau_d(d)$.
    \end{itemize}
    To prove that $\Imc$ is a model of $\Tmc$ it is now sufficient
    to show the following: for all $d\in\domain$ and
    $C\in\cclos{\tbox}{}$, $d\in\ext{C}$ iff $C\in t_d$. The proof is
    by structural induction.

    Let $C=A\in\conceptnames$. If $d\in\adom{\abox}$, let
    $(\interj,\tau)$ be any mosaic in $M$; and if
    $d\in\domain\setminus\adom{\abox}$, then let
    $(\interj,\tau)=(\inter_d,\tau_d)$. We have (i) $d\in\ext{B}$ iff
    $d\in \extj{B}$ for all $B\in\conceptnames\cap\cclos{\tbox}{}$ and
    (ii) $\tau(d)=t_d$. But then $d\in\ext{A}$ iff $d\in\extj{A}$ (by
    (i)) iff $A\in\tau(d)$ (by the definition of a mosaic) iff $A\in
    t_d$ (by (ii)).

    The boolean cases follow easily by the induction hypothesis and
    the fact that $t_d$ is a $\tbox$-type.

    Let $C=\existsr{r}{D}$. For the direction from left to right,
    suppose $d\in\ext{(\existsr{r}{D})}$. Then there is some
    $e\in\domain$ such that $(d,e)\in\ext{r}$ and $e\in\ext{D}$. If
    $d,e\in\adom{\abox}$, let $(\interj,\tau)$ be any mosaic in $M$;
    if $d,e\in\domain\setminus\adom{\abox}$, let
    $(\interj,\tau)=(\inter_{d'},\tau_{d'})$, where $d'$ is the
    element of $\{d,e\}$ that has the smaller depth in
    $\domain\setminus\adom{\abox}$; otherwise let
    $(\interj,\tau)=(\inter_{d'},\tau_{d'})$, where $d'$ is the only
    element of $(\domain\setminus\adom{\abox})\cap\{d,e\}$. Observe
    that $(d,e)\in\extj{r}$, $\tau(d)=t_d$, and $\tau(e)=t_e$. By
    $(d,e)\in\extj{r}$ and the definition of a mosaic, we obtain
    $\tau(d)\rightsquigarrow_r\tau(e)$ and by the induction hypothesis
    and $\tau(e)=t_e$, we obtain $D\in\tau(e)$. But then
    $\existsr{r}{D}\in\tau(d)$ and thus, $\existsr{r}{D}\in t_d$,
    which is what we wanted to show.

    For the direction from right to left, suppose $\existsr{r}{D}\in
    t_d$. We distinguish between $d\in\adom{\abox}$ or not. For the
    former case, we find by the coherency of $M$ a $(\interj,\tau)\in
    M$ such that for some $e\in\domainj$ we have $(d,e)\in\extj{r}$
    and $C\in\tau(e)$; for the latter case, we have by the definition
    of a mosaic and $\length{q}\geq 1$ that there is some
    $e\in\domainp{_d}$ with $(d,e)\in\extp{r}{_d}$ and
    $C\in\tau_d(e)$. In both cases, we have by the construction of
    $\inter$ that $(d,e)\in\ext{r}$ and by definition that $C\in
    t_e$. By the latter, the induction hypothesis yields
    $e\in\ext{C}$. Hence, $d\in\ext{(\existsr{r}{D})}$, as required.

\smallskip

  It remains to show that $\inter\not\models q(\vec{a})$. Assume $\vec{a}=(a_{1},\ldots,a_{n})$. For a proof by contradiction, 
  suppose that $\inter\models q(\vec{a})$. Then there is a disjunct 
  $\exists \vec{y}\varphi(\vec{x},\vec{y})$ of $q$ with $\vec{x}=(x_{1},\ldots,x_{n})$ and $\varphi$ a conjunction of
  atoms such that there is an assignment $\pi$ mapping the variables $\vec{x}\cup\vec{y}$ of $\varphi$ to $\Delta^{\Imc}$ 
  with $\pi(x_{i})=a_{i}$ for $1\leq i \leq n$ and $\inter \models_{\pi} \varphi$. Let $F=\{ \pi(x) \mid \pi (x)\not\in\adom{\abox}\}$.
  As $\inter$ is forest-shaped there are $T_1,\ldots,T_m$ with $F=T_{1}\cup \cdots \cup T_{m}$
  such that $T_{1},\ldots,T_{m}$ are maximal and pairwise disjoint trees in $F$.
  % By maximality, we mean that for all $i\in\{1,\ldots,n\}$ and $d\in
  % F$, if $d\not\in T_i$ then $T_i\cup\{d\}$ is not a tree.
  Fix an $i\in\{1,\ldots,m\}$. Let $d$ be the root of $T_i$. By the
  construction of $\inter$, there is an isomorphism $f_{i}$ trom $(\inter_d,\tau_d)$ to some $(\interj,\tau)\in M$. 
  Let $\pi_{i}$ be the restriction of $\pi$ to those variables that are mapped to $T_{i}$, and let $\pi_{\Amc}$ be the restriction of $\pi$
  to those variables that are mapped to $\mn{Ind}(\Amc)$. Define
  $\pi_{i}'=f_i\circ \pi_{i}$ and then
  $$
  \pi'=\bigcup^m_{i=1}\pi_i'\cup \pi_{\Amc}.
  $$
  $\pi'$ is an assignment in $\biguplus_{(\interj,\tau)\in M}\interj$ with $\pi'(x_{i})=a_{i}$ for $1\leq i \leq n$
  such that $\biguplus_{(\interj,\tau)\in M}\interj\models_{\pi'} \varphi$, and so we have derived a contradiction.
\end{proof}
\section{Missing Proofs for Section~\ref{sect:fullq}}
{\bf Lemma~\ref{lem:deco}}
{\em
  \mbox{A $\Sigma_{\Csf}$-ABox $\abox$ is consistent w.r.t.~$(\Tmc,\Sigma_{\Csf})$ iff}
  \begin{enumerate}
  \item $\abox$ has a $\tbox$-decoration $f$ whose image is a
    $\Sigma_{\Csf}$-realizable $\tbox$-typing and
  \item if $s(a,b)\in\abox$, $\tbox\models s\sqsubseteq r$, and $\mn{sig}(s\sqsubseteq r)\subseteq\Sigma_{\Csf}$,
     then $r(a,b)\in\abox$.
  \end{enumerate}
}
\begin{proof}
  $(\Rightarrow)$ Let $\inter$ be a model of $\abox$ and $\Tmc$ that
  respects closed predicates $\Sigma_{\Csf}$. For each $d\in \Delta^\Imc$,
  let $t_\inter^d=\{B\in \mn{con}(\tbox)\mid d \in B^\Imc \}$ and let
  $T_\inter=\{t_\inter^a\mid a\in\adom{\abox}\}$. We next show that
  the \Tmc-typing  $T_\inter$   is $\Sigma_{\Csf}$-realizable. Let $t_\inter^a,r_1,\ldots,r_n$ be a
  $\Sigma_{\Csf}$-participating path in $T_\inter$. Using $\inter$, we find a
  mapping $g:\{0,\ldots,n\}\rightarrow\domain$ such that $g(0)=a$ and
  for each $i\in\{1,\ldots,n\}$, we have
  \begin{itemize}
  \item[(a)] $(g(i-1),g(i))\in\ext{r_i}$,
  \item[(b)] $g(i)\in\ext{B}$ for all $B\in\subc{\tbox}{}$ with
    $\tbox\models\existsr{r^-_i}{}\sqsubseteq B$.
  \end{itemize}
  By definition of $\Sigma_{\Csf}$-participating paths, there is some
  $B^\star\in \mn{con}(\tbox)$ with $\mn{sig}(B^\star) \subseteq \Sigma_{\Csf}$
  such that $\tbox\models\existsr{r_n^-}{}\sqsubseteq B^\star$. By
  Point~(b), we obtain $g(n)\in\ext{{B^\star}}$. Since $\inter$ is a model
  of $\Amc$ and $\tbox$ that respects closed predicates $\Sigma_{\Csf}$, we
  have $g(n)=b$ for some $b\in\adom{\abox}$.  By Point~(b),
  $\tbox\models\existsr{r_n^-}{}\sqsubseteq B$ implies $B\in
  t_\inter^b$ for any $B\in\mn{con}(\tbox)$. Thus, $T_\inter$ is
  $\Sigma_{\Csf}$-realizable. Let $f(a)=t^a_\Imc$ for all
  $a\in\adom{\abox}$. It is clear that $f$ is a $\tbox$-decoration of
  $\abox$. The image of $f$ is $T_\inter$, thus a $\Sigma_{\Csf}$-realizable
  \Tmc-typing. Hence we conclude that \Amc satisfies Point~(1). Point~(2) holds by the fact that $\inter$ is a
  model of $\tbox$ and $\abox$ that respects closed predicates
  $\Sigma_{\Csf}$.

  $(\Leftarrow)$ Suppose that $\abox$ satisfies Points~(1) and~(2)
  and let $f$ be a \Tmc-decoration of \Amc whose image $T$ is a $\Sigma_{\Csf}$-realizable $\tbox$-typing. 
  Our goal is to construct a model \Imc of $\tbox$ and $\abox$ that respects closed
  predicates $\Sigma_{\Csf}$ as the limit of a sequence of interpretations
  $\Imc_0,\Imc_1,\dots$. The domains of these interpretations consist
  of the individual names from $\mn{Ind}(\Amc)$ and of paths in $T$ that
  are not $\Sigma_{\Csf}$-participating. The construction will ensure that
  for all $i$, we have
  \begin{enumerate}

  \item[(a)] for all $a \in \mn{Ind}(\Amc)$, we have $t^a_{\Imc_i}
    \subseteq f(a)$;

  \item[(b)] for all $p \in \Delta^{\Imc_{i}}$, if $p=t,r_1 \dots, r_n$, then we have
    $t^p_{\Imc_i} \subseteq \{ B \in \mn{con}(\Tmc) \mid \Tmc \models
    \exists r_n^- \sqsubseteq B \}$.

    \end{enumerate}
    Define $\inter_0=(\Delta^{\inter_0},\cdot^{\inter_0})$ where
    \begin{eqnarray*}
      \Delta^{\inter_0} & = & \mathsf{Ind}(\abox)\\
      r^{\inter_0} & = & \{(a,b)\mid s(a,b)\in\abox\text{ and
      }\tbox\models s\sqsubseteq r\} \\ %\text{, }\text{ for all roles } r\\
      A^{\inter_0} & = & \{a\mid A\in f(a)\} %\\
      % p_0(a) &=& f(a). %\text{, for all } A\in\conceptnames
    \end{eqnarray*}
    To construct $\Imc_{i+1}$ from
    $\Imc_i$, % and $p_{i+1}$ from $p_i$,
    choose $d \in \Delta^{\Imc_i}$ and $\exists s \in \mn{con}(\Tmc)$
    such that $\mn{sig}(s) \cap \Sigma_{\Csf} = \emptyset$, $\Tmc \models
    \bigsqcap t^d_{\Imc_i} \sqsubseteq \exists s$ and there is no
    $(d,e) \in s^{\Imc_i}$. Let $q=f(a),s$ if $d=a \in \mn{Ind}(\Amc)$
    and $q=d,s$ otherwise. Using Conditions~(a) and~(b), it is easy to
    verify that $q$ is a path in $T$. If $q$ is not
    $\Sigma_{\Csf}$-participating, then define $\inter_{i+1}$ as follows:
    \begin{eqnarray*}
      \Delta^{\inter_{i+1}} & = & \Delta^{\inter_i}\uplus\{q\}\\
      r^{\inter_{i+1}} & = & \begin{cases}

        r^{\inter_{i}}\cup\{(d,q)\} & \text{if }
        \tbox\models s\sqsubseteq r\\
        r^{\inter_{i}} & \text{otherwise}
      \end{cases}\\
      A^{\inter_{i+1}} & = & \begin{cases}
        A^{\inter_i}\cup\{q\} & \text{if }\tbox\models \existsr{s^-}{}\sqsubseteq A\\
        A^{\inter_i} & \text{ otherwise.}
      \end{cases}
    \end{eqnarray*}
    If $q$ is $\Sigma_{\Csf}$-participating, then by the fact that $T$ is
    $\Sigma_{\Csf}$-realizable, there is some $t\in T$ such that
    $\{B\in \mn{con}(\tbox)\mid \tbox\models\existsr{s^-}{}\sqsubseteq
    B\}\subseteq t$. We find a $b\in\mathsf{Ind}(\abox)$ with
    $t=f(b)$. Define $\inter_{i+1}$ as follows:
    \begin{eqnarray*}
      \Delta^{\inter_{i+1}} & = & \Delta^{\inter_i}\\
      r^{\inter_{i+1}} & = & \begin{cases}
        r^{\inter_{i}}\cup\{(d,b)\} & \text{if }\tbox\models s\sqsubseteq r\\
        r^{\inter_{i}} & \text{otherwise}
      \end{cases}
      \\
      A^{\inter_{i+1}} & = & A^{\inter_i}.
    \end{eqnarray*}
    Assume that the choice of $d \in \Delta^{\Imc_i}$ and $\exists s
    \in \mn{con}(\Tmc)$ is fair so that every possible combination of
    $d$ and $\exists s$ is eventually chosen. Let
    $\inter$ be the limit of the sequence $\Imc_0,\Imc_1,\dots$ (cf. the proof of Lemma~\ref{lem:alchi_forest_model}).
 % We call $\inter$
    % \emph{a chase based on $f$}. Note that the construction of
    % $\inter$ is non-deterministic because for $\Sigma$-participating
    % paths, we choose an individual in $\adom{\abox}$.
    We claim that $\inter$ is a model of $\tbox$ and~$\abox$ that
    respects closed predicates $\Sigma_{\Csf}$. By definition of~$\inter_0$
    and of \Tmc-decorations, it is straightforward to see that
    $\inter\models\abox$. Moreover, the RIs in \Tmc are
    clearly satisfied. To show that the CIs are
    satisfied as well, it is straightforward to first establish the
    following strengthenings of Conditions~(a) and~(b) above (details
    omitted):
    \begin{enumerate}

    \item[(a$'$)] for all $a \in \mn{Ind}(\Amc)$, we have $t^a_{\Imc} =
      f(a)$;

    \item[(b$'$)] for all $p \in \Delta^{\Imc}$, if $p=t,r_1 \dots, r_n$, then
      $t^p_{\Imc_i} =\{ B \in \mn{con}(\Tmc) \mid \Tmc \models
      \exists r_n^- \sqsubseteq B \}$.

    \end{enumerate}
    %
    % Note that in order to establish the first item, we use the fact
    % that
    % $\existsr{r}{},\existsr{r^-}{},\existsr{s}{},\existsr{s^-}{}\in\subc{\tbox}{}$,
    % for a $r\sqsubseteq s\in\tbox$.
    % We now have enough information to
    % show $\inter\models\tbox$.
    %
    Let $a\in\adom{\abox}$, $a\in\ext{B_1}$, and $B_1\sqsubseteq
    B_2\in\tbox$ (or $B_1\sqsubseteq\neg B_2\in\tbox$). Then by
    Condition~$(a')$ and since $f(a)$ is a \Tmc-type, we have
    $a\in\ext{B_2}$ (resp.\ $a\not\in\ext{B_2}$). Now let
    $d=t,r_1,\ldots,r_n$ be a path. First suppose $d\in\ext{B_1}$ and
    $B_1\sqsubseteq B_2\in\tbox$. By Condition~(b$'$), we conclude that
    $\tbox\models\existsr{r_n^-}{}\sqsubseteq B_1$. Since
    $B_1\sqsubseteq B_2\in\tbox$, it follows that
    $\tbox\models\existsr{r_n^-}{}\sqsubseteq B_2$ and thus again by
    the property above, $d\in\ext{B_2}$. Finally, suppose
    $d\in\ext{B_1}$ and $B_1\sqsubseteq \neg B_2\in\tbox$. By
    Condition~(b$'$) and $B_1\sqsubseteq \neg B_2\in\tbox$, we conclude
    $\tbox\models\existsr{r_n^-}{}\sqsubseteq \neg B_2$. For a proof
    by contradiction assume that $d\in\ext{B_2}$ and thus
    $\tbox\models\existsr{r_n^-}{}\sqsubseteq B_2$ and we already have
    $\tbox\models\existsr{r_n^-}{}\sqsubseteq \neg B_2$. Hence
    $\tbox\models\existsr{r^-_n}{}\sqsubseteq \bot$. But then
    $\tbox\models \existsr{r_n}{}\sqsubseteq \bot$. It follows that
    $\tbox\models\existsr{r_1}{}\sqsubseteq\bot$. This implies in
    particular $\tbox\models\existsr{r_1}{}\sqsubseteq\existsr{r_1}{}$
    and
    $\tbox\models\existsr{r_1}{}\sqsubseteq\neg\existsr{r_1}{}$. By
    definition we have $\existsr{r_1}{}\in f(a)$ and by
    $\tbox\models\existsr{r_1}{}\sqsubseteq\neg\existsr{r_1}{}$ and
    the fact that $f(a)$ is a $\Tmc$-type, we obtain
    $\existsr{r_1}{}\not\in f(a)$, i.e., a contradiction. Hence
    $d\not\in\ext{B_2}$ which finishes the proof that
    $\inter\models\tbox$.
    
    What remains to be shown are the following properties:
    \begin{itemize}
    \item for all $A\in\Sigma_{\Csf}$, $A^\inter=\{a\mid A(a)\in\abox\}$;
    \item for all $r\in\Sigma_{\Csf}$, $r^\inter=\{(a,b)\mid
      r(a,b)\in\abox\}$.
    \end{itemize}
    We show for each $i\geq 0$ that $\inter_i$ satisfies the
    properties above.

    Suppose $i=0$. First let $A(a)\in\abox$ with $A\in\Sigma_{\Csf}$. Then
    $a\in A^{\inter_0}$ by definition of $\Imc_0$. For the other
    direction, let $a\in A^{\inter_0}$ for an $A\in\Sigma_{\Csf}$. Then $A\in
    f(a)$. The definition of \Tmc-decorations yields $A\in t^a_\abox$,
    and thus $A(a)\in\abox$. Now let $r(a,b)\in\abox$ with
    $r\in\Sigma_{\Csf}$. Then $(a,b)\in r^{\inter_0}$ by definition of
    $\Imc_0$. For the other direction, let $(a,b)\in r^{\inter_0}$ for
    some $r\in\Sigma_{\Csf}$. Then there is some role $s$ such that
    $s(a,b)\in\abox$ and $\tbox\models s\sqsubseteq r$. By the adopted
    restriction on the allowed RIs, it follows that
    $\sig{s}{}\subseteq\Sigma_{\Csf}$. This yields $r(a,b)\in\abox$ since
    \Amc satisfies Point~(2) of Lemma~\ref{lem:deco}.

    For $i > 0$, we show that the extension of $\Sigma_{\Csf}$-predicates is
    not modified when constructing $\inter_{i+1}$ from
    $\inter_i$. Indeed, assume that $\Imc_{i+1}$ was obtained from
    $\Imc_i$ by choosing $d \in \Delta^{\Imc_i}$ and $\exists s \in
    \mn{con}(\Tmc)$ and let $q=f(a),s$ if $d=a \in \mn{Ind}(\Amc)$ and
    $q=d,s$ otherwise. Then $\mn{sig}(s) \cap \Sigma_{\Csf} = \emptyset$ and
    by the restriction on RIs, $\mn{sig}(r) \cap \Sigma_{\Csf} =
    \emptyset$ for any role $r$ with $\Tmc \models s \sqsubseteq
    r$. Consequently, none of the role names modified in the
    construction of $\Imc_{i+1}$ is from $\Sigma_{\Csf}$ (no matter whether
    $q$ is $\Sigma_{\Csf}$-participating or not).  In the case where $q$ is
    $\Sigma_{\Csf}$-participating, there is nothing else to show. If $q$ is
    not $\Sigma_{\Csf}$-participating, then each concept name $A$ with $\Tmc
    \models \exists s^- \sqsubseteq A$ is not from $\Sigma_{\Csf}$. Thus also
    none of the concept names modified in the construction of
    $\Imc_{i+1}$ is from $\Sigma_{\Csf}$.
  \end{proof}

\medskip
\noindent
{\bf Lemma~\ref{lem:deco2}}
{\em
	Let $\Amc$ be a $\Sigma_{\Csf}$-ABox consistent 
	w.r.t.~$(\Tmc,\Sigma_{\Csf})$. Then $\Amc \not\models 
	Q(\pi(x_{1}),\ldots,\pi(x_{n}))$ iff $\Amc$ realizes some 
	$(\Tmc,q)$-typing $T$ using $\pi$ that avoids $q$ and such that 
	$\mn{tp}(T)$ is $\Sigma_{\Csf}$-realizable.
}
\begin{proof} The proof is a modification of the proof of
    Lemma~\ref{lem:deco}. We only sketch the differences.

    \smallskip

	($\Rightarrow$) Let $\Amc \not\models 
	Q(\pi(x_{1}),\ldots,\pi(x_{n}))$.  We start with a model \Imc of 
	$\Tmc$ and $\Amc$ that respects closed predicates $\Sigma_{\Csf}$ such 
	that $\Imc \not\models q(\pi(x_{1}),\ldots,\pi(x_{n}))$. Read off a 
	$(\Tmc,q)$-typing
$$
T_\Imc=(\sim,f_0,\Gamma,\Delta)
$$
from \Imc by setting
\begin{itemize}
\item $x_{i}\sim x_{j}$ iff $\pi(x_{i})=\pi(x_{j})$;
\item $f_0(x_{i})=t^{\pi(x_{i})}_\Imc$ for all $1\leq i \leq n$;
\item $\Gamma = \{ t^a_\Imc \mid a\in
  \mn{Ind}(\Amc)\}\setminus\{\pi(x_{1}),\ldots,\pi(x_{n})\}$;
\item $\Delta = \{ r(x_{i},x_{j})\mid r\in \Sigma_{\Csf},
  r(\pi(x_{i}),\pi(x_{j}))\not\in \Amc\}$.
\end{itemize}
We show that $T_{\Imc}$ avoids $q=\bigvee_{i\in I}q_{i}$. Since $\Imc
\not\models q(\pi(x_{1}),\ldots,\pi(x_{n}))$ we find for every $i\in
I$ an atom $\alpha_{i}$ in $q_{i}$ such that $\Imc\not\models
\alpha_{i}(\pi(x_{1}),\ldots,\pi(x_{n}))$.  We show that $T_{\Imc}$
avoids $X=\{ \alpha_{i}\mid i\in I\}$. We distinguish the following
cases:
\begin{itemize}
\item Let $A(x)\in X$. Then $A\not\in t^{\pi(x)}_\Imc$ and so
  $A\not\in f_{0}(x)$, as required.
\item Let $\exists s\in f_{0}(x)$. Then $\exists s \in
  t^{\pi(x)}_\Imc$. Thus, there exists $d\in \Delta^{\Imc}$ such that
  $(\pi(x),d)\in s^{\Imc}$.  If $d\in \Delta^{\Imc}\setminus
  \mn{Ind}(\Amc)$, then $\mn{sig}(B) \cap \Sigma_{\Csf}=\emptyset$ for all
  $B\in t^{d}_{\Imc}$. Thus (i) holds.  If $d\in
  \mn{Ind}(\Amc)\setminus \{\pi(x_{1}),\ldots,\pi(x_{n})\}$, then (ii)
  holds.  Now assume that $d = \pi(y)$ for some $y\in
  \{\pi(x_{1}),\ldots,\pi(x_{n})\}$.  Then $y$ satisfies the
  conditions for (iii).
\item Let $r(x,y)\in X$. Then $(\pi(x),\pi(y))\not\in r^{\Imc}$. Hence
  $(\pi(x),\pi(y))\not\in s^{\Imc}$ for any $s\in \Sigma_{\Csf}$ with
  $\Tmc\models s \sqsubseteq r$. Thus $s(x,y)\in \Delta$ for any such
  $s$. Moreover, $(\pi(y),\pi(x))\not\in s^{\Imc}$ for any $s\in
  \Sigma_{\Csf}$ with $\Tmc\models s^{-} \sqsubseteq r$. Thus $s(y,x)\in
  \Delta$ for any such $s$.
\end{itemize}

\smallskip ($\Leftarrow$) Assume that a $\Sigma_{\Csf}$-Abox $\Amc$ that is
consistent w.r.t.~$(\Tmc,\Sigma_{\Csf})$ realizes some
$(\Tmc,q)$-typing $T=(\sim,f_{0},\Gamma,\Delta)$ using $\pi$ that avoids
$q$. Assume $f$ is a $\Tmc,q$-decoration of $\Amc$ that realizes $T$
using $\pi$. Let $X=\{\alpha_{i}\mid i\in I\}$ with $\alpha_{i}$ in $q_{i}$ 
such that $T$ avoids $X$ using $\pi$. We construct a model
$\Imc$ of $\Amc$ and $\Tmc$ that respects closed predicates $\Sigma_{\Csf}$
such that $\Imc\not\models \alpha_{i}[\pi(x_{1}),\ldots,\pi(x_{n})]$
for $i\in I$.  We build $\Imc$ as in the proof of Lemma~\ref{lem:deco}
based on $\mn{tp}(T)$. Some care is required in the construction of
$\Imc_{i+1}$. Assume $\Imc_{i}$ has been constructed. Choose $d \in
\Delta^{\Imc_i}$ and $\exists s \in \mn{con}(\Tmc)$ such that
$\mn{sig}(s) \cap \Sigma_{\Csf} = \emptyset$, $\Tmc \models \bigsqcap
t^d_{\Imc_i} \sqsubseteq \exists s$ and there is no $(d,e) \in
s^{\Imc_i}$. If $d\not\in \{\pi(x_{1}),\ldots,\pi(x_{n})\}$ or
$\{B\in\mn{con}(\tbox)\mid \tbox\models\existsr{s^-}{}\sqsubseteq B\}$
does not contain a $B$ with $\mn{sig}(B)\subseteq \Sigma_{\Csf}$ proceed as
in the proof of Lemma~\ref{lem:deco}.  Now assume that $d=\pi(x)$.  In
the proof of Lemma~\ref{lem:deco} we chose an \emph{arbitrary} $b\in
\mn{Ind}(\Amc)$ with $\{B \in \mn{con}(\Tmc) \mid \exists s^{-}
\sqsubseteq B\} \subseteq t$ and $t=f(b)$ and added $(a,b)$ to
$r^{\Imc_{i+1}}$ whenever $\Tmc\models s \sqsubseteq r$. Since we want
to refute all atoms $\alpha_{i}(\pi(x_{1}),\ldots,\pi(x_{n}))$ with
$i\in I$, we now have to choose $b$ more carefully.  If there exists
$b\in \mn{Ind}(\Amc)\setminus\{\pi(x_{1}),\ldots,\pi(x_{n})\}$ with
$\{B \in \mn{con}(\Tmc) \mid \exists s^{-} \sqsubseteq B\} \subseteq
t$ and $t=f(b)$, then we choose such a $b$ and proceed as in
Lemma~\ref{lem:deco}. Otherwise, since $f$ is a $\Tmc,q$-decoration of
$\Amc$ that realizes $T$ using $\pi$ and avoids $X$, there is $y$ such
that $\{B\in \mn{con}(\tbox)\mid \tbox\models\existsr{s^-}{}\sqsubseteq
B\}\subseteq f_{0}(y)$ such that there is no $\alpha_{i} \in X$ of the
form $t(x',y')$ or $t(y',x')$ with $x'\sim x$ and $y' \sim y$ such
that $\Tmc\models s\sqsubseteq t$ or $\Tmc\models s\sqsubseteq t^{-}$,
respectively.  We set $b=\pi(y)$ and proceed as in the proof of
Lemma~\ref{lem:deco}.

The resulting interpretation $\Imc$ is a model of $\Tmc$ and $\Amc$
that respects closed predicates $\Sigma_{\Csf}$. Moreover $\Imc\not\models
\alpha_{i}(\pi(x_{1}),\ldots,\pi(x_{n}))$ for all $i\in I$. Thus,
$\Imc\not\models q(\pi(x_{1}),\ldots,\pi(x_{n}))$, as required.
\end{proof}

\end{document}